\documentclass[aps,amsmath,amssymb,twocolumn,nofootinbib,pra,superscriptaddress,longbibliography]{revtex4-2} 
\usepackage[colorlinks,citecolor=blue,urlcolor=blue,bookmarks=false,hypertexnames=true]{hyperref} 

\usepackage{graphicx}
\usepackage{dcolumn}
\usepackage{bm}
\usepackage{amsmath,color,subfigure,floatrow,multirow,booktabs,amsthm} 
 \usepackage{algorithm}     
 \usepackage{algorithmicx}  
 \usepackage{boondox-cal} 
 \usepackage{algpseudocode}  

\newtheorem{theorem}{\indent Theorem}
\newtheorem{proposition}{\indent Proposition}
\newtheorem{lemma}{\indent Lemma}
\newtheorem{remark}{\indent Remark}
\newtheorem{problem}{\indent Problem}
\newtheorem{corollary}{\indent Corollary}
\begin{document}

\preprint{APS/123-QED}

\title{Unified formalism and adaptive algorithms for optimal quantum state, detector and process tomography}

\author{Shuixin Xiao}
\email{These authors contributed equally to this work.}
 \affiliation{School of Engineering, Australian National University, Canberra, ACT 2601, Australia}
 \author{Xiangyu Wang}
 \email{These authors contributed equally to this work.}
 \affiliation{Laboratory of Quantum Information, University of Science and Technology of China, Hefei 230026, China}
 \affiliation{CAS Center For Excellence in Quantum Information and Quantum Physics, University of Science and Technology of China, Hefei 230026, China}
 \author{Yuanlong Wang}\email{wangyuanlong@amss.ac.cn}
 \affiliation{State Key Laboratory of Mathematical Sciences, Academy of Mathematics and Systems Science, Chinese Academy of Sciences, Beijing 100190, China}
 \affiliation{School of Mathematical Sciences, University of Chinese Academy of Sciences, Beijing 100049, China}
 \author{Zhibo Hou}\email{houzhibo@ustc.edu.cn}
 \affiliation{Laboratory of Quantum Information, University of Science and Technology of China, Hefei 230026, China}
 \affiliation{CAS Center For Excellence in Quantum Information and Quantum Physics, University of Science and Technology of China, Hefei 230026, China}
  \author{Jun Zhang}
   \affiliation{University of Michigan--Shanghai Jiao Tong University Joint Institute, Shanghai Jiao Tong
		University, Shanghai, 200240, China}
 \author{Ian R. Petersen}
 \affiliation{School of Engineering, Australian National University, Canberra, ACT 2601, Australia}
\author{Wen-Zhe Yan}
 \affiliation{Laboratory of Quantum Information, University of Science and
 		Technology of China, Hefei 230026, China}
 \affiliation{CAS Center For Excellence in Quantum Information and Quantum Physics, University of Science and Technology of China, Hefei 230026, China}
  \author{Hidehiro Yonezawa}
   \affiliation{Center for Quantum Computing, RIKEN, Wakoshi, Saitama 351-0198, Japan}
   \author{Franco Nori}
      \affiliation{Center for Quantum Computing, RIKEN, Wakoshi, Saitama 351-0198, Japan}
         \affiliation{Physics Department, The University of Michigan, Ann Arbor, Michigan, 48109-1040, USA}
 \author{Guo-Yong Xiang}
 \affiliation{Laboratory of Quantum Information, University of Science and
 		Technology of China, Hefei 230026, China}
 \affiliation{CAS Center For Excellence in Quantum Information and Quantum Physics, University of Science and Technology of China, Hefei 230026, China}
 \author{Daoyi Dong}\email{daoyidong@gmail.com}
 \affiliation{Australian Artificial Intelligence Institute, Faculty of Engineering and Information Technology, University of Technology Sydney, NSW 2007, Australia}

\date{\today}
\begin{abstract}

Quantum tomography is a standard technique for characterizing, benchmarking and verifying quantum systems/devices and plays a vital role in advancing quantum technology and understanding the foundations of quantum mechanics. Achieving the highest possible tomography accuracy remains a central challenge. Here we unify the infidelity indices for quantum state, detector and process tomography in a single index $1-F(\hat S,S)$, where $S$ represents the true density matrix, positive operator-valued measure (POVM) element, or process matrix, and $\hat S$ is its estimator. We establish a sufficient and necessary condition for any tomography protocol to attain the optimal scaling
$1-F= O(1/N) $  where $N$ is the number of state copies consumed, in contrast to the $O(1/\sqrt{N})$ worst-case scaling of static methods. Guided by this result, we propose adaptive algorithms with provably optimal infidelity scalings for state, detector, and process
tomography. Numerical simulations and quantum optical experiments validate the proposed methods, with our experiments reaching, for the first time, the optimal infidelity scaling in ancilla-assisted process tomography. This work has potential applications in quantum device calibration, quantum precision measurement, quantum metrology and quantum engineering.

\end{abstract}

\maketitle

\section{Introduction}

In recent decades, quantum technologies have made remarkable advances in various domains, including quantum computation \cite{DiVincenzo255}, quantum communications \cite{ren2017}, quantum sensing \cite{RevModPhys.89.035002} and quantum simulation \cite{RevModPhys.86.153}. Within these technologies, a crucial objective lies in the development of precise estimation and identification algorithms to capture comprehensive information about the quantum systems in question. This endeavor is commonly referred to as quantum tomography \cite{Jullien2014,qci}.
Quantum tomography encompasses three primary tasks as presented in Fig. \ref{qt}: quantum state tomography (QST), quantum detector tomography (QDT), and quantum process tomography (QPT). These tasks involve the reconstruction of an unknown quantum state, detector, or process~\cite{qci,You2011,Gebhart2023,PhysRevA.88.022101,Lundeen2009,Bialczak2010,burgarth_2011}. Numerous efficient or highly accurate algorithms have been proposed to address various challenges, as highlighted in \cite{pure2,PhysRevA.90.062123,high2021,Hou2016,Zhu2022,PhysRevB.92.075312,PhysRevLett.108.080502,Gross2010,PhysRevLett.113.190404,chen2019nc,PhysRevLett.127.140502,PhysRevLett.131.113601}.

For QST, to evaluate the accuracy of the estimator $\hat \rho$ relative to the true state $\rho$, several indices---such as mean squared error (MSE), trace distance, and infidelity ($1-$fidelity)---have been proposed. Infidelity quantifies (inversely proportional to) the resource number (i.e., the total number of state copies) required to reliably distinguish two quantum states \cite{che1,7956181,PhysRevLett.111.183601} and is widely adopted.  The optimal scaling of QST infidelity between $\hat\rho$ and $\rho$ has been proved to be $ O(1/N) $  using $N$ identical and independent probes, achieved for well-conditioned full-rank states by certain static (non-adaptive) QST algorithms. In contrast, all static algorithms only give $O(1/\sqrt{N})$ infidelity in the worst case~\cite{PhysRevLett.111.183601, zhu}---namely for general rank-deficient or nearly rank-deficient
targets where statistical noise dominates over small eigenvalues. Similar scaling hierarchies in terms of standard deviation exist in quantum metrology~\cite{science1,Giovannetti2011,PhysRevLett.96.010401}, though quantum tomography requires distinct theoretical treatment due to its multiparameter nature.

Adaptivity is a powerful tool to enhance the parameter estimation accuracy, e.g., in optical phase estimation \cite{w1,w2,hy2012}.
Thus to achieve optimal scaling, various adaptive algorithms have been proposed, in QST by updating measurement operators according to a state estimator \cite{PhysRevLett.111.183601,pereira2020high,b1,b2,Qi2017,PhysRevA.98.012339,PhysRevLett.122.100404}, and in QDT by selecting probe states based on a preliminary  detector characterization \cite{xiao2021optimal}. For QPT, there are 
generally three classes based on different system architectures: Standard Quantum Process Tomography \cite{qci}, Ancilla-Assisted Process Tomography (AAPT) \cite{meaqo,aaqpt,AAPT2} and Direct Characterization of Quantum Dynamics \cite{dcqd}. Most recently, adaptive approaches focused on adaptive standard QPT \cite{Wang2016,inada}, with adaptive AAPT and adaptive Direct Characterization of Quantum Dynamics rarely investigated.

Despite these achievements, three key challenges persist in this field, inhibiting the design of optimal tomography algorithms: 

(i) Few existing tomography algorithms have been rigorously proven to achieve optimal infidelity scaling in general scenarios. For instance, while Bayesian adaptive QST~\cite{b1,b2} has been numerically and experimentally demonstrated to reach optimal scaling, a rigorous theoretical proof remains elusive. Within the framework of maximum likelihood estimation (MLE), Ref.~\cite{PhysRevLett.111.183601} presented a sufficient condition and a specific algorithm to achieve $O(1/N)$ infidelity in single-qubit QST. This method was then generalized to the case of a single qudit and further investigated in~\cite{PhysRevA.98.012339}, and subsequently validated by experiment in ~\cite{pereira2020high}, while regretfully overlooking the degenerate scenarios where the target to be estimated has multiple equal eigenvalues. 
Ref.~\cite{PhysRevA.98.032330}
later proposed separate necessary and sufficient conditions for adaptive QST. While that framework employs factorized measurements applicable to higher-dimensional systems, their numerical examples demonstrate optimal infidelity scaling exclusively for pure states, a limitation similarly observed in~\cite{Qi2017}. Furthermore, a universal statistical distribution governing the fidelity of the reconstructed quantum state was derived in~\cite{Bo2009}, which elegantly highlighted the critical importance of accurately estimating zero eigenvalues. Regarding other tomographic tasks, Ref.~\cite{xiao2021optimal} proposed an adaptive QDT algorithm with provably optimal infidelity scaling, however, degenerate scenarios were again not addressed. For QPT, Ref.~\cite{PhysRevA.95.012302} introduced an adaptive scheme by optimizing input states and measurements in standard QPT that achieves optimal infidelity scaling, but only for restricted classes of quantum channels. Overall, strict proofs of achieving the optimal infidelity scaling for general cases, particularly accounting for degenerate cases and without special restrictions on the target to be estimated, remain scarce. 

(ii) Many efforts have been devoted to proposing optimal tomography algorithms, which serve as sufficient conditions to achieve optimal infidelity scaling, and little is known regarding non-trivial necessary conditions an optimal scheme should satisfy, particularly in QDT and QPT. As a result, the critical necessity of the MSE converging as $O(1/N)$ has been long overlooked, and closing the theoretical gap between sufficiency and necessity remains a crucial open challenge.  Ref.~\cite{PhysRevA.98.032330} proposed a necessary condition to achieve optimal infidelity scaling in QST, but this condition alone is not sufficient.

(iii) The existing research on tomography algorithms has largely focused on individual types of QST, QDT or QPT. The mathematical similarity between density matrices, POVM elements, and process matrices has not been fully utilized, and a unified analysis on optimality across the three tasks is still missing.

To address these problems, in this work, first we extend the infidelity definition from states to arbitrary positive semidefinite operators, which thus also includes POVM elements and process matrices and reduces to the normal infidelity definitions for states and trace-preserving process matrices. We then propose a sufficient and necessary condition on when the estimator of any QST, QDT or QPT algorithm can \emph{achieve the optimal infidelity scaling} $ O(1/N) $. This result applies to any finite-dimension quantum systems, including degenerate cases. To the best of our knowledge, this is the first equivalent characterization on optimal infidelity scaling of tomography outcomes, providing a unified formalism for the optimal three tomography tasks and representing a significant progress in filling the gaps and loopholes aforementioned in the field of quantum tomography. Our result significantly facilitates the study of performance limits in existing methods and provides a principled approach for designing effective or even optimal algorithms. In particular, we introduce novel adaptive algorithms for QST, QDT, and AAPT to \emph{achieve optimal infidelity scaling}. For AAPT, we experimentally implement this algorithm on a quantum optical system, which, based on available literature, is the first work to prove and realize optimal infidelity scaling in AAPT experiments. 

\section{Results}
\subsection{Preliminaries and a generalized infidelity definition for quantum tomography}
\begin{figure}[!htb]     
    \centering 
    \includegraphics[width=3.3in]{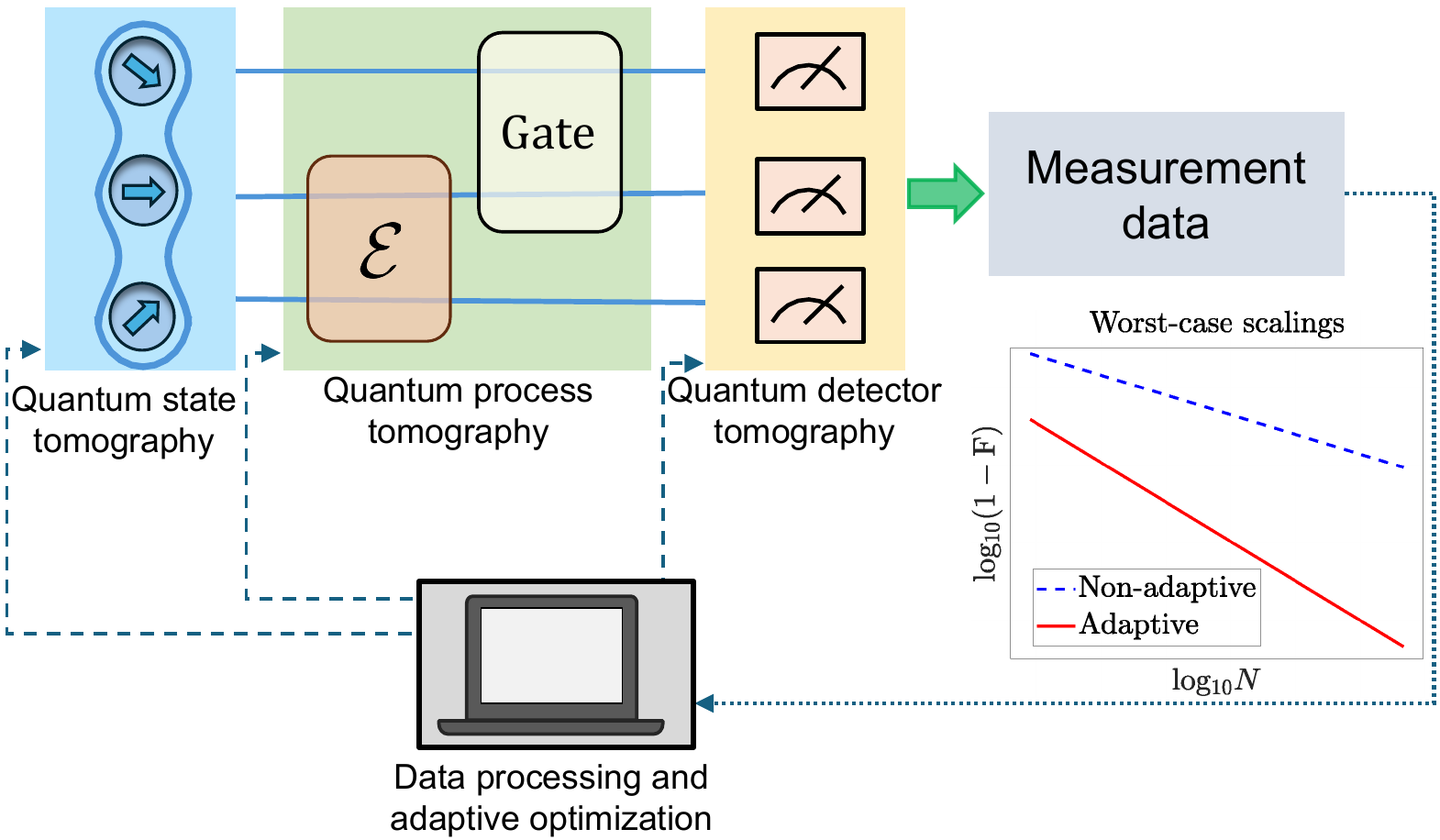}
    \centering{\caption{
    Three quantum tomography tasks: QST, QDT, and QPT \cite{Gebhart2023}. Examples of quantum processes are the quantum gate and general quantum process $\mathcal{E}$  which may be trace-preserving and non-trace-preserving in our work. Leveraging the measurement data, we develop adaptive algorithms to achieve optimal infidelity scaling $O(1/N)$ in QST, QDT, and QPT. 
    }\label{qt}}     
\end{figure}

As presented in Fig. \ref{qt},
the tasks of quantum state tomography (QST), quantum detector tomography (QDT), and quantum process tomography (QPT) are to fully infer an unknown quantum state, POVM and process $\mathcal{E}$, respectively, based on known quantities and experimental measurement outcomes. A more detailed illustration is presented in Supplementary Section I. For QPT, multiple representations exist for the same $\mathcal{E}$, and here we employ the widely-used process matrix $0\leq X \in \mathbb{C}^{d^2\times d^2}$, satisfying $ \operatorname{Tr}_{1}\left(X\right)\leq I_d $ and in one-to-one correspondence with $ \mathcal{E} $ \cite{qci}, where $\text{Tr}_1(X)$ takes the partial trace of $ X \in \mathbb H_1\otimes \mathbb H_2$ on $\mathbb H_1$. When $ \operatorname{Tr}_{1}\left(X\right)= I_d $, the process is called trace-preserving, and otherwise non-trace-preserving, both of which are considered in this paper.

To benchmark estimation accuracy, various metrics such as MSE, trace distance, and infidelity have been proposed.
As opposed to alternative metrics, infidelity measures (inversely proportional to) the resources needed to effectively differentiate between two quantum states \cite{che1,7956181,PhysRevLett.111.183601}, making it a widely accepted metric. 
The fidelity between two arbitrary states $ \sigma , \rho \in \mathbb{C}^{d\times d}$ is defined as
\begin{equation}\label{deinf}
	F_{s}(\sigma,\rho)\triangleq\Big(\operatorname{Tr} \sqrt{\sqrt{\sigma} \rho \sqrt{\sigma}}\Big)^{2},
\end{equation}
with $0\leq F_{s}(\sigma,\rho) \leq 1$. A natural extension of $F_{s}$ to QDT  \cite{Lundeen2009} and QPT \cite{PhysRevA.82.042307} is given by
\begin{equation}
	F_{d,p}(\hat{S},S)\triangleq \frac{\left(\operatorname{Tr} \sqrt{\sqrt{\hat{S}} S \sqrt{\hat{S}}}\right)^{2}}{\operatorname{Tr}(\hat{S})\operatorname{Tr}(S)},
\end{equation}
where $S$ can be either the process matrix $X$ or POVM element $P_i$ satisfying $\sum_{i=1}^n P_i=I$, and $\hat{S}$ is the corresponding estimate.

To equivalently characterize infidelity scalings in QDT and QPT without ambiguity, $F_{d,p}$ needs to be further modified. We consider
\begin{equation}\label{s2}
		F_1(\hat{S}, S)\triangleq \frac{\left(\operatorname{Tr} \sqrt{\sqrt{\hat{S}} S \sqrt{\hat{S}}}\right)^{2}}{\operatorname{Tr}(\hat{S})\operatorname{Tr}(S)}-\frac{\left(\operatorname{Tr}(S-\hat{S})\right)^{2}}{d^2},
\end{equation}
whose tight lower bounds (denoted by $f$) in QDT and QPT are $d^{-1}-1$ and $-1$, respectively. Normalizing the range of $F_1$ from $[f,1]$ to $[0,1]$, we redefine the fidelity as 
\begin{equation}
    F(\hat{S},S)\triangleq \frac{F_1(\hat{S}, S)-f}{1-f}.
\end{equation}
The benefits of this definition are threefold: 
\begin{itemize}
    \item[(i)] $F(\hat{S}, S)$ is universally applicable to density matrices, POVM elements and process matrices, and incorporates $F_s(\hat{\rho},\rho)$ as a special case; 
    \item[(ii)] Zero infidelity, $1-F(\hat{S},S)=0$, necessarily indicates $\hat{S}=S$, and this property is absent in $F_{d,p}$ and guarantees that an equivalent rather than merely sufficient characterization of the optimal scaling can be obtained later; 
    \item[(iii)]  Crucially, $(1-F)$ inherits the same scaling behavior as $(1-F_s)$ (with a constant difference) up to the optimal order. Hence, scaling analysis using $(1-F)$ can also be performed based on $(1-F_s)$ and $(1-F_{d,p})$. In the rest of this paper, infidelity by default refers to $(1-F)$ unless otherwise specified.
\end{itemize}

For QST, QDT, and QPT tasks with $N$ state copies consumed, the optimal scalings of $1-F_s$, $1-F_{d,p}$, and $1-F$ are all $O(1/N)$, which can be achieved for full-rank $S$. However, for rank-deficient $S$, without introducing extra resources such as adaptivity or prior knowledge, the infidelity scaling is often only  $O(1/\sqrt{N})$ (see Supplementary Section II-A,B for details of all the above analysis in this section). Thus, the main focus of this paper is to investigate how to achieve  $1-F(\hat{S}, S)= O\left(1/N\right) $ for an arbitrary operator $S\geq 0$. 

\subsection{A unified characterization of optimal infidelity scaling}
We present the following sufficient and necessary condition (with its proof in Supplementary Section II-B) to characterize how to achieve optimal infidelity scaling across different tomography tasks, where $\mathbb{E}(\cdot)$ denotes the expectation over all possible measurement outcomes and $\|\cdot\|$ denotes the Frobenius norm.

\begin{theorem}	\label{theorem1}
For any unknown positive semidefinite operator $S \in \mathbb C^{d\times d}$ encoded in quantum tomography, denote its spectral decomposition as $$ S=\sum_{j=1}^d \lambda_j\left\lvert\lambda_j\right\rangle\left\langle\lambda_j\right\lvert$$
	where $ \lambda_1 \geq  \cdots \geq \lambda_r >0, \lambda_{r+1}=\cdots=\lambda_d=0 $. From the measurement results of $ N $ state copies, an estimate $ \hat S \geq 0$ is inferred,  with eigenvalues $ \hat{\lambda}_i $ also in non-increasing order.
	The infidelity $ \mathbb{E}(1-F(\hat S, {S}))$ scales as $O\left(1/{N}\right)$  if and only if the following conditions are both satisfied:
	\begin{enumerate}
		\item[C1:] The MSE $ \mathbb{E}(\|\hat S-{S}\|^2)$  scales as  $ O(1/N) $;
		\item[C2:] The partial sum of the eigenvalues of  $\hat S$  scales as $\mathbb{E}(\sum_{j=r+1}^{d}\hat{\lambda}_j)= O(1/N)  $.
	\end{enumerate}
\end{theorem}

When $S$ is a density matrix/POVM element/process matrix, Theorem \ref{theorem1} corresponds to the typical tomography scenarios QST/QDT/QPT. The $ O(1/N) $ scaling in C1 is in fact optimal for MSE, established by quantum Cram\'{e}r-Rao bound \cite{Liu_2020}. C1 itself is often easier to achieve than the optimal infidelity scaling. First, the MSE metric used in C1 is analytically more tractable since it relates directly to the Frobenius norm, which facilitates the application of numerous norm inequalities; in contrast, infidelity is not a norm and involves complicated matrix square roots. Second, C1 can already be achieved with strict proofs by static algorithms through separable measurements (individual measurements each on a single copy of $S$) in QST~\cite{qstmle,Qi2013}, QDT~\cite{wang2019twostage}, and QPT~\cite{xiaoaapt,xiaoqpt}. Conversely, to the best of our knowledge, no existing algorithms can achieve optimal infidelity scaling for general rank-deficient $S$ without employing adaptivity or collective measurements (joint measurements each on multiple copies of $S$).

However, C1 alone can only guarantee $O(1/\sqrt{N})$ accuracy for the estimated eigenvalues. If the true $S$ has zero eigenvalues, the estimated zero eigenvalues usually scale as $O(1/\sqrt{N})$ rather than $O(1/N)$. Since infidelity exhibits high sensitivity to errors in estimating small eigenvalues \cite{PhysRevLett.111.183601}, it now only scales as $O(1/\sqrt{N})$. To improve the infidelity scaling to $O\left(1/N\right)$, Theorem~\ref{theorem1} requires to (and it suffices to) accurately (in $O(1/N)$ scaling) estimate both the entirety of the target $S$ and its zero eigenvalues. Motivated by the qubit scenario in~\cite{PhysRevLett.111.183601}, one empirical approach to satisfy C2 is to measure the unknown $S$ in its (estimated) eigenbasis.

The necessity of C1 has not been noticed previously. Instead, our result elucidates the necessity of C1 and C2, which is helpful to provide one with confidence and confirmation when designing tomography algorithms with optimal infidelity scalings or even optimal infidelity limit in the future. If a protocol fails to satisfy C1 or C2 through theoretical analysis at first sight, then it cannot achieve the optimal infidelity scaling.

Existing related research is mostly sporadic in individual tomography tasks, and the theoretical analysis has been focusing on giving algorithm examples only sufficient for achieving optimal scaling, with the necessary direction seldom discussed formally. Most intimately associated is \cite{PhysRevLett.111.183601}, where the sufficiency of accurately estimating the small eigenvalues (C2) was proposed for single-qubit QST, without noticing the significance of accurately estimating the entirety of $\rho$ (C1). Current QST and QDT algorithms achieving the optimal infidelity scaling in \cite{PhysRevLett.111.183601,PhysRevA.98.012339,xiao2021optimal,pereira2020high} now can thus all be covered and explained by Theorem \ref{theorem1}. Since $S$ can be any non-degenerate or degenerate (several eigenvalues can be equal, a case important yet seldom covered before) finite-dimension quantum state, POVM element, or process matrix, this theorem serves as a unified framework for characterizing the optimal scaling of infidelity universally in quantum tomography. 
Therefore, Theorem~\ref{theorem1} provides a substantial advance beyond prevailing tomography methodologies. 

Other related indices include the single infidelity for POVMs~\cite{Hou2018} and the Bures distance for quantum processes~\cite{PhysRevA.95.012302}. Extending Theorem~\ref{theorem1}, Theorems 2 and 3 in Supplementary Section II-C present similar characterizations for their scalings.

Theorem~\ref{theorem1} outlines design principles for achieving optimal infidelity scaling $O(1/N)$. One approach is to first perform a static tomography on a fraction of the state copies, producing an estimate that satisfies C1, and then adaptively adjust the experimental settings in a second step to perform another tomography operation ensuring C2. Guided by this, Methods Section~\ref{methods} introduces three exemplary algorithms: two-step adaptive QST, two-step adaptive QDT, and three-step adaptive AAPT.
While many adaptive QST algorithms exist, few rigorously prove optimal scaling. For instance, Bayesian methods \cite{b1,b2,b3} empirically reach the optimal scaling but lack formal proofs. Previous proofs exist for one-qubit systems \cite{PhysRevLett.111.183601} and were later extended to qudits \cite{PhysRevA.98.012339}. However, Ref.~\cite{PhysRevA.98.012339} overlooks eigenvalue degeneracy and requires at least $2d-1$ detectors in Step-2 \cite{pereira2020high}. In contrast, our two-step adaptive QST explicitly resolves the degeneracy issue (Supplementary Section III-A, Theorem 4) and achieves optimal scaling for multi-qubit or qudit systems using only a single detector in Step-2, which is independent of the dimension $d$.
For QDT without prior knowledge, Ref.~\cite{xiao2021optimal} proposed a two-step algorithm achieving $O(1/N)$ scaling but requires $d^2$ probe states per POVM element. Our two-step adaptive QDT significantly reduces this requirement to just $d$ probe states in Step-2 while maintaining the optimal $O(1/N)$ scaling.
Finally, building on the AAPT framework \cite{meaqo,aaqpt}, we present a three-step adaptive algorithm with a pure input state. To the best of our knowledge, this is the first work rigorously proven to realize optimal infidelity scaling in AAPT.

\subsection{Numerical results of adaptive quantum state tomography}

\begin{figure}
	\centering
	\subfigure{
		\begin{minipage}[b]{1\textwidth}
			\centering	\includegraphics[width=3.5in]{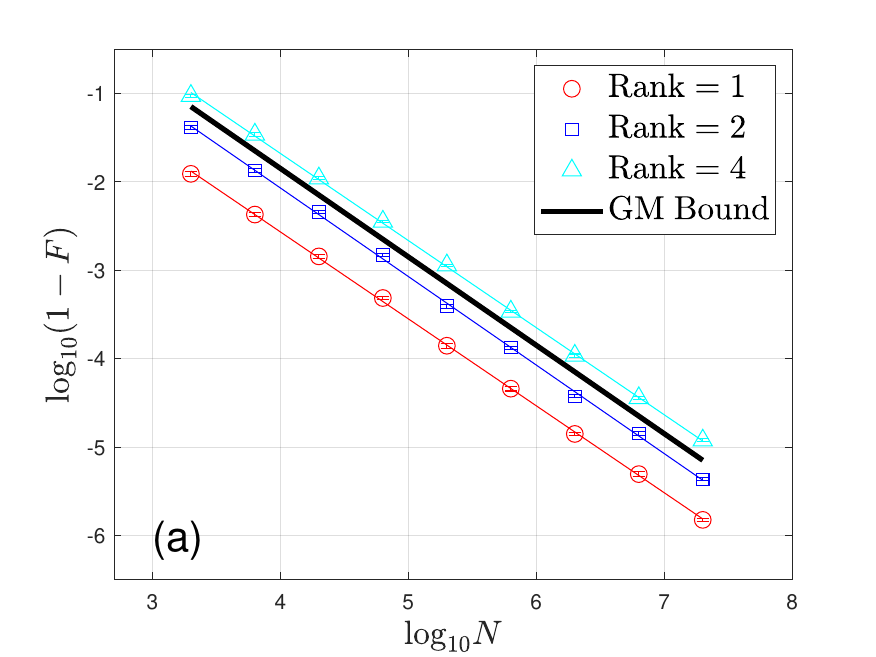}
		\end{minipage}
	}
	\subfigure{
		\begin{minipage}[b]{1\textwidth}
			\centering	\includegraphics[width=3.5in]{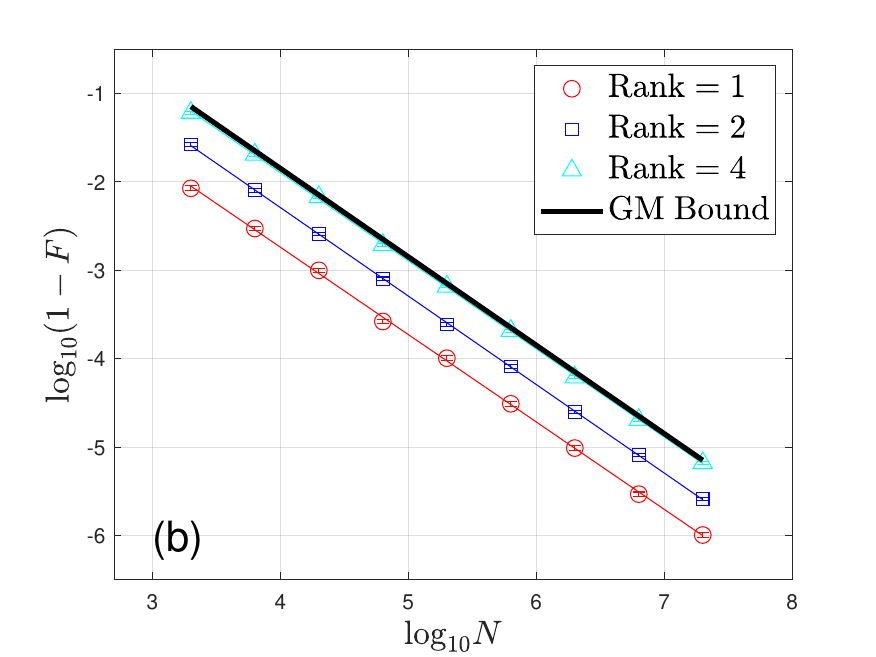}
		\end{minipage}
	}
	\caption{Log-log plot of the infidelity $ \mathbb{E}\left(1-F(\hat\rho, {\rho})\right) $ versus  the total resource number $N$ for the rank-1, rank-2, rank-4 quantum states in Eq.~\eqref{state} using the two-step adaptive QST algorithm in Section \ref{sec31}.  For each resource number, we repeat our algorithm $ 100 $ times and obtain the average infidelity and error bars. (a) $ N_0=0.5N $, the infidelities for the rank-1 and rank-2 states surpass the Gill-Massar (GM) bound, but the infidelity for the rank-4 state is slightly larger than the GM bound. (b) $ N_0=0.9N $, all the infidelities  surpass the GM bound and are slightly smaller than that of $ N_0=0.5N $.} \label{adqstinf}
\end{figure}

In this section, we demonstrate the effectiveness of the proposed adaptive QST scheme via numerical examples. Results for the proposed QDT and AAPT algorithms are deferred to Supplementary Section IV.


In QST, the Gill-Massar (GM) lower bound is a version of the quantum Cram\'{e}r-Rao bound which is applicable to individual measurements on each  copy of the state \cite{gmbound,zhu}. In the case of a $d$-dimensional quantum state, the GM bound for the mean infidelity is $\frac{1}{4}(d+1)^2(d-1) \frac{1}{N}$ and holds for any unbiased estimation~\cite{gmbound,Hou2018}. However, the GM bound may be violated when the state  has zero or near zero eigenvalues (with the specific threshold depending on $N$), in which case common estimators can be biased due to the boundary of the state space \cite{Hou2018}. This phenomenon has also been observed in the adaptive QST for two-qubit systems~\cite{Qi2017}. 
For QST, we consider the unknown states as
\begin{equation}\label{state}
	\begin{aligned}
		\rho=&U \operatorname{diag}(1,0,0,0,0,0,0,0) U^{\dagger},\\
		\rho=&U \operatorname{diag}(1/2,1/2,0,0,0,0,0,0) U^{\dagger},\\
		\rho=&U \operatorname{diag}(1/4,1/4,1/4,1/4,0,0,0,0) U^{\dagger},
	\end{aligned}
\end{equation}
and $ U $ is a random unitary matrix generated by the algorithm in \cite{qetlab,Zyczkowski_1994,qutip},
which is then fixed in each repeated simulation run such that noise in the measurement is only from the finite number of copies of the input states.

We begin by applying the two-step adaptive QST algorithm described in Methods Section~\ref{sec31} to the quantum states given in Eq.~\eqref{state}, assuming no prior knowledge of the states.
In Step-1, we employ three-qubit Cube measurements using \( N_0 = 0.5N \) state copies. For single-qubit systems, the Cube measurements consist of the projectors \( \frac{I \pm \sigma_x}{2} \), \( \frac{I \pm \sigma_y}{2} \), and \( \frac{I \pm \sigma_z}{2} \). For multi-qubit systems, the Cube measurements are defined as the tensor products of the corresponding one-qubit measurements.
In Step-2, we apply adaptive measurement operators \( \{ \lvert \tilde{\lambda}_i \rangle \langle \tilde{\lambda}_i \rvert \}_{i=1}^d \), as introduced in Section~\ref{sec31}, using the remaining \( N - N_0 = 0.5N \) state copies.
The results are shown in Fig.~\ref{adqstinf}(a), where all the infidelities scale as $ O(1/N) $. 
All the three states considered here contain zero eigenvalues. In such cases, the regularity assumptions underlying the GM bound may fail, and common estimators can be biased by the boundary of the state space~\cite{Hou2018}. This provides a possible explanation on why the observed infidelities exceed the GM bound. The effect is more visible for the rank-1 and rank-2 states, whereas the rank-4 state is only slightly above the GM bound, likely due to a weaker boundary effect together with finite-sample fluctuations.

As a further investigation,
we change the resource distribution to $ N_0=0.9N $ in Step-1 and $ N-N_0=0.1N $ in Step-2. The results are shown in Fig.~\ref{adqstinf}(b) where all the infidelities surpass the GM bound and are slightly smaller than that for $ N_0=0.5N $, indicating that the resource distribution proportion affects the tomography error. Furthermore,
as the rank increases, the mean infidelities in Fig. \ref{adqstinf}  also increase and a similar phenomenon was also observed in \cite{PhysRevA.98.012339}. The reason may be that the distance between the zero eigenvalue and the smallest positive eigenvalue decreases from $1$ to $1/4$.

\subsection{Experimental results of adaptive ancilla-assisted process tomography}

\begin{figure}[!htb]
    \centering
    \subfigure{
       \begin{minipage}[b]{1\textwidth}
          \centering \includegraphics[width=3.5in]{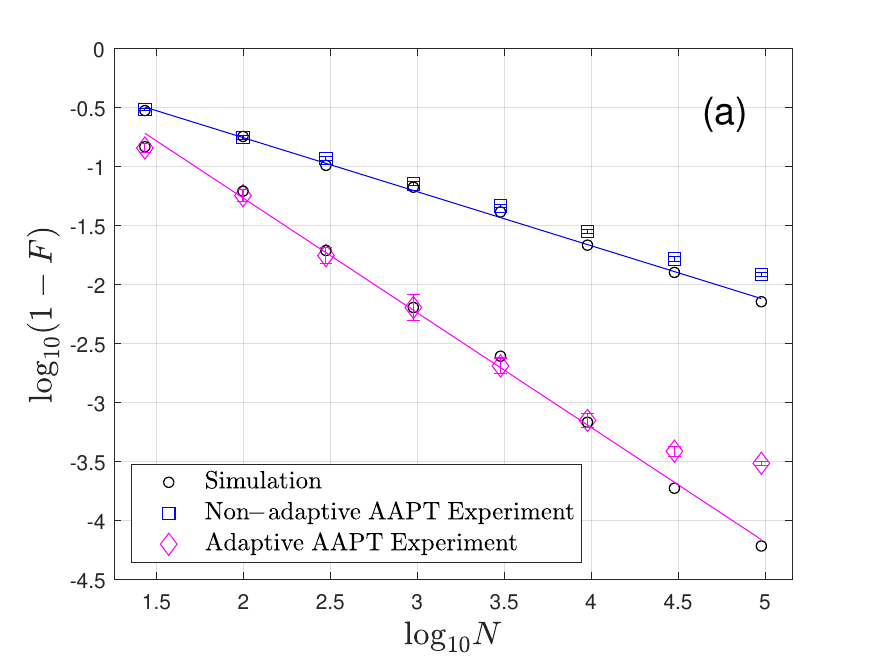}
       \end{minipage}
    }
   \subfigure{
      \begin{minipage}[b]{1\textwidth}
        \centering \includegraphics[width=3.5in]{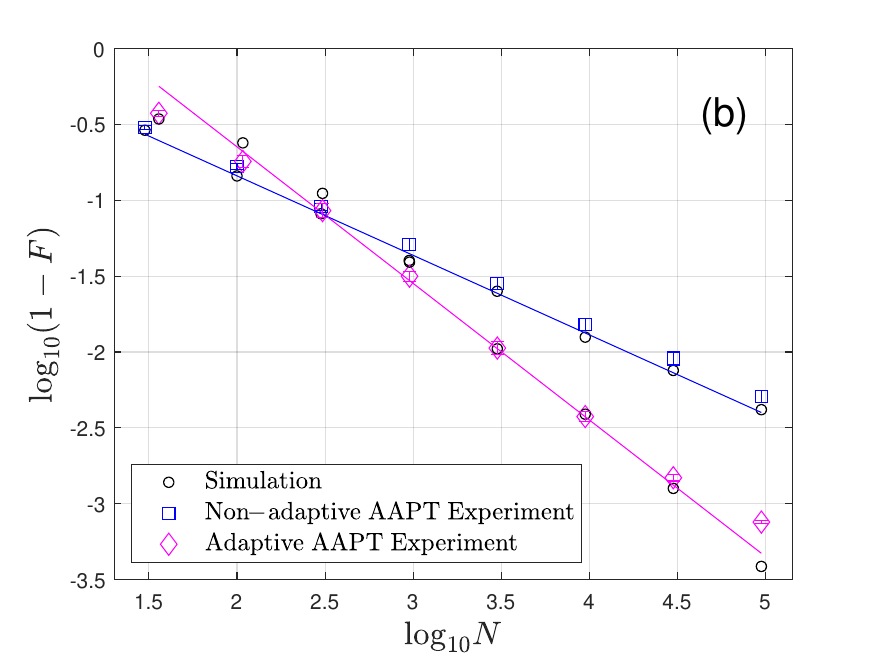}
        \end{minipage}
   }
    \caption{Experimental results for adaptive and non-adaptive AAPT. (a) The unknown quantum process is the Hadamard ($H$) gate as defined in Eq.~\eqref{Hgate}. (b) The unknown quantum process is a phase damping process as specified in Eq.~\eqref{phase2}.  Square markers indicate experimental results obtained via the non-adaptive method, and diamond markers represent results using the adaptive AAPT method detailed in Methods Section \ref{sec31b}. Black circles represent simulation outcomes, with solid lines representing fitted simulations. Error bars show the standard deviations across $100$ repeated experiments. The adaptive method achieves infidelity scaling as $O(1/N)$, in alignment with \textbf{Theorem 1}, and \textbf{Theorem 6} in Supplementary Section III-C, whereas the non-adaptive method scales as $O(1/\sqrt{N})$.}\label{adaaptinf}
\end{figure}

We apply the  adaptive AAPT described in Methods Section \ref{sec31b} with $N_0=0.5N$ to both unitary and non-unitary processes. As a comparison, we also simulate the outcomes for a non-adaptive approach to AAPT.
For unitary processes, we choose the  Hadamard ($H$) gate
\begin{align}\label{Hgate}
 H=   \frac{1}{\sqrt{2}}\begin{bmatrix}
        1 & 1  \\
        1 & -1  \\
    \end{bmatrix}. 
\end{align}
For non-unitary processes, we implement a phase damping process characterized by two Kraus operators
\begin{equation}\label{phase2}
	\mathcal{A}_1=\operatorname{diag}(1,\sqrt{1-\lambda}),\quad \mathcal{A}_2=\operatorname{diag}(0,\sqrt{\lambda}),
\end{equation}
where $\lambda=0.989$.

From Fig. \ref{adaaptinf}, we observe that in both numerical and experimental settings, the infidelity of the reconstructed process matrix $\mathbb{E}(1 - F(\hat{X}, X))$ scales as $O(1/\sqrt{N})$  using the non-adaptive AAPT. This scaling arises because the non-adaptive method cannot achieve $O(1/N)$ scaling in estimating zero eigenvalues of the process matrix. In contrast, the adaptive AAPT method introduced in Section \ref{sec31b} achieves $O(1/N)$ scaling in both numerical and experimental results. The experimental outcomes align closely with the numerical simulations, demonstrating robustness across both unitary and non-unitary processes and validating the effectiveness of our proposed tomography approach. For the unitary process $H$ gate, the infidelity of the adaptive AAPT is consistently lower than that of the non-adaptive method. However, in the case of phase damping process, when $N$ is less than $10^{2.5}$, the infidelity of the non-adaptive approach is comparatively lower. This is because the adaptive method relies on an initial estimate in the first step, and when $N$ is relatively small, this initial estimate is inaccurate, impacting the subsequent outcome in the adaptive step.
In our experiment, the $O(1/N)$ infidelity scaling holds up to $ N \approx 10^{4.5}$, beyond which systematic imperfections including state preparation and measurement (SPAM) errors dominate over statistical noise and limit the practically achievable scaling.

\section{Discussion}
In this work, we have presented a unified framework for investigating the achievable
optimality of infidelity in quantum tomography. We extended the typical state fidelity/infidelity definition to a general metric between two positive semidefinite operators and proposed a sufficient and necessary condition to achieve the optimal infidelity scaling $O(1/N)$ for quantum tomography, which is applicable to the three typical tomography tasks QST, QDT, and QPT with arbitrary finite dimensions meanwhile allowing for degenerate situations. Among existing tomography methods, our analytical characterization appears to be the first result to close the gap between sufficiency and necessity of optimal infidelity scaling. 
It encompasses and extends the related research on adaptive algorithms with improved infidelity scalings introduced in previous work \cite{PhysRevLett.111.183601,PhysRevA.98.012339,xiao2021optimal}.

Guided by Theorem \ref{theorem1}, we have designed adaptive algorithms in QST and QDT while achieving optimal infidelity scaling. 
Compared with the existing adaptive methods in \cite{PhysRevA.98.012339, xiao2021optimal, pereira2020high}, our algorithms need fewer types of measurements or probe states. Numerical examples demonstrate the effectiveness of our algorithms. Furthermore, an adaptive AAPT algorithm has also been presented, which can provably attain the optimal infidelity scaling, representing the first AAPT result with such optimal performance.

Our adaptive QST algorithm is based on individual measurements, performed on each single copy of the state. Collective measurement  \cite{7956181,Hou2018}, also covered by Theorem \ref{theorem1}, offer an alternative route to achieve optimal QST infidelity scaling without requiring adaptivity. Both approaches face experimental challenges  when scaling up: our algorithm often requires generalized measurements on single copies, whereas collective measurements on four or more copies remain open.

Our algorithms are designed for full tomography, where both the number of measurements and the classical post-processing overhead scale exponentially with the number of qubits. There have been several efficient algorithms developed for retrieving partial information from quantum systems \cite{Huang2020,PhysRevLett.121.170502,Aolita2015}. Future endeavors will thus explore the integration of our algorithms with efficient partial quantum tomography methods.
Additionally, our methods are tailored to discrete-variable systems; extending them to the continuous-variable regime would also be valuable for further research.

Our work highlights the power of adaptivity in quantum tomography, and opens an avenue for various subsequent research, such as whether mixed input states suffice to achieve the optimal infidelity scaling in AAPT, and how to further optimize the infidelity value in tomography. As such, our work can be potentially useful in investigating the performance limits of quantum tomography methods and designing more efficient tomography algorithms. In addition, control strategies have been proven effective in enhancing the accuracy of Hamiltonian estimation \cite{PhysRevLett.115.110401,PhysRevLett.117.160801}. Therefore, it is also desirable to investigate the application of control strategies to improve the estimation accuracy in quantum tomography, together with the practical benefits of the optimized tomography methods in quantum metrology and more general quantum information processing tasks.

\section{Methods}\label{methods}

\subsection{Two-step adaptive quantum state tomography and detector tomography}\label{sec31}

\begin{algorithm}
\caption{Two-step Adaptive QST}
\label{alg:adaptive_qst}
\begin{algorithmic}[1]

\Statex \textbf{Resource Allocation}: Total copies: $N$; allocation factor: $\alpha \in (0,1)$; Step--1 copies: $N_0=\alpha N$; Step--2 copies: $N-N_0$; 

\Statex \textbf{Step 1: Static State Eigenbasis Estimation}
\State $\tilde{\rho} \leftarrow \mathrm{MLE/LRE}$
\State $ \tilde{U} \leftarrow [\lvert\tilde{\lambda}_{1}\rangle, \cdots, \lvert\tilde{\lambda}_{d}\rangle] 
       \leftarrow \mathrm{Eigenvectors}(\tilde{\rho})$

\Statex \textbf{Step 2: Adaptive State Eigenvalue Estimation}
\State Measure $N-N_0$ copies using POVM 
       $\{\lvert\tilde{\lambda}_{i}\rangle\langle\tilde{\lambda}_{i}\rvert\}_{i=1}^{d}$
\State $\{\hat{p}_i\}_{i=1}^d \leftarrow \text{measurement outcome frequencies}$
\State $\hat{\rho} \leftarrow 
       \tilde{U}\, \operatorname{diag}(\hat{\lambda}_1,\cdots,\hat{\lambda}_d)\,
       \tilde{U}^{\dagger}$, $\hat{\lambda}_i \leftarrow \hat{p}_i\; (i=1,\cdots,d)$

\State \Return $\hat{\rho}$
\end{algorithmic}
\end{algorithm}

First we present a two-step adaptive QST method. In Step-1, we apply MLE \cite{qstmle}  or LRE \cite{Qi2013} (without employing the correction algorithm in \cite{effqst}) with $ N_0 =\alpha N$ copies where $\alpha \in (0,1)$ is a constant independent of $N$ and obtain a preliminary estimator $$ \tilde\rho=\tilde U \operatorname{diag}(\tilde\lambda_1,\cdots,\tilde\lambda_d) \tilde U^{\dagger} $$
where $ \tilde U=[\lvert\tilde{\lambda}_{1}\rangle,\cdots,\lvert\tilde{\lambda}_{d}\rangle] $.
Here, when LRE is solely employed, $  \tilde\rho $ may have negative eigenvalues, which is acceptable because only $ \tilde{ U} $ (not $ \{\tilde\lambda_i\} $) is required for the subsequent step. Step-1 lays the foundation for C1, i.e., $\tilde\rho$ already accurately estimates the entirety of $\rho$, as $\mathbb{E}(\|\tilde\rho-\rho\|^2)=O(1/N_0)=O(1/N)$ \cite{qstmle,Qi2013}.

Then in Step-2, we use the eigenbasis of $\tilde\rho$, $\{\lvert\tilde{\lambda}_{i}\rangle\langle\tilde{\lambda}_{i}\lvert\}_{i=1}^{d}$, as the new measurement operators consuming the remaining $ N-N_0 $ state copies, and obtain the corresponding new measurement frequency data $ \{\hat p_i\}_{i=1}^{d} $.
These adaptive measurement operators  correspond to a POVM because $ \sum_{i=1}^{d}\lvert\tilde{\lambda}_{i}\rangle\langle\tilde{\lambda}_{i}\lvert=I_{d}$ and we thus have $ \sum_{i=1}^{d} \hat p_i=1$. We set the final estimated eigenvalues to be $ \hat{\lambda}_{i}= \hat p_i$ and the final estimator to be
\begin{equation*}
	\hat{\rho}=\tilde{U} \operatorname{diag}(\hat{\lambda}_1, \cdots, \hat{\lambda}_d) \tilde{U}^{\dagger}.
\end{equation*}
Step-2 measures almost in the eigenbases of $\rho$, and directly setting the measurement frequency as the estimated eigenvalues can satisfy C2. See the intuitive explanation at the end of Supplementary Information Section II for why measurements in (or close to) the eigenbasis yield higher accuracy. The final estimator $\hat\rho$ is physical because $ \hat{\lambda}_{i}= \hat p_i\geq0 $ and $\operatorname{Tr}(\hat{\rho})= \sum_{i=1}^{d} \hat{\lambda}_{i}=\sum_{i=1}^{d} \hat p_i=1$. All the procedures are outlined in Algorithm \ref{alg:adaptive_qst}.

 Then the MSE scales as
	\begin{equation*}
		\begin{aligned}
			\mathbb{E}\|  \hat{\rho}-{\rho}\|  ^2 = \;& O\left(\frac{1}{N_0}\right)+O\left(\frac{1}{N-N_0}\right)\\
            &+O\left(\frac{1}{\sqrt{N_0\left(N-N_0\right)}}\right)\\
			&+O\left(\frac{1}{N_0^{3 / 4}\left(N-N_0\right)^{1 / 4}}\right),
		\end{aligned}
	\end{equation*}
	and the estimated zero eigenvalues of $ \hat{\rho} $ scale as
    \begin{equation*}
        \mathbb{E}\hat{\lambda}_i =O\left(\frac{1}{N_0}\right)+O\left(\frac{1}{\sqrt{N_0(N-N_0)}}\right) \text{for}\; r+1\leq i\leq d.
    \end{equation*}
    Therefore,
	if $ N_0=\alpha N $, where $ 0<\alpha<1$ is a constant, 
	the infidelity $ \mathbb{E}\left(1-F\left(\hat{\rho}, \rho\right)\right)$ can achieve the optimal scaling $O\left(1/N\right)  $. 
The proof is presented in Supplementary Section III-A. The actual performance of this algorithm also depends on the resource allocation proportion $\alpha$, and its optimal value remains open. References \cite{PhysRevLett.111.183601,PhysRevA.98.012339} chose $\alpha=\frac{1}{2}$ for their versions of two-step adaptive QST algorithms. In this work, we simulate both $\alpha=\frac{1}{2}$ and $\frac{9}{10}$ in Supplementary Section IV. Our two-step adaptive approach requires only one POVM set with $d$ elements in Step-2, reducing the number of measurement operators in comparison with the method in \cite{pereira2020high,PhysRevA.98.012339} which requires $2d-1+(d\ \text{mod}\ 2)$ POVM sets, each with at least $d$ elements. In fact, all the POVM elements in Step-2 of our two-step adaptive QST algorithm are also measured in the adaptive step of the algorithm in \cite{pereira2020high}. But the converse is not true; i.e., Ref. \cite{pereira2020high} requires strictly more POVM elements.

For QDT, we propose a similar two-step adaptive algorithm. 
Given a total budget of $N$ copies, we first allocate $N_{0}=\alpha N$ ($0<\alpha<1$) to obtain a preliminary POVM estimate $\{\bar P_i\}_{i=1}^n$ using either static MLE \cite{PhysRevA.64.024102}  or the two-stage method \cite{wang2019twostage}. 
Each operator is diagonalized as 
$\bar P_i=\bar U_i\,\mathrm{diag}(\bar\lambda^{i}_{1},\dots,\bar\lambda^{i}_{d})\,\bar U_i^\dagger$, yielding an estimated eigenbasis $\{\lvert\bar\lambda^{\,i}_{j}\rangle\}_{j=1}^{d}$.
 For each $1\leq i\leq n$ and $1\leq j\leq d$, we measure $(N-N_0)/(nd)$ copies of the probe state 
$\tilde\rho^{\,i}_{j}=\lvert\bar\lambda^{\,i}_{j}\rangle\langle\bar\lambda^{\,i}_{j}\rvert$  with the POVM $\{P_k\}_{k=1}^n$. 
The empirical frequencies $\tilde p^{\,i,k}_{j}$ give the adaptive eigenvalue estimates 
$\tilde\lambda^{\,i}_{j}=\tilde p^{\,i,i}_{j}$, 
and we form the refined operators 
$\tilde P_i=\bar U_i\,\mathrm{diag}(\tilde\lambda^{i}_{1},\dots,\tilde\lambda^{i}_{d})\,\bar U_i^\dagger$. Since the set $\{\tilde P_i\}$ may not satisfy $\sum_i\tilde P_i=I_d$, we enforce physicality via the normalization
$\hat P_i \triangleq \tilde I_d^{-1/2}\,\tilde P_i\,\tilde I_d^{-1/2}$ where
$\tilde I_d \triangleq  \sum_{i=1}^{n}\tilde P_i$, 
guaranteeing $\hat P_i\ge0$ and $\sum_i\hat P_i=I_d$. 
This yields the final estimator $\{\hat P_i\}_{i=1}^{n}$. Algorithm~\ref{alg:adaptive_qdt} provides a succinct summary of the method.
We can prove both C1 and C2 are satisfied and Theorem~\ref{theorem1} yields $\mathbb{E}(1-F(\hat{P}_{i}, P_{i})) = O(1/N) $ for $1\leq i\leq n$, with the details shown in Supplementary Section III-B. This algorithm reduces the required number of probe states from $nd^2$ (as in \cite{xiao2021optimal}) to $nd$, while still achieving $O(1/N)$ infidelity.

\begin{algorithm}
\caption{Two-step Adaptive QDT}
\label{alg:adaptive_qdt}
\begin{algorithmic}[1]

\Statex \textbf{Resource Allocation}: Total copies: $N$; allocation factor: $\alpha \in (0,1)$; Step--1 copies: $N_0=\alpha N$; Step--2 copies per probe state: $\frac{N-N_0}{nd}$; 

\Statex \textbf{Step 1: Static POVM Eigenbasis Estimation}
\State $\{\bar{P}_i\}_{i=1}^{n} \leftarrow$ MLE/Two-stage algorithm
\State $ \bar{U}_{i} \leftarrow [\lvert\bar{\lambda}_{1}^{i}\rangle, \cdots, \lvert\bar{\lambda}_{d}^{i}\rangle]  \leftarrow  \mathrm{Eigenvectors}(\bar{P}_i)$

\Statex \textbf{Step 2: Adaptive POVM Eigenvalue Estimation}
\State $\tilde{\rho}_{j}^{\,i} \leftarrow 
       \lvert \bar{\lambda}_{j}^{\,i} \rangle\langle \bar{\lambda}_{j}^{\,i} \rvert$
\State $\tilde{\lambda}_{j}^{\,i} \leftarrow \tilde{p}_{j}^{\,i,i}$, $\tilde{p}_{j}^{\,i,k} \leftarrow 
       \text{Measure } \tilde{\rho}_{j}^{\,i} \text{ with } \{P_k\}_{k=1}^{n}$
\State $\tilde{P}_i \leftarrow 
       \bar{U}_i\,\mathrm{diag}(\tilde{\lambda}_{1}^{i},\cdots,\tilde{\lambda}_{d}^{i})\,\bar{U}_i^{\dagger}$

\Statex \textbf{Normalization}
\State $\hat{P}_i \leftarrow \tilde{I}_d^{-1/2}\,\tilde{P}_i\,\tilde{I}_d^{-1/2}$, $\tilde{I}_d \leftarrow \sum_{i=1}^{n} \tilde{P}_i$

\State \Return $\{\hat{P}_i\}_{i=1}^{n}$
\end{algorithmic}
\end{algorithm}

\subsection{Adaptive ancilla-assisted process tomography}\label{sec31b}

Here, we propose an adaptive  AAPT algorithm allowing for both trace-preserving and non-trace-preserving processes. In
AAPT  \cite{aaqpt,meaqo},  apart from the principal system $A$, an ancilla system $B$ is introduced, with the input state $ \sigma^{\text{in}} $ and the measurement of the output state $ \sigma^{\text{out}} $  both on the composite space $\mathbb H_A\otimes \mathbb H_B$, as illustrated in Supplementary Section I-D. Given a pure input state $ \sigma^{\text{in}}=|\Phi\rangle\langle\Phi| $, where 
$|\Phi\rangle=\sum_{i=1}^{d} h_i\, (U|i\rangle)\otimes(V|i\rangle)$ with $h_i>0$ and $U,V$ unitary, we begin by 
performing adaptive QST on $\sigma^{\text{out}}$ using Algorithm~\ref{alg:adaptive_qst} with $N$ copies, yielding an estimate $\hat{\sigma}^{\mathrm{out}}$ (Steps~1--2).
We set 
$H\triangleq\mathrm{diag}(h_1,\dots,h_d)$ and obtain a  pseudo
process matrix
$\tilde{X}_{0}=(I\otimes U^{*}H^{-1}V^{\dagger})\,\hat{\sigma}^{\mathrm{out}}\,(I\otimes VH^{-1}U^{T})$, 
which ensures $\tilde{X}_0\ge 0$ but often violates the partial-trace constraint. 
We finally impose physicality by setting $\hat X=\mathcal{P}(\tilde{X}_0)$, where $\mathcal{P}(\cdot)$ is the 
algorithm of \cite{xiaoqpt} which modifies $\tilde{X}_0$ to $\hat X$ such that $\text{Tr}_1(\hat X)=I_d$ or $\text{Tr}_1(\hat X)\leq I_d$ is further satisfied for trace-preserving or non-trace-preserving processes (see Supplementary Section II-C for details). The complete procedure is outlined in Algorithm~\ref{alg:aapt}. This yields the final three-step adaptive AAPT estimate $\hat X$. The detailed procedures and proof of optimal infidelity scaling are in Supplementary Section III-C.

\begin{algorithm}
\caption{Three-step Adaptive AAPT}
\label{alg:aapt}
\begin{algorithmic}[1]

\Statex \textbf{Input}: $\sigma^{\mathrm{in}}=|\Phi\rangle\langle\Phi|$, where $|\Phi\rangle=\sum_{i=1}^{d} h_i (U|i\rangle)\otimes(V|i\rangle)$;

\Statex \textbf{Steps 1-2: Adaptive QST on Output}

\State $\hat{\sigma}^{\mathrm{out}} \leftarrow$ adaptive QST using Algorithm \ref{alg:adaptive_qst} with $N$ copies

\Statex \textbf{Step 3: Construct Estimate and Partial-Trace Correction}
\State Let $H=\mathrm{diag}(h_1,\cdots,h_d)$ and 
$\tilde{X}_{0} \leftarrow (I\otimes U^{*}H^{-1}V^{\dagger})\,\hat{\sigma}^{\mathrm{out}}\,(I\otimes VH^{-1}U^{T})$

\State  $\hat{X} \leftarrow \mathcal{P}(\tilde{X}_0)$

\State \Return $\hat{X}$

\end{algorithmic}
\end{algorithm}

\subsection{Experimental details}

\begin{figure*}[!htb]     
    \centering 
    \includegraphics[width=6.7in]{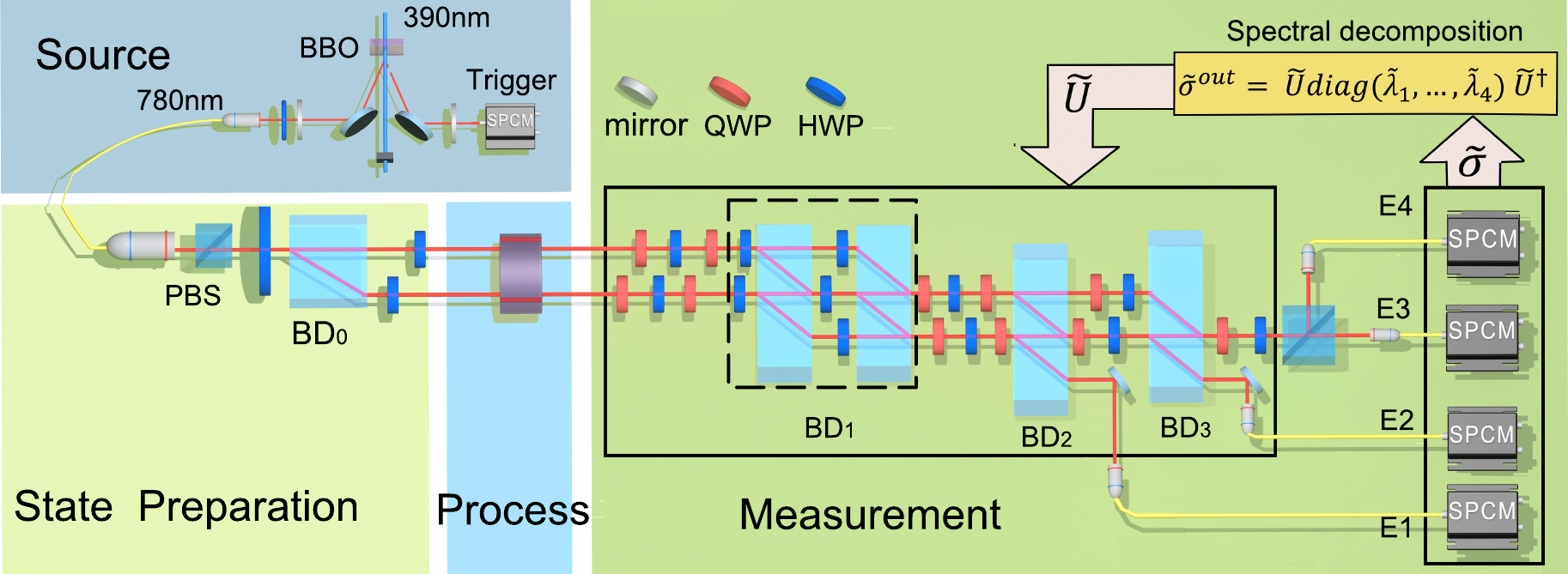}
    \centering{\caption{The experimental setup involves the use of the spontaneous parametric down-conversion (SPDC) process to generate a single-photon source. In the state preparation module, the polarization qubit and the path qubit of the photons are entangled to prepare a maximally entangled state. Arbitrary unitary process can be achieved using a combination of wave plates QWP-HWP-QWP, and a quartz crystal can be utilized to implement phase damping process. By adjusting the rotation angles of the wave plates in the measurement module, arbitrary two-qubit state projective measurements can be performed. Key components include a polarizing beam splitter (PBS), half-wave plates (HWPs), quarter-wave plates (QWPs), Single-Photon Counting Module (SPCM), and beam displacers (BDs). }\label{experimental1}}
     
\end{figure*}

The experimental setup for performing adaptive AAPT is shown in Fig. \ref{experimental1}. The apparatus consists of three modules: the state preparation module, the unknown quantum process, and the measurement. A two-qubit Bell state is used as the fixed input state, corresponding to the AAPT framework, where one qubit serves as an ancilla qubit and the other as the target qubit. The unknown quantum process acts only on the target qubit. Therefore, in our experiment, we choose to encode the target qubit in the polarization of a photon and the ancilla qubit in the path of the photon, allowing the quantum process to act on the polarization.

In the state preparation stage, a Ti-sapphire laser first outputs a light pulse centered at a wavelength of $780$ nm, with a repetition rate of approximately $76$ MHz and a pulse duration of about $150$ fs. After passing through a frequency doubler, the pulse frequency is doubled, converting the infrared light into ultraviolet light. Subsequently, the enhanced ultraviolet pulse is focused onto a BBO crystal that is specifically cut to achieve a type-II phase-matched spontaneous parametric down-conversion (SPDC). High-intensity ultraviolet pulses within the crystal trigger nonlinear optical effects, thereby splitting a high-energy ultraviolet photon into a pair of lower-energy infrared photons.

One of the photons serves as a trigger and is detected by a single-photon counter, while the other photon is used as a single-photon source and input at the beginning of the optical path. This photon is prepared in the $\lvert H\rangle$ state  after passing through a polarizing beam splitter (PBS). It then becomes a $\frac{1}{\sqrt{2}}(\lvert H \rangle + \lvert V\rangle)$ state after passing through a half-wave plate oriented at $22.5^{\circ}$. To transform the polarization state into a path state, a beam displacer (BD) is used to direct the $H$ component into path $1$ and the $V$ component into path $0$. A half-wave plate oriented at $0^{\circ}$ is placed in path $1$, and a half-wave plate oriented at $90^{\circ}$ is placed in path $0$. This successfully prepares the photon in the $\frac{1}{\sqrt{2}}(\lvert H,1\rangle + \lvert V,0\rangle)$ state, which represents the maximally entangled state of the polarization and path qubits.

Subsequently, both beams of the paths pass through an unknown quantum process and are then sent to the measurement module for projective measurements. By adjusting the rotation angles of the HWP (half-wave plate) and QWP (quarter-wave plate) combinations in the measurement optical path, arbitrary two-qubit state projections can be achieved (see Supplementary Section V-A).

For the non-adaptive method, here we apply two-qubit Cube measurements. There are nine detectors and each detector has four POVM elements, corresponding to the four exits in Fig. \ref{experimental1}, which are detected by four single-photon counters. The number of resources used for each detector is ${N}/9$. In the first step of adaptive QST procedure, a Cube measurement is used, utilizing  $N_0=0.5N$ resources to determine the adaptive measurement operators $\{\lvert  \tilde{\lambda}_{i}\rangle\langle\tilde{\lambda}_{i}\lvert  \}_{i=1}^{d}$. In the second step, a single measurement is conducted using these operators, consuming $N-N_0=0.5N$ resources. In our experiment, for the non-adaptive approach, the total number of resources ${N}$ used is set as $27$, $99$, $297$, $945$, $2997$, $9486$, $29997$ and $94869$, respectively. For the two-step adaptive experiment, to guarantee the resource numbers to be integers, ${N}$ is adjusted to $36$, $108$, $306$, $954$, $3006$, $9486$, $30006$ and $94878$, respectively. Each measurement is repeated $100$ times to obtain the average infidelity and error bars.

\textbf{Supplementary information.}
Supplementary Information accompanies this paper at the additional supplementary materials.

\textbf{Acknowledgments.} We thank Dr. Clemens Gneiting, Dr. Yanming Che, and Prof. Yuerui Lu for valuable comments on the manuscript. This research was supported by the Quantum Science and Technology-National Science and Technology Major Project (Grant No. 2023ZD0301400), the Australian Research Council Future Fellowship Funding Scheme under Project FT220100656, the Discovery Project Funding Scheme under Projects DP200102945, DP210101938, DP240101494, the National Natural Science Foundation of China (62173229, 12288201, 62222512, 12104439, and 12134014), the CAS Project for Young Scientists in Basic Research (YSBR-131), and the Anhui Provincial Natural Science Foundation (Grant No.2208085J03).
F.N. is supported in part by the Japan Science and Technology Agency (JST) (via the CREST Quantum Frontiers program Grant No. JPMJCR24I2,
the Quantum Leap Flagship Program (Q-LEAP), and the Moonshot R\&D Grant Number JPMJMS2061),
and the Office of Naval Research (ONR) Global (via Grant No. N62909-23-1-2074).

\textbf{Competing interests.}
The authors declare no competing interests.

\textbf{Author contributions}
S.X. and Y.W. developed the theory and methods with
assistance from J.Z., I.R.P. and D.D.;  X.W., W.Z.Y. and Z.H. designed and implemented the experiments with assistance from G.Y.X., S.X., Y.W. and D.D.;
S.X., Y.W., J.Z., I.R.P. and D.D. performed numerical simulations and analyzed the experimental data with assistance from X.W., W.Z.Y., Z.H. and G.Y.X.; All the authors discussed and analyzed the results and contributed to the writing of the manuscript; D.D. supervised the project.

\bibliography{adaptive}

\onecolumngrid

\newpage 


\setcounter{equation}{0}
\renewcommand{\theequation}{S.\arabic{equation}}

\setcounter{section}{0}

\setcounter{figure}{0}
\renewcommand{\thefigure}{S\arabic{figure}}

\begin{center}
	\large {\bf Supplemental Material}
\end{center}

\section*{Notation and key symbols}
We present all relevant notations below, with the acronyms summarized in Table~\ref{tablet2}, and the key symbols used in our proposed adaptive algorithms listed in Table~\ref{tablet1}.

\textbf{Notation:} The $ i $-th row and $ j $-th column of a matrix $ X $ is $(X)_{ij} $. The $ j $-th column of $ X $ is $ \operatorname{col}_{j}(X) $. The transpose of $X$ is $X^T$. The conjugate $ (*) $ and transpose of $X$ is $X^\dagger$. The sets of  real and complex numbers are $\mathbb{R}$ and $\mathbb{C}$, respectively. The sets of  $d$-dimension complex vectors and  $d\times d$ complex matrices are $\mathbb{C}^d$ and $\mathbb{C}^{d\times d}$, respectively. The identity matrix is $ I $. $\rm i=\sqrt{-1}$.     The Dirac Delta function is $ \delta(\cdot) $. As $ N $ tends to infinity, the convergence of the sequence $ a(N) $ to $ b $
is denoted as $ a(N)\rightarrow b $. 
The trace of $X$ is $\text{Tr}(X)$. The Frobenius norm of a matrix $X$ is denoted as $||X||$ and the 2-norm of a vector $ x $ is $||x||$.  We use the density matrix $ \rho $ to  represent a quantum state and use a unit complex vector  $|\psi\rangle$ to represent a pure state. The conjugate $|\psi\rangle^{\dagger}$  is usually denoted as $\langle\psi|$, and $\langle\psi|\cdot |\phi\rangle $ is often simplified as  $\langle\psi|\phi\rangle $.
Denote the standard basis as $ \{|i\rangle\}_{i=1}^n $,  such that $ \langle i| j\rangle=\delta_{ij}$.
The estimate of $X$ is $ \hat X $. The inner product of two matrices $X$ and $Y$ is defined as $\langle X, Y\rangle\triangleq\text{Tr}(X^\dagger Y)$.
The inner product of two vectors $x$ and $y$ is defined as $\langle x,
y\rangle\triangleq x^\dagger y$. The tensor product of $A$ and $B$ is $A\otimes B$. We define a Hilbert space by $\mathbb{H}$.
The partial trace of $ X \in \mathbb H_1\otimes \mathbb H_2$ on the space $\mathbb H_1$ is $\text{Tr}_1(X)$.
The diagonal elements of the diagonal matrix  $\operatorname{diag}(a)$ is the vector $a$.
For any positive semidefinite $X_{d \times d}$ with spectral decomposition $X=U P U^{\dagger}$, define $\sqrt{X}$ or $X^{1/2}$ as $U \operatorname{diag}(\sqrt{(P)_{11}}, \sqrt{(P)_{22}}$, $\ldots, \sqrt{(P)_{d d}}) U^{\dagger}$.  For any positive definite $X_{d \times d}$, define $X^{-1/2}$ as $U \operatorname{diag}(1/\sqrt{(P)_{11}}, 1/\sqrt{(P)_{22}},$ $\ldots, 1/\sqrt{(P)_{d d}}) U^{\dagger}$. Pauli matrices are
\begin{equation*}
	\sigma_x=\begin{bmatrix}
		0 & 1 \\
		1 & 0
	\end{bmatrix},\quad
	\sigma_y=\begin{bmatrix}
		0 & -\mathrm{i} \\
		\mathrm{i} & 0
	\end{bmatrix},\quad
	\sigma_z=\begin{bmatrix}
		1 & 0 \\
		0 & -1
	\end{bmatrix}.
\end{equation*}

\begin{table}[htbp]
	\centering
	\caption{List of acronyms.}
	\label{tablet2}
	\renewcommand{\arraystretch}{1.2}
	\begin{tabular}{ll}
		\hline
		\textbf{Acronym} & \textbf{Full Name} \\
		\hline
		AAPT & Ancilla-Assisted Process Tomography \\
		BD & Beam Displacer \\
		GM Bound & Gill–Massar Bound \\
		HWPs & Half-Wave Plates \\
		LRE & Linear Regression Estimation \\
		MLE & Maximum Likelihood Estimation \\
		MSE & Mean Squared Error \\
		PBS & Polarizing Beam Splitter \\
		POVM & Positive Operator-Valued Measure \\
		QDT & Quantum Detector Tomography \\
		QPT & Quantum Process Tomography \\
		QPST & Quantum Pseudo-State Tomography \\
		QST & Quantum State Tomography \\
		QWPs & Quarter-Wave Plates \\
		SPDC & Spontaneous Parametric Down-Conversion \\
		SVD  & Singular Value Decomposition \\
		\hline
	\end{tabular}
\end{table}

\begin{table}[htbp]
	\centering
	\caption{List of key symbols.}
	\label{tablet1}
	\renewcommand{\arraystretch}{1.3}
	\begin{tabular}{lll}
		\hline
		\textbf{Symbol} & \textbf{Description} & \textbf{Defined in} \\
		\hline
		$N$ & Total number of state copies used in QST/QDT/QPT & -- \\
		$S$, $\hat{S}$ & A positive semidefinite matrix and its estimate & Sec.~\ref{defs} \\
		$F_{d,p}(\hat{S}, S)$ & Previously defined fidelity between $S$ and $\hat{S}$ & Eq.~\eqref{s1} \\
		$F(\hat{S}, S)$ & Proposed fidelity between $S$ and $\hat{S}$ & Eq.~\eqref{s3} \\
		$\rho$ & True quantum state (density operator) & Sec.~\ref{pre1} \\
		$F_{s}(\sigma, \rho)$ & Conventional fidelity between two quantum states $\sigma$ and $\rho$& Eq.~\eqref{deinf} \\
		$\tilde{\rho}$ & Intermediate estimate of $\rho$ after Step 1 (adaptive QST) & Sec.~\ref{sec61} \\
		$\tilde\lambda_i$, $\lvert\tilde{\lambda}_{i}\rangle$ & $i$-th eigenvalue and eigenvector of $\tilde{\rho}$ & Sec.~\ref{sec61} \\
		$\hat{\rho}$ & Final estimate of $\rho$ (adaptive QST) & Sec.~\ref{sec61} \\
		$\hat\lambda_i$, $\lvert\hat{\lambda}_{i}\rangle$ & $i$-th eigenvalue and eigenvector of $\hat{\rho}$ & Sec.~\ref{sec61} \\
		$P_i$ & True $i$-th POVM element & Sec.~\ref{pre2} \\
		$\bar{P}_i$ & Step-1 estimate of $P_i$ (adaptive QDT) & Eq.~\eqref{qq1} \\
		$\bar\lambda_j^i$, $\lvert\bar{\lambda}_j^i\rangle$ & $j$-th eigenvalue and eigenvector of $\bar{P}_i$ & Eq.~\eqref{qq1} \\
		$\tilde{P}_i$ & Step-2 estimate of $P_i$ (adaptive QDT) & Eq.~\eqref{qq2} \\
		$\tilde\lambda_j^i$, $\lvert\tilde{\lambda}_j^i\rangle$ & $j$-th eigenvalue and eigenvector of $\tilde{P}_i$ & Eq.~\eqref{qq2} \\
		$\tilde{I}_{d}$& Sum of $\{\tilde{P}_i\}$, $\sum_{i} \tilde{P}_i=\tilde{I}_{d}$ & Eq.~\eqref{tildei} \\
		$\hat{P}_i$ & Final estimate of $P_i$ (adaptive QDT) & Eq.~\eqref{qdtpre} \\
		$\hat{\sigma}^{\text{out}}$ & Estimated output state in adaptive AAPT & Eq.~\eqref{qqpst} \\
		$X$ & True process matrix of the quantum channel $\mathcal{E}$ & Sec.~\ref{pre3} \\
		$\tilde{X}_0$ & Intermediate estimate of $X$ in adaptive AAPT & Eq.~\eqref{ghzprue} \\
		$\hat{X}$ & Final estimate of $X$ in adaptive AAPT & Eq.~\eqref{aaptpro} \\
		$\operatorname{vec}$ & Column-vectorization operator & Eq.~\eqref{vec} \\
		\hline
	\end{tabular}
\end{table}

\section{Preliminaries on quantum tomography }\label{Sec2}
There are three typical problems in quantum tomography; i.e., quantum state/detector/process tomography (QST/QDT/QPT),
aiming at fully estimating an unknown state/detector/process \cite{qci,You2011,Gebhart2023,Lundeen2009,Bialczak2010}.
Here we introduce several existing quantum tomography algorithms.

\subsection{Quantum state tomography and linear regression estimation}\label{pre1}
Quantum state tomography (QST)  aims to estimate an unknown quantum state $ \rho \in \mathbb{C}^{d\times d} $ satisfying $ \rho=\rho^{\dagger} $, $ \rho\geq0 $, and $ \operatorname{Tr}(\rho)=1 $. When $\rho=|\psi\rangle\langle\psi|$, where $|\psi\rangle\in\mathbb{C}^{d}$ is a unit vector, $ \rho $ is called a pure state. 
When we apply a set of measurement operators $\left\{P_i\right\}_{i=1}^n$  to the quantum state $\rho$, the probability of obtaining the $i$-th result is given by the Born's rule
\begin{equation}
	p_i=\operatorname{Tr}\left(P_i \rho\right).
\end{equation}
In practical experiments, suppose that $N$ (also called the resource number; i.e., the number of state copies) identical copies of $\rho$ are prepared and the $i$-th results occur $N_i$ times. Then the measured frequency  $\hat p_i=\frac{N_i}{N}$ gives an experimental
estimation of the true value $p_i$ and the measurement error is $ e_i=\hat p_i- p_i $.
According to the central
limit theorem, the distribution of $ e_i $ converges to a normal
distribution \cite{MU2020108837,Qi2013} with mean zero and variance $\left(p_i-p_i^{2}\right)/N$.
The task of QST is to identify an unknown density operator $ \rho $ using the measurement results. Among various algorithms to perform QST such as Maximum Likelihood Estimation (MLE) \cite{qstmle}, Bayesian Mean Estimation \cite{qstbme} and Linear Regression Estimation (LRE)  \cite{Qi2013}, here we consider the  LRE method since it can give  analytical formulas for the error upper bound and computational complexity.

Ref. \cite{Qi2013} formulated QST  into a  linear equation
\begin{equation}
	\mathcal{Y}=\mathcal{X}\phi +\mathcal{e},
\end{equation}
where $ \mathcal{Y} $ is a vector containing all of the measured frequency, $ \mathcal{X}$ is the parameterization matrix for the measurement operators, $ \phi $ is the parameterization vector of $\rho$, and $\mathcal{e}  $ is the vector of the measurement errors.
We assume that the measurement operators are informationally complete, i.e., $ \mathcal{X} $ has a full column rank. Thus, the unique least squares solution is 
\begin{equation}\label{lse}
	\tilde\phi=(\mathcal{X}^{\dagger}\mathcal{X})^{-1}\mathcal{X}^{\dagger}\mathcal{Y}.
\end{equation}
Using $ \tilde\phi $, we can reconstruct the estimated quantum state $ \tilde{ \chi} $.
Ref. \cite{Qi2013} proved that the computational complexity of this algorithm is $ O(Ld^2) $, where $ L $ is the type number of different measurement operators, and an error upper bound for the MSE is
\begin{equation}\label{rhomse}
	\mathbb{E}(\|\tilde{\chi}-{\rho}\|^2) \leq \frac{J}{4 N} \operatorname{Tr}\left(\mathcal{X}^{\dagger}\mathcal{X} \right)^{-1},
\end{equation}
where  $ J $ is the number of the POVM sets (i.e., measurement basis sets), $ N $ is the total resource number, and $ \mathbb{E} $ denotes  expectation w.r.t. all possible measurement
results. 

The LRE estimate $ \tilde{ \chi} $ may not satisfy the positive semidefinite constraint and a fast  correction algorithm in \cite{effqst} can subsequently be  applied. We assume that the spectral decomposition of $ \tilde{ \chi} $ is
$$ \tilde{ \chi}=\tilde U \operatorname{diag}(\tilde\lambda_1,\cdots,\tilde\lambda_d) \tilde U^{\dagger}, $$
where $ \tilde\lambda_1 \geq \cdots \geq \tilde\lambda_d $ and $ \sum_{i=1}^{d}\tilde\lambda_{i}=1 $.
Using the fast algorithm in \cite{effqst}, we need to find the maximum $ k \:(1\leq k \leq d)  $ such that 
\begin{equation}\label{qstcorrect}
	\tilde{\lambda}_k+\frac{\sum_{i=k+1}^d \tilde{\lambda}_i}{k} \geq 0, \quad \tilde{\lambda}_{k+1}+\frac{\sum_{i=k+2}^d \tilde{\lambda}_i}{k+1}<0.
\end{equation}
Then $ \hat{\lambda}_j=\tilde{\lambda}_j+\frac{\sum_{i=k+1}^d \tilde{\lambda}_j}{k} $ for $ 1\leq j \leq k $ and $ \hat{\lambda}_j=0 $ for $ k+1\leq j \leq d $. Thus, the final physical estimate is
$$ \hat\chi=\tilde U \operatorname{diag}(\hat\lambda_1,\cdots,\hat\lambda_d) \tilde U^{\dagger}. $$
Ref. \cite{Qi2013} utilized the above algorithm without proving the MSE scaling.
Here we show that the scaling of the MSE is 
\begin{equation}\label{qst1}
	\mathbb{E}(\|\hat\chi-{\rho}\|^2)=O\left(\frac{1}{N}\right).
\end{equation}
We first introduce the following lemma.
\begin{lemma}(\cite{bhatia2007perturbation} Theorem 8.1 and Theorem 28.3)\label{lemma3}
	Let $X$, $Y$ be Hermitian matrices with eigenvalues $\lambda_1(X)\geq\cdots\geq\lambda_n(X)$ and $\lambda_1(Y)\geq\cdots\geq\lambda_n(Y)$, respectively. Then
	\begin{equation}\label{weyl}
		\max_j|\lambda_j(X)-\lambda_j(Y)|\leq||X-Y||,
	\end{equation}
	and
	\begin{equation}\label{weyl2}
		\sum_{j=1}^{n}\left(\lambda_j(X)-\lambda_j(Y)\right)^2\leq||X-Y||^2.
	\end{equation}
\end{lemma}

If $ \sum_{j=k+1}^d \tilde{\lambda}_j > 0 $, we have $\tilde{\lambda}_{k+1} >0$ and
$$ \tilde{\lambda}_{k+1}+\frac{\sum_{i=k+2}^d \tilde{\lambda}_i}{k+1}=\frac{k\tilde{\lambda}_{k+1}+\sum_{i=k+1}^d \tilde{\lambda}_i}{k+1}>0, $$ which conflicts with Eq.~\eqref{qstcorrect}.
Therefore, $ \sum_{j=k+1}^d \tilde{\lambda}_j \leq 0 $. Without loss of
generality, we assume that $ \tilde{\lambda}_{k+1}\geq \cdots \geq\tilde{\lambda}_{h}\geq 0  $ and $0>\tilde{\lambda}_{h+1}\geq \cdots \geq\tilde{\lambda}_{d}.  $
For $\tilde{\lambda}_j \leq 0 $, using Lemma  \ref{lemma3} and $ \lambda_j\geq 0 $,
we have
\begin{equation}
	\tilde{\lambda}_j^2 \leq\left(\tilde{\lambda}_j-\lambda_j\right)^2 \leq \|\tilde\chi-{\rho}\|^2.
\end{equation}
Thus,  $ \mathbb E\tilde{\lambda}_j^{2}=O(1 / {N}) $ and  $\mathbb E\left(\sum_{j=h+1}^d \tilde{\lambda}_j\right)^2=O(1 / {N})$. Since $ \sum_{j=k+1}^d \tilde{\lambda}_j= \sum_{j=k+1}^h \tilde{\lambda}_j +\sum_{j=h+1}^d \tilde{\lambda}_j \leq 0$, we have $\mathbb E(\sum_{j=k+1}^h \tilde{\lambda}_j)^2=O(1 / {N})$.
Due to $\sum_{j=k+1}^h \tilde{\lambda}_j^2  \leq \left(\sum_{j=k+1}^h \tilde{\lambda}_j\right)^2  $, we have
$\mathbb{E}\left(\sum_{j=k+1}^h \tilde{\lambda}_j^2 \right) =O(1 / {N}) $.
Therefore, for $ k+1\leq j \leq d $, $ \mathbb E\tilde{\lambda}_j^2=O(1 / {N}) $ and thus
the  MSE $ \mathbb{E}\|\hat{\rho}-\tilde{ \rho}\|^2 $ scales as
\begin{equation}\label{rho2}
	\begin{aligned}
		\mathbb{E}(\|\hat\chi-\tilde{ \rho}\|^2)& =\mathbb{E} \left(\sum_{j=1}^{k} \left(\frac{\sum_{i=k+1}^d \tilde{\lambda}_i}{k}\right)^2+\sum_{j=k+1}^{d}\tilde{\lambda}_j^2\right)= O\left(\frac{1}{N}\right).
	\end{aligned}
\end{equation}
Since $ \|\hat\chi-{ \rho}\|\leq \|\hat\chi-\tilde{\chi}\|+\|\tilde\chi-{ \rho}\| $, using Eqs.~\eqref{rhomse}, \eqref{rho2} and Cauchy–Schwarz inequality, we have
\begin{equation}
	\begin{aligned}
		\mathbb{E} (\|\hat\chi-{ \rho}\|^2) \leq&\;   \mathbb{E} (\|\hat\chi-\tilde{\rho}\|^2)+\mathbb{E} (\|\tilde\chi-{ \rho}\|^2)\\
		&+2\mathbb{E} (\|\hat\chi-\tilde{\rho}\|)\mathbb{E} (\|\tilde\chi-{ \rho}\|)\\
		\leq &\;\mathbb{E} (\|\hat\chi-\tilde{\rho}\|^2)+\mathbb{E} (\|\tilde\chi-{ \rho}\|^2)\\
		&+2\sqrt{\mathbb{E} (\|\hat\chi-\tilde{\rho}\|^2) \mathbb{E} (\|\tilde\chi-{ \rho}\|^2)}\\
		=&\;O\left(\frac{1}{N}\right).
	\end{aligned}
\end{equation}

\subsection{Quantum detector tomography and the two-stage algorithm}\label{pre2}
In quantum physics, measurement is almost everywhere and the measurement device is usually called a detector, which can be characterized by a set of measurement operators $\left\{P_i\right\}_{i=1}^n$.
These operators are named a Positive-Operator-Valued Measure (POVM) and  each POVM element $P_i\in \mathbb{C}^{d\times d}$ satisfies $ P_i=P_i^{\dagger} $ and $ P_i \geq 0 $. Moreover, together they satisfy the completeness constraint $\sum_{i=1}^n P_i=I_{d}$.
The target for quantum detector tomography (QDT) is to identify the unknown POVM elements $\{P_i\}_{i=1}^n $ using known probe states $ \{\rho_j\}_{j=1}^{{M}} $ where ${M} $ is the type number of different probe states. We assume that these probe states are informationally complete and thus $ M\geq d^2 $.
Ref. \cite{wang2019twostage} proposed an analytical two-stage algorithm to solve QDT where  the computational complexity is $ O(nd^{2}M) $, and the MSE scales as
\begin{equation}\label{pmse}
	\mathbb{E}\left(\sum_{i=1}^{n}\left\|\hat{\mathcal{P}}_i-P_i\right\|^2\right)=O\left(\frac{1}{N}\right),
\end{equation}
where $\hat{\mathcal{P}}_i$ is the estimate of $P_i$ and $ N $ is the total resource number. For each probe state, the resource number is $ N/M $.
\subsection{Standard quantum process tomography and the two-stage solution}\label{pre3}
We firstly introduce the column-vectorization function $\operatorname{vec}: \mathbb{C}^{m\times n}\mapsto \mathbb{C}^{mn}$. For a matrix $A_{m\times n}$, 
\begin{equation}\label{vec}
	\begin{aligned}
		\operatorname{vec}(A_{m\times n})\triangleq&[(A)_{11},(A)_{21},\cdots,(A)_{m1},(A)_{12},\cdots,(A)_{m2},\cdots,(A)_{1n},\cdots,(A)_{mn}]^T.
	\end{aligned}
\end{equation}
We also define $\text{vec}^{-1}(\cdot)$ which maps a $d^2\times 1$ vector into a $d\times d$ square matrix. The common properties of $ \text{vec}(\cdot) $ are listed as follows \cite{horn_johnson_2012,watrous2018theory}:
\begin{equation}\label{property2}
	\text{vec}(ABC)=(C^T\otimes A)\text{vec}(B),
\end{equation}
\begin{equation}\label{property4}
	\text{Tr}_1(\text{vec}(A)\text{vec}(B)^{\dagger})=AB^{\dagger},
\end{equation}
where $\mathbb{H}_1$ is the space of $A$.

For a $ d $-dimensional quantum system, its dynamics can be described by a completely-positive (CP) linear map
$ \mathcal{E} $ and quantum process tomography (QPT) aims to identify the unknown $ \mathcal{E} $.
If we input a quantum state $\rho^{ \text{in}}\in \mathbb{C}^{d \times d}$, using the
Kraus operator-sum representation \cite{qci}, the output state $\rho^{ \text{out}}$ is given by 
\begin{equation}\label{Kraus}
	\rho^{ \text{out }}=\mathcal{E}\left(\rho^{ \text{in }}\right)=\sum_{i=1}^{d^2} \mathcal A_{i} \rho^{\text{in }}\mathcal  A_{i}^{\dagger},
\end{equation}
where  $\mathcal A_{i} \in \mathbb{C}^{d \times d}$ and they satisfy
\begin{equation}\label{aleq}
	\sum_{i=1}^{d^2}\mathcal  A_{i}^{\dagger}\mathcal  A_{i}\leq I_{d}.
\end{equation}
Choosing $\{E_i\}_{i=1}^{d^2}$  as the natural basis  $\{|j\rangle\langle k|\}_{1\leq j,k\leq d}$, where $ i=(j-1)d+k $ \cite{qci,8022944}, we expand  $\left\{\mathcal A_i\right\}_{i=1}^{d^2}$ as
\begin{equation}
	\mathcal{A}_i=\sum_{j=1}^{d^2} c_{{ij}} E_j.
\end{equation}
We define the matrix $ (C)_{ij}=c_{{ij}} $ and the matrix $ X $ as $ X\triangleq C^T C^* $, which is called the process matrix \cite{qci,8022944}. Then $ X \in \mathbb{C}^{d^2\times d^2}$  is in a one-to-one correspondence with $ \mathcal{E} $. In addition, it satisfies $X=X^{\dagger}, X\geq0, \operatorname{Tr}_{1}\left(X\right)\leq I_d $. When the equality in Eq.~\eqref{aleq} holds, we have $ \operatorname{Tr}_{1}\left(X\right)= I_d $ in the natural basis \cite{qci,8022944} and the process $ \mathcal{E}  $ or $ X $ is trace-preserving. Otherwise, the process is non-trace-preserving.

The target for standard QPT is to identify the unknown process matrix $X $ using the known input states $ \{\rho_m^{\text{in}}\}_{m=1}^{ M } $ and the measurement operators $ \{P_l\}_{l=1}^{ L } $ where $ M $ and $  L $ are the type numbers of different input states and measurement operators, respectively. Standard QPT also assumes that the input states and measurement operators are both informationally complete and thus $  L\geq d^2, M\geq d^2 $. Ref. \cite{xiaoqpt} formulated the QPT into the following optimization problem.
\begin{problem}\label{problem2}
	Given the parameterization matrix $V$ of all the input states, the permutation matrix $ R $ and the reconstructed parameterization matrix of all the output states $\hat Y$, find a Hermitian and positive semidefinite estimated process matrix $\hat{\mathbb{X}}$ minimizing
	$$\left\|\hat{\mathbb{X}}-\operatorname{vec}^{-1}\left(R^{T}\left(I_{d^2} \otimes\left(V^{*} V^{T}\right)^{-1} V^{*}\right) \operatorname{vec}(\hat{Y})\right)\right\|,$$ 
	such that $\operatorname{Tr}_1(\hat{\mathbb{X}})\leq I_d$.
\end{problem}

Ref. \cite{xiaoqpt} proposed a  two-stage solution algorithm to solve Problem \ref{problem2}.   Let $$\hat D\triangleq\operatorname{vec}^{-1}\left(R^{T}\left(I_{d^2} \otimes\left(V^{*} V^{T}\right)^{-1} V^{*}\right) \operatorname{vec}(\hat{A})\right)$$ be a given matrix.  Stage-1 finds a Hermitian and positive semidefinite $d^2\times d^2$ matrix $\hat G$ minimizing $||\hat G-\hat D||$ by performing the spectral decomposition $\frac{\hat D+\hat D^{\dagger}}{2}=U\hat{K}U^{\dagger}  $  where $\hat{K}=\operatorname{diag}\left(k_{1}, \cdots, k_{d^{2}}\right)$ is a  diagonal matrix.
The unique optimal solution is $ \hat G=U\operatorname{diag}(z)U^{\dagger}  $, where
\begin{equation}
	z_{i}= \begin{cases}
		k_{i}, & k_{i} \geq 0, \\ 
		0, & k_{i}<0.
	\end{cases}
\end{equation}

Then, in Stage-2, define $ \hat{Q}\triangleq\operatorname{Tr}_{1}(\hat{G}) $. For trace-preserving processes, one can assume that  $\hat Q >0  $ because $ \hat Q $  converges to $ I_{d} $ as $ N $ tends to infinity. Thus, the final estimate is 
\begin{equation}\label{tss21}
	\hat{\mathbb{X}}=(I_{d} \otimes \hat{Q}^{-1/2}) \hat{G}(I_{d} \otimes \hat{Q}^{-1/2})^{\dagger}.
\end{equation}
For non-trace-preserving processes, let the spectral decomposition of $ \hat Q $ be
\begin{equation}\label{hatf}
	\hat{Q}=\hat W \operatorname{diag}\left(\hat f_{1}, \cdots, \hat f_{d}\right) \hat W^{\dagger},
\end{equation}
where $ \hat f_{1} \geq\cdots \geq\hat f_{c}>0$ and $ \hat f_{c+1}= \cdots =\hat f_{d}=0 $, i.e., the rank of $ \hat Q $ is $ \operatorname{Rank}(\hat Q)=c $. 
Then define
\begin{equation}\label{barf}
	\bar{Q}\triangleq \hat W \operatorname{diag}\left(\bar f_{1}, \cdots, \bar f_{d}\right) \hat W^{\dagger},
\end{equation}
where $ \bar f_{i}=\hat f_{i} $ for $ 1\leq i \leq c $, $ \bar f_{i}=\frac{\hat f_{c}}{N} $ for $ c+1\leq i \leq d$, and $ N $ is the number of copies. Since $ \bar{Q} $ is invertible, $ \bar{Q}^{-1/2} $ is well defined and we also define
\begin{equation}
	\tilde{Q}\triangleq \hat W \operatorname{diag}\left(\tilde{f}_{1}, \cdots, \tilde{f}_{d}\right) \hat W^{\dagger},
\end{equation}
where $ \tilde{f}_{i}=\min\left({\bar f}_{i},1\right) $ for $ 1\leq i\leq d $. Thus, $ \tilde{f}_{i} \leq 1 $ for $ 1\leq i \leq d $.
The final estimate is 
\begin{equation}\label{tss22}
	\hat{\mathbb{X}}=(I_{d} \otimes  \tilde Q^{1/2}\bar Q^{-1/2}) \hat{G}(I_{d} \otimes  \tilde  Q^{1/2}\bar Q^{-1/2})^{\dagger},
\end{equation}
which satisfies $\hat{\mathbb{X}} \geq 0$ and $\operatorname{Tr}_{1}(\hat{\mathbb{X}}) \leq I_{d}$.
The detailed analysis can be found in \cite{xiaoqpt}. 

Stage-2 in the two-stage solution will be utilized as an important correction method in this work and
Ref. \cite{xiaoqpt} proved that
\begin{equation}\label{stage2}
	\mathbb{E}\left(\left\|\hat{G}-X\right\|^2\right)=O\left(\frac{1}{N}\right),
\end{equation}
\begin{equation}\label{stage22}
	\mathbb{E}\left(\left\|\hat{\mathbb{X}}-\hat G\right\|^2\right)=O\left(\frac{1}{N}\right).
\end{equation}
In addition,
Ref. \cite{xiaoqpt} also proved that the computational complexity of the two-stage solution is $ O(MLd^2) $ and the MSE scales as 
\begin{equation}\label{processmse}
	\mathbb{E}\left(\left\|\hat{\mathbb{X}}-X\right\|^2\right)=O\left(\frac{1}{N}\right),
\end{equation}
where $ N $ is the total number of copies in SQPT.

\subsection{Ancilla-assisted quantum process tomography and the two-stage solution} 
\begin{figure}
	\centering
	\includegraphics[width=4.5in]{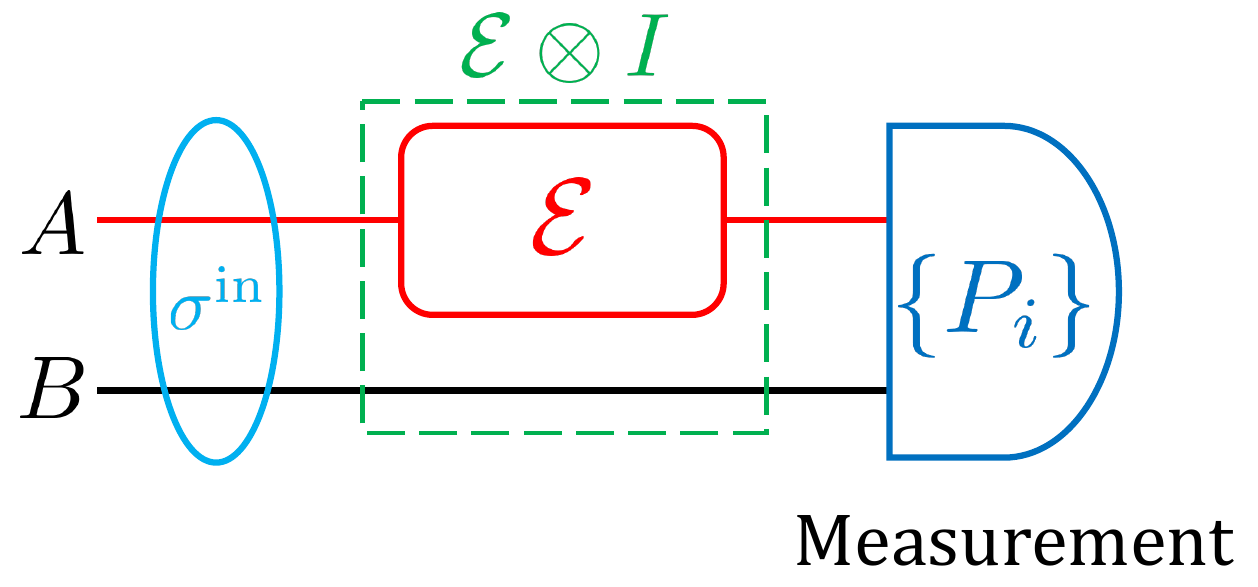}
	\centering{\caption{Schematic diagram of ancilla-assisted quantum process tomography (AAPT). The measurement operators $ \{P_i\} $ can be separable operators or entangled operators.}\label{AAPT}}
\end{figure}
Another framework for QPT is ancilla-assisted quantum process tomography (AAPT). In AAPT, an auxiliary system (ancilla) $B$ is attached to the principal system $ A$, and the input states and the measurements on the outputs are both on the extended Hilbert space as shown in Fig. \ref{AAPT}. The Hilbert space dimension of $ B $ is not smaller than that of $A$, i.e., $ d_B \geq d_A $ \cite{aaqpt}.
We firstly review the AAPT procedure  in \cite{aaqpt}.
For the input state $ \sigma^{\text{in}} $, its operator--Schmidt decomposition \cite{PhysRevA.67.052301} is
\begin{equation}\label{instate}
	\sigma^{\text{in}}=\sum_{l=1}^{d_{A}^{2}} s_l A_l \otimes B_l,
\end{equation}
where $s_l$ are non-negative real numbers and the sets $\left\{A_l\right\}$ and $\left\{B_l\right\}$ form orthonormal operator bases for systems $A$ and $B$, respectively \cite{aaqpt}. The number of nonzero terms $ s_l $ in the Schmidt decomposition is defined as 
the Schmidt number of the input state $\operatorname{Sch}(\sigma^{\text{in}})$. AAPT requires $\operatorname{Sch}(\sigma^{\text{in}})=d_{A}^2$, i.e., $ s_l>0$, $\forall l$ that satisfy $1\leq l \leq d_A^2 $, where $ \sigma^{\text{in}} $ can be separable or entangled \cite{aaqpt}.  Almost all of the states  on the combined space $A B$ can be used for AAPT because  the set of states with the Schmidt number less than $d_A^2$ is of zero measure \cite{qptre}. However, a pure $\sigma^{\text{in}}$ with $\operatorname{Sch}(\sigma^{\text{in}})=d^2  $ is necessarily entangled, because an input pure state is separable if and only if  $\operatorname{Sch}(\sigma^{\text{in}})=1  $.

After the process $ \mathcal{E} $, the output state becomes
\begin{equation}
	\sigma^{\text{out}}=(\mathcal{E} \otimes I)(\sigma^{\text{in}})=\sum_{l=1}^{d_{A}^{2}} s_l\;
	\mathcal{E}\left(A_l\right) \otimes B_l.
\end{equation}
Since
\begin{equation}\label{ptraceout}
	\operatorname{Tr}_B\left[\left(I_{d_A} \otimes B_m^{\dagger}\right) \sigma^{\text{out}}\right]=\sum_{l=1}^{d_{A}^{2}} s_l\; \mathcal{E}\left(A_l\right) \operatorname{Tr}\left(B_m^{\dagger} B_l\right)=s_m \;\mathcal{E}\left(A_m\right),
\end{equation}
we have
\begin{equation}\label{ea}
	\mathcal{E}\left(A_m\right)=\operatorname{Tr}_B\left[\left(I_{d_A} \otimes B_m^{\dagger}\right) \sigma^{\text{out}}\right] /s_m.
\end{equation}
Therefore, we can obtain $ d_{A}^2 $ linearly independent input-output relationship equations. From these equations, one can design proper algorithm to reconstruct the process $\mathcal{E}$ or $X$. Ref. \cite{xiaoaapt} also applied the two-stage solution to identify the process matrix. 
The computational complexity is $O(Ld_A^2d_B^2)$, where $  L$ is the type number of different measurement operators on the output state,
and the MSE scales as
\begin{equation}\label{aptmse}
	\mathbb{E}(\|\hat{\mathbb{X}}-X\|^2)=O\left(\frac{1}{{N}}\right),
\end{equation}
where $ N $ is the resource number for the input state.
\subsection{Maximally entangled state}

For a quantum process $ \mathcal{E} $,
using the Choi-Jamio\l{}kowski isomorphism \cite{PhysRevA.95.012302,CHOI1975285}, we have
\begin{equation}\label{choi}
	\rho_{\mathcal{E}}\triangleq({\mathcal{E}} \otimes {I})\left(|\Psi\rangle\langle\Psi | \right)=\frac{1}{d}\sum_{i, j=1}^{d, d} \mathcal{E}(|i\rangle\langle j|) \otimes|i\rangle\langle j|,
\end{equation}
where $ |\Psi\rangle=\sum_{j=1}^d|j\rangle \otimes|j\rangle / \sqrt{d} $ is a maximally entangled state ($ d=2 $ is a Bell state). When the process is trace-preserving, $ \rho_{\mathcal{E}} $ is a density matrix.
The relationship between $ \rho_{\mathcal{E}} $ and the process matrix $ X $ is \cite{PhysRevA.95.012302}
\begin{equation}\label{ghzx}
	X=d\times\rho_{\mathcal{E}}.
\end{equation}

Using the maximally entangled state $ |\Psi\rangle $, we give the following proposition.
\begin{proposition}\label{pro1}
	For any  $A, B, C, D\in \mathbb{C}^{d\times d}  $, we have
	\begin{equation}
		\begin{aligned}
			&\sum_{i, j=1}^{d, d} A B|i\rangle\langle j|D^{\dagger} C^{\dagger} \otimes| i\rangle\langle j| =\sum_{i, j=1}^{d, d} A|i\rangle\langle j\left|C^{\dagger} \otimes B^T\right| i\rangle\langle j| D^*.
		\end{aligned}
	\end{equation}
\end{proposition}
\begin{proof}
	For any matrix $ B\in \mathbb{C}^{d\times d} $, using Eq.~\eqref{property2}, we have
	\begin{equation}
		\begin{aligned}
			&(B \otimes I_{d}) \sum_{j=1}^d|j\rangle \otimes|j\rangle=(B \otimes I_{d}) \operatorname{vec}\left(I_d\right)= \operatorname{vec}\left(B^T\right)=(I_{d} \otimes B^{T}) \operatorname{vec}\left(I_d\right)
			=\left(I_{d} \otimes B^T\right) \sum_{j=1}^d|j\rangle \otimes|j\rangle.
		\end{aligned}
	\end{equation}
	Therefore, for any matrix $ A\in \mathbb{C}^{d\times d} $, we have
	\begin{equation}
		\begin{aligned}
			(A B \otimes I_{d}) \sum_{j=1}^d|j\rangle \otimes|j\rangle=&(A \otimes I_{d})(B \otimes I_{d}) \sum_{j=1}^d|j\rangle \otimes|j\rangle \\
			=&(A \otimes I_{d})\left(I_{d} \otimes B^T\right) \sum_{j=1}^d|j\rangle \otimes|j\rangle\\
			=&(A \otimes B^T) \sum_{j=1}^d|j\rangle \otimes|j\rangle.
		\end{aligned}
	\end{equation}
	Similarly, for any matrices  $ C\in \mathbb{C}^{d\times d},  D\in \mathbb{C}^{d\times d} $, we have
	\begin{equation}
		\begin{aligned}
			&(CD \otimes I_{d}) \sum_{j=1}^d|j\rangle \otimes|j\rangle=\left(C \otimes D^T\right) \sum_{j=1}^d|j\rangle \otimes|j\rangle.
		\end{aligned}
	\end{equation}
	Therefore, 
	\begin{equation}
		\begin{aligned}
			\sum_{i, j=1}^{d, d} A B|i\rangle\langle j|D^{\dagger} C^{\dagger} \otimes| i\rangle\langle j| &=(A B \otimes I_{d}) \sum_{i, j=1}^{d, d}|i\rangle\langle j|\otimes| i\rangle\langle j|\left(D^{\dagger} C^{\dagger} \otimes I_{d}\right) \\
			&=\left((A B \otimes I_{d}) \sum_{i=1}^d|i\rangle \otimes|i\rangle\right)\left(\sum_{j=1}^d\langle j| \otimes\langle j|\left(D^{\dagger} C^{\dagger} \otimes I_{d}\right)\right) \\
			&=\left(\left(A \otimes B^T\right) \sum_{i=1}^d|i\rangle \otimes|i\rangle\right)\left(\sum_{j=1}^d\langle j| \otimes\langle j|\left(C^{\dagger} \otimes D^*\right)\right) \\
			&=\left(A \otimes B^T\right) \sum_{i, j=1}^{d, d}|i\rangle\langle j|\otimes| i\rangle\langle j|\left(C^{\dagger} \otimes D^*\right) \\
			&=\sum_{i, j=1}^{d, d} A|i\rangle\langle j|C^{\dagger} \otimes B^T| i\rangle\langle j| D^*.
		\end{aligned}
	\end{equation}
\end{proof}
Now, letting $ D=B$, $C=A $, we have
\begin{equation}\label{sp1}
	\begin{aligned}
		&\sum_{i, j=1}^{d, d} A B|i\rangle\langle j|B^{\dagger} A^{\dagger} \otimes| i\rangle\langle j| =\sum_{i, j=1}^{d, d} A|i\rangle\langle j|A^{\dagger} \otimes B^T| i\rangle\langle j| B^*.
	\end{aligned}
\end{equation}

\section{On the optimal scaling of infidelity}\label{Sec4}
In this section, we modify the definition of the fidelity/infidelity for two arbitrary positive semidefinite operators to avoid distortion, and propose a sufficient and necessary  condition for tomography methods to reach the optimal infidelity scaling $ O(1/N) $.
\subsection{Definition of fidelity/infidelity in quantum tomography}\label{defs}

The fidelity between two arbitrary states $ \sigma $ and $ \rho $ is defined as
\begin{equation}\label{deinf}
	F_s(\sigma,\rho)\triangleq\left[\operatorname{Tr} \sqrt{\sqrt{\sigma} \rho \sqrt{\sigma}}\right]^{2},
\end{equation}
which has three basic properties:
\begin{eqnarray}
	(i)&  F_s(\sigma,\rho)=F_s(\rho,\sigma);\\
	(ii)&  0\leq F_s(\sigma,\rho) \leq 1;\\
	(iii)&  F_s(\sigma,\rho) = 1 \ \ \Leftrightarrow \ \ \sigma=\rho. \label{property3}
\end{eqnarray}
To perform QDT and QPT, a new fidelity definition between one positive semidefinite  operator $S\in \mathbb{C}^{d\times d}$ and its estimate $ \hat S $ is needed. A natural idea is to normalize them to the form of density operators, which leads to the definition 
\begin{equation}\label{s1}
	\begin{aligned}
		F_{d,p}\left(\hat{S}, S\right)\triangleq&\left[\operatorname{Tr}\sqrt{\sqrt{\hat S} {S} \sqrt{\hat S}}\right]^{2} \Big{/}\left[\operatorname{Tr}\left(S\right) \operatorname{Tr}\left(\hat{S}\right)\right],
	\end{aligned}
\end{equation}
where the subscript $d$ indicates detector and the subscript $p$ indicates  process. 
The definition Eq.~\eqref{s1} has been applied in both   QDT \cite{Lundeen2009,Zhang_2012} and QPT \cite{PhysRevA.82.042307}. However, Ref. \cite{xiao2021optimal} pointed out that this definition may result in a \emph{distortion} in QDT; i.e., in certain circumstances the property \eqref{property3} does not hold for $F_{d,p}(\hat S, S)$. For example, 
\begin{equation}\label{disqdt}
	\begin{aligned}
		F_{d,p}\left(\hat{P}_{1}=\frac{I}{3}, P_{1}=\frac{I}{4}\right)=1,\quad F_{d,p}\left(\hat{P}_{2}=\frac{I}{3}, P_{2}=\frac{I}{4}\right)=1, \quad F_{d,p}\left(\hat{P}_{3}=\frac{I}{3}, P_{3}=\frac{I}{2}\right)=1.
	\end{aligned}
\end{equation}
Although all the fidelities are one in Eq.~\eqref{disqdt}, the estimated and true POVM elements  are not  equal to each other. The condition that distortion exists in QDT is provided in Proposition 3 of \cite{xiao2021optimal}. 

In QPT, Ref. \cite{PhysRevA.82.042307} also normalized quantum processes to quantum states and the corresponding definition is 
\begin{equation}\label{infqpt1}
	F_{p}\left(\hat X,X \right)\triangleq\frac{\left(\operatorname{Tr}\sqrt{\sqrt{\hat{X}} X \sqrt{\hat{X}}}\right)^{2}}{\operatorname{Tr}(\hat{X}) \operatorname{Tr}(X)}.
\end{equation}
However, for a non-trace-preserving process $ X $, the property \eqref{property3} does not hold for $F_p$ in Eq.~\eqref{infqpt1}; as e.g., $F_{p}\left(X, a{X} \right)= F_{p}\left(X, X\right)=1 $, for any $ 0<a\leq 1 $. We also call this phenomenon \emph{distortion} in QPT. Note that for trace-preserving processes, distortion will not happen, considering the Choi-Jamiołkowski isomorphism \cite{PhysRevA.95.012302,CHOI1975285} as in Eq.~\eqref{ghzx}.

Therefore, to avoid distortion,  we propose the following new definition for the fidelity between $\hat{S}$ and $S$:
\begin{equation}\label{s2}
	\begin{aligned}
		F_1\left(\hat{S}, S\right)\triangleq&\frac{\left(\operatorname{Tr}\sqrt{\sqrt{\hat S} {S} \sqrt{\hat S}}\right)^{2}}{ \operatorname{Tr}\left(S\right) \operatorname{Tr}\left(\hat{S}\right)}-\frac{\left(\operatorname{Tr}\left(S-\hat{S}\right)\right)^{2}}{d^{2}},
	\end{aligned}
\end{equation}
and clearly $F_1\left(\hat{S}, S\right)=1$, if and only if $S=\hat S$.  Given that we subtract a new term in Eq.~\eqref{s2}, its lower bound is no longer zero. 
For QDT, we can designate $S=P_i$
and $F_1\left(\hat{S}, S\right)$ takes values in $(\frac{1}{d}-1,1]$ \cite{xiao2021optimal}. 
When it comes to QPT, we can assign $S=X$. Since distortion does not exist in trace-preserving QPT, the nontrivial case is non-trace-preserving QPT. It is clear that  $F_{d,p}(\hat X, X)$ is in the interval $ [0,1] $ and $ \left(\operatorname{Tr}\left(X- \hat X\right)\right)^{2}<d^2 $. Thus, $ 1\geq F_{1}\left(\hat X, X\right) > -1 $. Furthermore, consider for example $ X=\operatorname{diag}(0_{d^2-d},I_{d}) $ and $\hat X=\operatorname{diag}({\tau},0,0, \cdots,0)  $, where $ \tau>0 $. Then, as $ \tau $ tends to zero, $ F_1\left(\hat X,  X\right) $ can be arbitrarily close to $ -1 $. Hence, $ -1 $ is a tight but unattainable lower bound. Here, we assume that the tight lower bound is denoted as $f\triangleq \inf F_1(\hat{S}, S)$. Hence, $f=d^{-1}-1$ in QDT and $f=-1$ in QPT.
To normalize the range of Eq.~\eqref{s2} into the interval $[0,1]$, we define 
\begin{equation}\label{s3}
	F\left(\hat{S}, S\right)\triangleq\frac{1}{1-f}F_1\left(\hat{S}, S\right)-\frac{f}{1-f},
\end{equation}
and the corresponding infidelity is defined as $1-F\left(\hat{S}, S\right)$. Therefore, 
\begin{equation}
	1-F\left(\hat{S}, S\right)=\frac{1}{1-f}\left(1-F_1\left(\hat{S}, S\right)\right).   
\end{equation}
In QST, since $\operatorname{Tr}(\rho)=\operatorname{Tr}(\hat\rho)=1$, we still maintain $ F\left(\hat{\rho}, \rho\right) =  F_s\left(\hat{\rho}, \rho\right)$. To sum up, for quantum states and trace-preserving quantum processes, the proposed infidelity $1-F$ reduces to $1 - F_s$ and $1 - F_p$, respectively, consistent with previous definitions and serving as a proper generalization to QDT and non-trace-preserving QPT.

\subsection{A sufficient  and necessary condition on the optimal scaling of infidelity}
For QST, let the spectral decompositions of the true quantum state $ \rho $ and the estimated quantum state $ \hat\rho $ be
\begin{equation}
	\begin{aligned}
		&\rho=\sum_{i=1}^{d} \lambda_{i}\left|\lambda_{i}\right\rangle\left\langle\lambda_{i}\right|, \quad \hat{\rho}=\sum_{i=1}^{d} \hat{\lambda}_{i}\left|\hat{\lambda}_{i}\right\rangle\left\langle\hat{\lambda}_{i}\right|,
	\end{aligned}
\end{equation}
where $ \lambda_{1}\geq \cdots \geq \lambda_{d}\geq 0$ and $ \hat\lambda_{1}\geq \cdots \geq \hat\lambda_{d}\geq 0$.  Assume the true rank of $ \rho $ is $ r\leq d $.  Define $ \Delta=\hat{\rho}-\rho $, and the Taylor series expansion of the infidelity up to second order is \cite{PhysRevA.98.012339}:
\begin{equation}\label{taylorinf}
	\begin{aligned}
		\mathbb{E}\left(1-F(\hat\rho, {\rho})\right)=&\mathbb{E}\left( \sum_{i=r+1}^{d}\left\langle\lambda_{i}|\Delta| \lambda_{i}\right\rangle\right)+\frac{1}{2}\mathbb{E}\left( \sum_{i, k=1}^{r} \frac{\left|\left\langle\lambda_{i}|\Delta| \lambda_{k}\right\rangle\right|^{2}}{\lambda_{i}+\lambda_{k}}\right) \\
		&-\frac{1}{4}\mathbb{E}\left(\left[\sum_{i=r+1}^{d}\left\langle\lambda_{i}|\Delta| \lambda_{i}\right\rangle\right]^{2}\right)+O\left(\mathbb{E}\|\Delta\|^{3}\right).
	\end{aligned}
\end{equation}
Utilizing the quantum Cram\'{e}r-Rao bound, the optimal  MSE scaling in QST, QDT, and QPT has been established as 
$ O(1/N) $ \cite{Liu_2020}. Furthermore, in QST, the optimal scaling of the infidelity has been rigorously shown to be $ O(1/N) $ \cite{7956181}. When the true state is full-rank, the first-order term $ \sum_{i=r+1}^{d}\left\langle \lambda_{i}|\Delta| \lambda_{i} \right\rangle $ vanishes, causing Eq.~\eqref{taylorinf} to be asymptotically dominated by the second-order term. If further the MSE of a given estimation algorithm (e.g., MLE or LRE) achieves $ O(1/N) $ scaling, i.e., $ \mathbb{E}(\left\|\Delta\right\|^2) = O(1/N) $, the infidelity then reaches the optimal scaling $ O(1/N) $ \cite{PhysRevA.98.012339,7956181,PhysRevLett.111.183601}. However, when the quantum state $\rho$ is rank-deficient, the first-order term dominates, resulting in an infidelity scaling of $O(1/\sqrt{N})$ in general non-adaptive scenarios. Hence, rank-deficient states represent cases where there is room for improvement in the infidelity scaling. For QPT and QDT, the worst-case and optimal  $1-F_{d,p}$ scalings are obviously still $O(1/\sqrt{N})$ and $O(1/N)$, respectively. As for $F$, the added term in Eq.~\eqref{s2} is $ -\left[\operatorname{Tr}\left(P_{i}-\hat{P}_{i}\right)\right]^{2} / d^{2} $ or $-\left[\operatorname{Tr}\left(X- \hat X\right)\right]^{2}/d^2  $. The scalings of these terms are also $ O\left(1/N\right) $ when $ \mathbb{E}(\left\|\Delta\right\|^2)=O(1/N)  $, which thus  do not affect the worst-case scalings and optimal scalings. Therefore, for $\mathbb{E}\left(1-F(\hat{\rho},\rho)\right)  $, $ \mathbb{E}\left(1-F\left(\hat{P}_{i}, P_{i}\right)\right) $ and $\mathbb{E}\left( 1-F\left(\hat X, X\right)\right) $, the worst-case scalings are all $ O(1/\sqrt{N}) $, and the optimal scalings are all $ O(1/N) $.
We thus focus on the rank-deficient scenario in this work and investigate how to improve the general $O(1/\sqrt{N})$ scaling to the optimal result $O(1/N)$. Moreover, under the mild assumption of $O(1/N)$ scaling for MSE, from Lemma \ref{lemma3} we know the added term  $ -\left[\operatorname{Tr}\left(S-\hat S\right)\right]^{2} / d^{2} $ in Eq.~\eqref{s2} scales as $O(1/N)$, and hence $1-F$ maintains the infidelity scaling of $1-F_{s,d,p}$ up to the optimal scaling. Hence, in this paper we focus on studying the index $1-F$, and its scaling behavior can effectively reflect those of $1-F_{s,d,p}$.

Here, we introduce the following proposition.
\begin{proposition}\cite{yu2015useful}\label{theoremyu}
	Let $\Sigma, \hat{\Sigma} \in \mathbb{R}^{p \times p}$ be symmetric, with eigenvalues $\lambda_1 \geq \ldots \geq \lambda_p$ and $\hat{\lambda}_1 \geq$ $\ldots \geq \hat{\lambda}_p$, respectively. Fix $1 \leq r \leq s \leq p$ and assume that $\min \left(\lambda_{r-1}-\lambda_r, \lambda_s-\lambda_{s+1}\right)>0$, where $\lambda_0:=\infty$ and $\lambda_{p+1}:=-\infty$. Let $d:=s-r+1$, and let $V=\left(v_r, v_{r+1}, \ldots, v_s\right) \in \mathbb{R}^{p \times d}$ and $\hat{V}=\left(\hat{v}_r, \hat{v}_{r+1}, \ldots, \hat{v}_s\right) \in \mathbb{R}^{p \times d}$ have orthonormal columns satisfying $\Sigma v_j=\lambda_j v_j$ and $\hat{\Sigma} \hat{v}_j=\hat{\lambda}_j \hat{v}_j$, for $j=r, r+1, \ldots, s$. Let $\|\cdot\|_{\mathrm{op}} $ be the operator norm. Then
	\begin{equation}\label{yu1}
		\|\sin \Theta(\hat{V}, V)\| \leq \frac{2 \min \left(d^{1 / 2}\|\hat{\Sigma}-\Sigma\|_{\mathrm{op}},\|\hat{\Sigma}-\Sigma\|\right)}{\min \left(\lambda_{r-1}-\lambda_r, \lambda_s-\lambda_{s+1}\right)}.
	\end{equation}
\end{proposition}

More specifically, Ref. \cite{yu2015useful} defines
\begin{equation}
	\hat V_1\triangleq\left[\hat v_1, \cdots, \hat v_{r-1}, \hat v_{s+1}, \cdots, \hat v_p\right],
\end{equation}
and
\begin{equation}\label{yu2}
	\|\sin \Theta(\hat{V}, V)\|\triangleq\left\|\hat{V}_1^T V\right\|,
\end{equation}
which characterizes the distance between subspaces spanned by eigenvectors.
The proof of Proposition \ref{theoremyu} can be found in \cite{yu2015useful}, which may be straightforwardly 
extended to $  \mathbb{C}^{p \times p} $. Using Proposition \ref{theoremyu},  we propose a sufficient and necessary condition on achieving  $ O(1/N) $ optimal scaling of infidelity.
\setcounter{theorem}{0}
\begin{theorem}	
	For any unknown positive semidefinite operator $S \in \mathbb C^{d\times d}$ encoded in quantum tomography, denote its spectral decomposition as $$ S=\sum_{j=1}^d \lambda_j\left\lvert\lambda_j\right\rangle\left\langle\lambda_j\right\lvert$$
	where $ \lambda_1 \geq  \cdots \geq \lambda_r >0, \lambda_{r+1}=\cdots=\lambda_d=0 $. From the measurement results of $ N $ state copies, an estimate $ \hat S \geq 0$ is inferred,  with eigenvalues $ \hat{\lambda}_i $ also in non-increasing order.
	The infidelity $ \mathbb{E}(1-F(\hat S, {S}))$ scales as $O\left(1/{N}\right)$  if and only if the following conditions are both satisfied:
	\begin{enumerate}
		\item[C1:] The MSE $ \mathbb{E}(\|   \hat S-{S}\|^2)$  scales as  $ O(1/N) $;
		\item[C2:] The partial sum of the eigenvalues of  $\hat S$  scales as $\mathbb{E}(\sum_{j=r+1}^{d}\hat{\lambda}_j)= O(1/N)  $.
	\end{enumerate}
\end{theorem}

\begin{proof}
	Firstly, we consider a quantum state which is rank-deficient, i.e., $S=\rho$ and $ r<d$. We start from the sufficiency of C1 and C2. Note that since $\hat{S}\geq 0$, the condition C2 is equivalent to that the estimated eigenvalues scale as $ \mathbb{E}(\hat{\lambda}_i)= O(1/N)  $ for $ r+1 \leq j \leq d $.
	Define
	\begin{equation*}
		\begin{aligned}
			U_1&\triangleq\left[|\lambda_{1}\rangle, \cdots,|\lambda_r\rangle\right],U_2\triangleq\left[|\lambda_{r+1}\rangle, \cdots,|\lambda_d\rangle\right],\\
			\hat U_1&\triangleq[|\hat\lambda_{1}\rangle, \cdots,|\hat\lambda_r\rangle],\hat U_2\triangleq[|\hat\lambda_{r+1}\rangle, \cdots,|\hat\lambda_d\rangle],
		\end{aligned}
	\end{equation*}
	and $ W\triangleq U_{2}U_{2}^{\dagger} \geq 0 $. Using Proposition \ref{theoremyu}, we have
	\begin{equation}\label{u1}
		\begin{aligned}
			\|\hat U_2^{\dagger}U_1\|&=\|\sin \Theta(\hat{U}_1, U_1)\| \leq \frac{2 \min \left(d^{1 / 2}\|\hat{\rho}-\rho\|_{\mathrm{op}},\|\hat{\rho}-\rho\|\right)}{ \lambda_{r}},
		\end{aligned}
	\end{equation}
	and
	\begin{equation}\label{u2}
		\begin{aligned}
			\|\hat U_1^{\dagger}U_2\|&=\|\sin \Theta(\hat{U}_2, U_2)\| \leq \frac{2 \min \left(d^{1 / 2}\|\hat{\rho}-\rho\|_{\mathrm{op}},\|\hat{\rho}-\rho\|\right)}{ \lambda_{r}},
		\end{aligned}
	\end{equation}
	where $\|\sin \Theta(\hat{U}_1, U_1)\|$ and $\|\sin \Theta(\hat{U}_2, U_2)\| $ are defined as in Eq.~\eqref{yu2} characterizing the distance between subspaces spanned by the eigenvectors \cite{yu2015useful}.
	If the MSE $ \mathbb{E}\|\hat\rho-{\rho}\|^2=O\left(1/N\right) $, using Eqs.~\eqref{u1} and \eqref{u2}, we have $\mathbb{E}\left(\|\hat U_2^{\dagger}U_1\|^2\right)=O\left(1/N\right)$,  $\mathbb{E}\left(\|\hat U_1^{\dagger}U_2\|^2\right)=O\left(1/N\right)$ and thus
	\begin{equation}\label{orthlam}
		\begin{aligned}
			&\mathbb{E}\left(\left|\left\langle\lambda_{i} \mid \hat{\lambda}_{j}\right\rangle\right|^{2}\right)=O(1 / N), \quad \forall i,j \text{ that satisfy } 1\leq i \leq r,\; r+1\leq j \leq d,\\
			&\mathbb{E}\left(\left|\left\langle\lambda_{i} \mid \hat{\lambda}_{j}\right\rangle\right|^{2}\right)=O(1 / N), \quad \forall i,j \text{ that satisfy } r+1\leq i \leq d,\; 1\leq j \leq r.
		\end{aligned}
	\end{equation}
	Therefore, 
	\begin{equation}\label{w1}
		\mathbb{E}\left(\left\langle\hat{\lambda}_j|W| \hat{\lambda}_j\right\rangle\right)=\mathbb{E}\left(\left\langle\hat{\lambda}_j|U_{2}U_{2}^{\dagger}| \hat{\lambda}_j\right\rangle\right)=O\left(\frac{1}{N}\right)
	\end{equation}
	for $ 1\leq j \leq r $. Moreover,
	\begin{equation}
		\mathbb{E}\left(\left\langle\hat{\lambda}_j|U_{1}U_{1}^{\dagger}| \hat{\lambda}_j\right\rangle\right)=\mathbb{E}\left(\left\langle\hat{\lambda}_j|I_{d}-W| \hat{\lambda}_j\right\rangle\right)=1+O\left(\frac{1}{N}\right)
	\end{equation}
	for $ 1\leq j \leq r $.
	Similarly,
	\begin{equation}\label{w2}
		\mathbb{E}\left(\left\langle\hat{\lambda}_j|U_2 U_{2}^{\dagger}| \hat{\lambda}_j\right\rangle\right)=\mathbb{E}\left(\left\langle\hat{\lambda}_j|W| \hat{\lambda}_j\right\rangle\right)= 1+O\left(\frac{1}{N}\right)
	\end{equation}
	for $r+1 \leq j \leq d$.

	Since the MSE scales as $ O(1/N) $, using Lemma \ref{lemma3}, we have $ \mathbb{E}\left(\hat\lambda_j-\lambda_j\right)^2=O(1/N) $ for $ 1\leq j \leq r $. For $ r+1\leq j\leq d $, if  $ \mathbb{E}(\hat{\lambda}_j)=O(1/N)  $, the first-order term in Eq.~\eqref{taylorinf} scales as
	\begin{equation}\label{firstterm}
		\begin{aligned}
			\left|\mathbb{E}\left(\sum_{i=r+1}^d\left\langle\lambda_i|\Delta| \lambda_i\right\rangle\right)\right|&=\left|\mathbb{E}\left(\sum_{i=r+1}^d\left\langle\lambda_i|\hat{\rho}| \lambda_i\right\rangle\right)-\mathbb{E}\left(\sum_{i=r+1}^d\left\langle\lambda_i|\rho| \lambda_i\right\rangle\right)\right| \\
			&= \mathbb{E}\left(\operatorname{Tr}(\hat{\rho} W)\right)\\
			&=\mathbb{E}\left(\sum_{j=1}^r \hat{\lambda}_j\left\langle\hat{\lambda}_j|W| \hat{\lambda}_j\right\rangle\right)+\mathbb{E}\left(\sum_{j=r+1}^d \hat{\lambda}_j\left\langle\hat{\lambda}_j|W| \hat{\lambda}_j\right\rangle\right) \\
			&\leq\mathbb{E}\left(\sum_{j=1}^r \left\langle\hat{\lambda}_j|W| \hat{\lambda}_j\right\rangle\right)+\mathbb{E}\left(\sum_{j=r+1}^d \hat{\lambda}_j\right)\\
			&= \sum_{j=1}^r O\left(\frac{1}{N}\right)+\sum_{j=r+1}^d O\left(\frac{1}{N}\right) \\
			&= O\left(\frac{1}{N}\right).
		\end{aligned}
	\end{equation}
	Therefore, the first-order term also scales as $ O(1/N) $ and thus the infidelity $\mathbb{E}\left(1-F(\hat{\rho},\rho)\right)  $ has the optimal scaling $ O(1/N) $. 
	
	Now we come to the necessity of C1 and C2.
	From \cite{761271}, the Fuchs–van de Graaf inequalities are
	\begin{equation}\label{tr2}
		1-\sqrt{F(\hat\rho, \rho)} \leq \frac{1}{2}\|\hat\rho-\rho\|_{\text{tr}}\leq \sqrt{1-F(\hat\rho, \rho)}
	\end{equation}
	where $ \|\cdot\|_{\text{tr}} $ is the  trace norm.
	If $\mathbb{E}\left( 1-F(\hat\rho, \rho)\right)=O(1/N) $, we have $\mathbb{E}\|\hat\rho-\rho\|_{\text{tr}}^2 = O(1/N) $. From \cite{horn_johnson_2012},
	since 
	\begin{equation}\label{tr1}
		\frac{1}{r}\|\hat{\rho}-\rho\|_{\mathrm{tr}}^{2} \leq\|\hat{\rho}-\rho\|^{2 }\leq\|\hat{\rho}-\rho\|_{\mathrm{tr}}^{2},
	\end{equation}
	the MSE $ \mathbb{E}\|\hat\rho-{\rho}\|^2 $ scales as $O(1/N)  $, and both Eqs.~\eqref{w1} and \eqref{w2} hold. Therefore, the two second-order terms in Eq.~\eqref{taylorinf} also scale as $ O(1/N) $. 
	For the first-order term  Eq.~\eqref{firstterm}, each element in the third equality is non-negative and the first term $ \sum_{j=1}^r \hat{\lambda}_j\left\langle\hat{\lambda}_j|W| \hat{\lambda}_j\right\rangle $ scales as $O(1/N)  $ because the MSE scales as $O(1/N)  $.
	Since $ \mathbb{E}(1-F(\hat\rho, \rho))=O(1/N) $, the second term $ \sum_{j=r+1}^d \hat{\lambda}_j\left\langle\hat{\lambda}_j|W| \hat{\lambda}_j\right\rangle $ also scales as $O(1/N)  $.
	Because of $ \mathbb{E}(\langle\hat{\lambda}_j|W| \hat{\lambda}_j\rangle)=1+O\left(1/N\right) $ for $ r+1\leq j \leq d $, 
	we must have  $ \mathbb{E}(\hat{\lambda}_j)=O(1/N) $ for $ r+1 \leq j \leq d $.
	Therefore, if the infidelity scales as $O(1/N)$, the MSE $ \mathbb{E}(\|\hat\rho-{\rho}\|^2) $ scales as $O\left(1/N\right)$ and  $ \mathbb{E}(\hat{\lambda}_j)=O(1/N) $ for $ r+1 \leq j \leq d $.
	
	Then we consider the case that $\rho$ has full rank; i.e., $r=d$.
	If $ \mathbb{E}(\|   \hat \rho-{\rho}\|^2)= O(1/N) $, using Eq.~\eqref{taylorinf},   we have $ \mathbb{E}(1-F(\hat \rho, {\rho}))=O(1/N)$. Since  $ \mathbb{E}\left(\sum_{j=r+1}^{d}\hat{\lambda}_j\right)=0  $, condition C2 still holds. If $ \mathbb{E}(1-F(\hat \rho, {\rho}))=O(1/N)$, using Eqs.~\eqref{tr2} and \eqref{tr1}, we have $ \mathbb{E}(\|\hat \rho-{\rho}\|^2)= O(1/N) $ and condition C2 also holds if $\rho$ has full rank.
	
	We then focus on a general positive semidefinite matrix $S$ which is rank-deficient; i.e., $r<d$.
	For the sufficiency of C1 and C2, 
	if the MSE scales as $\mathbb{E}(\left\|\hat{S}-S\right\|^2)=O\left({1}/{N}\right)$, we have
	\begin{equation}\label{qdt3}
		\mathbb{E}\left(\operatorname{Tr}\left(\hat{S}\right)-\operatorname{Tr}\left(S\right)\right)^2=O\left(\frac{1}{N}\right),
	\end{equation}
	and $\mathbb{E} \operatorname{Tr}\left(\hat{S}\right)\rightarrow \operatorname{Tr}\left(S\right) $.
	Thus, 
	\begin{equation}
		\begin{aligned}
			\left\|\frac{\hat{S}}{\operatorname{Tr}\left(\hat{S}\right)}-\frac{S}{\operatorname{Tr}\left(S\right)}\right\|=&\left\| \frac{\hat{S} \operatorname{Tr}\left(S\right)-S \operatorname{Tr}\left(\hat{S}\right)}{\operatorname{Tr}\left(\hat{S}\right) \operatorname{Tr}\left(S\right)} \right\|\\
			=&\left\| \frac{\hat{S} \operatorname{Tr}\left(S\right)-S \operatorname{Tr}\left(S\right)+S \operatorname{Tr}\left(S\right)-S \operatorname{Tr}\left(\hat{S}\right)}{\operatorname{Tr}\left(\hat{S}\right) \operatorname{Tr}\left(S\right)} \right\|\\
			\leq& \frac{1}{\operatorname{Tr}\left(\hat{S}\right)}\left\|\hat{S}-S\right\|+\frac{\left|\operatorname{Tr}\left(S\right)-\operatorname{Tr}\left(\hat{S}\right)\right|\left\|S\right\|}{\operatorname{Tr}\left(\hat{S}\right) \operatorname{Tr}\left(S\right)}.
		\end{aligned}
	\end{equation}
	Therefore, using the Cauchy–Schwarz inequality, it follows that
	\begin{equation}\label{qdt5}
		\begin{aligned}
			& \mathbb{E}\left(\left\|\frac{\hat{S}}{\operatorname{Tr}\left(\hat{S}\right)}-\frac{S}{\operatorname{Tr}\left(S\right)}\right\|^2\right)\\
			\leq& \mathbb{E}\left(\frac{\left\|\hat{S}-S\right\|}{\operatorname{Tr}\left(\hat{S}\right)}\right)^2+\left(\frac{\left\|S\right\|}{\operatorname{Tr}\left(S\right)}\right)^2 \mathbb{E}\left(\frac{\operatorname{Tr}\left(S\right)-\operatorname{Tr}\left(\hat{S}\right)}{\operatorname{Tr}\left(\hat{S}\right)}\right)^2+2 \frac{\left\|S\right\|}{\operatorname{Tr}\left(S\right)} \mathbb{E}\left(\frac{\left\|\hat{S}-S\right\|}{\operatorname{Tr}\left(\hat{S}\right)} \cdot \frac{\operatorname{Tr}\left(S\right)-\operatorname{Tr}\left(\hat{S}\right)}{\operatorname{Tr}\left(\hat{S}\right)}\right)\\
			\leq & \mathbb{E}\left(\frac{\left\|\hat{S}-S\right\|}{\operatorname{Tr}\left(\hat{S}\right)}\right)^2+\left(\frac{\left\|S\right\|}{\operatorname{Tr}\left(S\right)}\right)^2 \mathbb{E}\left(\frac{\operatorname{Tr}\left(S\right)-\operatorname{Tr}\left(\hat{S}\right)}{\operatorname{Tr}\left(\hat{S}\right)}\right)^2+ 2 \frac{\left\|S\right\|}{\operatorname{Tr}\left(S\right)} \sqrt{\mathbb{E}\left(\frac{\left\|\hat{S}-S\right\|}{\operatorname{Tr}\left(\hat{S}\right)}\right)^2 {\mathbb{E}}\left(\frac{\operatorname{Tr}\left(S\right)-\operatorname{Tr}\left(\hat{S}\right)}{\operatorname{Tr}\left(\hat{S}\right)}\right)^2}\\
			=&O\left(\frac{1}{N}\right).
		\end{aligned}
	\end{equation}
	Since the eigenvalues of $ \hat{S} $ scale as $ \mathbb{E}(\hat{\lambda}_j)= O(1/N)  $ for $ r+1 \leq j \leq d $ and $\mathbb{E} \operatorname{Tr}\left(\hat{S}\right)\rightarrow \operatorname{Tr}\left(S\right) $ from Eq.~\eqref{qdt3}, the eigenvalues of $\hat{S}/\operatorname{Tr}(\hat{S})  $ also scale as $  O(1/N)  $ for $ r+1 \leq j \leq d $.
	After normalization, $ \hat{S}/\operatorname{Tr}(\hat{S}) $ is the same as a quantum state and the added term in the new fidelity Eq.~\eqref{s2} also  scales as $ O\left(1/N\right) $ from Eq.~\eqref{qdt3}. 
	Thus  the remaining analysis is similar to the QST case and the infidelity $\mathbb{E}\left( 1-F\left(\hat{S}, S\right)\right) $ scales as $ O\left(1/N\right) $.
	
	For the necessity part,
	if $ \mathbb{E}\left(1-F\left(\hat{S}, S\right)\right)=O\left(1/N\right) $, we have  $ \mathbb{E}\left(1-F_1\left(\hat{S}, S\right)\right)=O\left(1/N\right) $.
	Noticing that $$\mathbb{E}\left(1-F_1\left(\hat{S}, S\right)\right)= \mathbb{E}\left(1-F_{d,p}\left(\hat{S}, S\right)\right)+\mathbb{E}\left( \left[\operatorname{Tr}\left(S-\hat{S}\right)\right]^{2} / d^{2}\right),$$
	we thus have $ \mathbb{E}\left(1-F_{d,p}\left(\hat{S}, S\right)\right)=O\left(1/N\right) $ and
	\begin{equation}\label{qdt1}
		\mathbb{E} \left(\left[\operatorname{Tr}\left(S-\hat{S}\right)\right]^{2} / d^{2}\right)=O\left(\frac{1}{N}\right).
	\end{equation}
	Since we normalize $\hat{S}$ and $S$ to a quantum state in $1- F_{d,p}\left(\hat{S}, S\right) $,  similarly to the QST case, we have
	\begin{equation}\label{qdt2}
		\mathbb{E}\left(\left\|\frac{\hat{S}}{\operatorname{Tr}\left(\hat{S}\right)}-\frac{S}{\operatorname{Tr}\left(S\right)}\right\|^2\right)=O\left(\frac{1}{N}\right),
	\end{equation}
	and the eigenvalues of  $ \hat{S}/\operatorname{Tr}(\hat{S}) $ scale as $\mathbb{E}\left(\hat{\lambda}_j / \operatorname{Tr}\left(\hat{S}\right)\right)=O\left(1/N\right)  $ for $ r+1 \leq j \leq d $.
	Using Eq.~\eqref{qdt1}, we have $\mathbb{E} \operatorname{Tr}\left(\hat{S}\right)\rightarrow \operatorname{Tr}\left(S\right) $. Therefore,  $\mathbb{E}(\hat{\lambda}_j) =O\left(1/N\right)  $ for $ r+1 \leq j \leq d $. For $ \left\|\hat{S}-S\right\| $, we have
	\begin{equation}\label{qdt4}
		\begin{aligned}
			\left\|\hat{S}-S\right\|=&\operatorname{Tr}\left(S\right)\left\|\frac{\hat{S}}{\operatorname{Tr}\left(S\right)}-\frac{\hat{S}}{\operatorname{Tr}\left(\hat{S}\right)}+\frac{\hat{S}}{\operatorname{Tr}\left(\hat{S}\right)}-\frac{S}{\operatorname{Tr}\left(S\right) \|}\right\| \\
			\leq& \operatorname{Tr}\left(S\right)\left\|\frac{\hat{S}}{\operatorname{Tr}\left(S\right)}-\frac{\hat{S}}{\operatorname{Tr}\left(\hat{S}\right)}\right\|+\operatorname{Tr}\left(S\right)\left\|\frac{\hat{S}}{\operatorname{Tr}\left(\hat{S}\right)}-\frac{S}{\operatorname{Tr}\left(S\right)}\right\| \\
			=&\operatorname{Tr}\left(S\right)\left\|\hat{S}\right\|\left\|\frac{\operatorname{Tr}\left(\hat{S}\right)-\operatorname{Tr}\left(S\right)}{\operatorname{Tr}\left(S\right) \operatorname{Tr}\left(\hat{S}\right)}\right\|+\operatorname{Tr}\left(S\right)\left\|\frac{\hat{S}}{\operatorname{Tr}\left(\hat{S}\right)}-\frac{S}{\operatorname{Tr}\left(S\right)}\right\|.
		\end{aligned}
	\end{equation}
	Then using  Eqs.~\eqref{qdt1} and \eqref{qdt2}, and the Cauchy–Schwarz inequity, it follows that 
	the MSE  scales as  $\mathbb{E}\left\|\hat{S}-S\right\|^2=O(1/N)$.
	
	Finally, we discuss the case in which $S$ is of full rank, i.e., $r=d$. If $ \mathbb{E}(\| \hat S-{S}\|^2)= O(1/N) $ using Eqs.~\eqref{taylorinf} and \eqref{qdt3}--\eqref{qdt5},   we still have $ \mathbb{E}(1-F(\hat S, {S}))=O(1/N)$ which is similar to QST.  Since  $ \mathbb{E}\left(\sum_{j=r+1}^{d}\hat{\lambda}_j\right)=0  $, condition C2 still holds. If $ \mathbb{E}\left(1-F\left(\hat{S}, S\right)\right)=O\left(1/N\right) $, using Eqs.~\eqref{qdt1}--\eqref{qdt4}, we still have  $\mathbb{E}\left\|\hat{S}-S\right\|^2=O(1/N)$ and condition C2 also holds if $S$ is of full rank.
\end{proof}

In QST and trace-preserving QPT, since $\operatorname{Tr}(\rho)=1$ and $\operatorname{Tr}(X)=d$, respectively, distortion does not exist for these two scenarios and thus Theorem \ref{theorem1} still holds.
For QDT and non-trace-preserving QPT, if we use the direct infidelity definition $ 1-F_{d,p}\left(\hat{P}_{i}, P_{i}\right) $ and $ 1-F_{d,p}\left(\hat X,X \right) $, the sufficiency of C1 and C2 in Theorem \ref{theorem1} still holds but the necessity fails due to distortion. For example, even if $ \mathbb{E}\left(1-F_{d,p}\left(\hat P_{i},P_{i} \right)\right)=O(1/N) $, since Eq.~\eqref{qdt1} does not hold anymore, we cannot conclude that the MSE $ \mathbb{E}\left\|\hat{P}_{i}-P_{i}\right\|^2= O(1/N) $. This situation also occurs when we consider $ 1-F_{d,p}\left(\hat X,X \right) $ in non-trace-preserving QPT.
Therefore, we need to modify the definition of infidelity in QDT and QPT as in Eqs.~\eqref{s2} and \eqref{s3}.
Noting that the direct definition of infidelity $ \mathbb{E}\left(1-F_{d,p}\left(\hat{P}_{i}, P_{i}\right)\right)$ or $ \mathbb{E}\left(1-F_{d,p}\left(\hat X,X \right)\right) $ and the added term $\left[\operatorname{Tr}\left(P_{i}-\hat{P}_{i}\right)\right]^{2} / d^{2} \text{ or } \left[\operatorname{Tr}\left(X- \hat X\right)\right]^{2}/d^2  $ are both non-negative for QDT and QPT, if the modified infidelity $ \mathbb{E}\left(1-F\left(\hat{P}_{i}, P_{i}\right)\right)$ or $\mathbb{E}\left(1-F\left(\hat X,X \right)\right) $ scales as $O\left(1/{N}\right)$,  $ \mathbb{E}\left(1-F_{d,p}\left(\hat{P}_{i}, P_{i}\right)\right)$ or $\mathbb{E}\left(1-F_{d,p}\left(\hat X,X \right)\right) $ also has the optimal scaling $O\left(1/{N}\right)$. 
Furthermore, we highlight that Theorem \ref{theorem1} can be applied to achieve optimal infidelity scaling between any positive semidefinite physical quantity $S$
including but not limited to quantum states, POVM elements, and process matrices, and corresponding estimate $ \hat{S} \geq 0 $ using any reconstruction algorithm. For instance, it can also be applied to quantum assemblage tomography \cite{pepper2024scalable,VillegasAguilar2024}.

Condition C2 ($\mathbb{E}(\sum_{j=r+1}^{d}\hat{\lambda}_j)= O(1/N)$), although not immediately obvious how to achieve, is in fact easier to analyze and achieve than the optimal infidelity scaling. One empirical approach to satisfy C2, as indicated in the main text, is to measure the unknown $S$ along its (estimated) eigenvectors. This is because the variance of the error $e$ between a measured outcome frequency $\hat p$ (such as $\hat p_i$ in Algorithm 1 in the main text) and its true value $p$ is asymptotically
\begin{equation}\label{equadd2}
	\text{Var}(e)=\frac{p-p^2}{N},
\end{equation}
which typically scales as $O(1/N)$. Measuring (almost) along the eigenvectors of $S$ can make the numerator of \eqref{equadd2} approach zero, further reducing the scaling of $\text{Var}(e)$ to ultimately outperform $O(1/N)$. This improvement of the variance scaling can occur for those outcome frequencies close to zero, which correspond to the zero or one eigenvalues of $S$, and is therefore beneficial for accurately estimating these zero eigenvalues of $S$. This motivates our adaptive algorithms and intuitively explains why measuring along the eigenvectors is helpful in estimating the zero or small eigenvalues. This phenomenon has been explicitly analyzed in the qubit setting in \cite{PhysRevLett.111.183601}.

\subsection{Scaling analysis of other related indices}
The infidelity definition for QDT is not unique and Ref. \cite{Hou2018} gives a new fidelity/infidelity definition $F_H$ for a detector.  
Consider two detectors $\left\{P_j\right\}_{j=1}^n$ and $\left\{\hat P_{j}\right\}_{j=1}^n$ on a $d$-dimensional Hilbert space with the same number of elements. Construct two normalized quantum states as $\sigma=\frac{1}{d} \sum_{j=1}^n \left(P_j \otimes|j\rangle\langle j|\right)$ and $\hat \sigma=\frac{1}{d} \sum_{j=1}^n \left(\hat P_{j} \otimes|j\rangle\langle j|\right)$, where the $\{|j\rangle\}_{i=1}^{n}$ form an orthonormal basis for an ancilla system. The fidelity between the two detectors $\left\{P_j\right\}_{j=1}^n$ and $\left\{\hat P_{j}\right\}_{j=1}^n$ is defined as the fidelity between the two states $\sigma$ and $\hat\sigma$,
\begin{equation}
	F_H\left(\left\{P_j\right\}_{j=1}^n, \left\{\hat P_{j}\right\}_{j=1}^n \right)\triangleq F\left(\sigma, \hat\sigma\right)=\left(\operatorname{Tr} \sqrt{\sqrt{\sigma} \hat\sigma \sqrt{\sigma}}\right)^2.
\end{equation}
A sufficient and necessary condition similar to our Theorem \ref{theorem1} can still be obtained to characterize that infidelity.
\begin{theorem}
	The infidelity $1-F_H\left(\left\{P_j\right\}_{j=1}^n, \left\{\hat P_{j}\right\}_{j=1}^n \right)$ scales as $ O(1/N) $ if and only if the MSE of each POVM element $\mathbb{E}(\|\hat P_j-P_j\|^2)$ scales as  $ O(1/N) $ and the  estimated null eigenvalues  of each POVM element scale as  $ O(1/N) $.
\end{theorem}
\begin{proof}
	According to \cite{LOAN200085}, if $A$ and $B$ are square matrices, then ${A} \otimes {B}$ and ${B} \otimes {A}$ are  permutations similar  as
	\begin{equation}
		{A} \otimes {B}=Q\left({B} \otimes {A} \right)Q^{T},
	\end{equation}
	where $ Q $ is a permutation matrix. Thus, $Q\left( P_{j} \otimes|j\rangle\langle j|\right)Q^{T}=|j\rangle\langle j|\otimes P_j  $.
	We can define $ \tau=\frac{1}{d} \sum_{j=1}^n (|j\rangle\langle j|$ $\otimes P_j)  $ and $ \hat\tau=\frac{1}{d} \sum_{j=1}^n \left(|j\rangle\langle j|\otimes \hat P_j \right) $, and thus $  F_s\left(\sigma, \hat\sigma\right)= F_s\left(\tau, \hat\tau\right) $ and $\|\hat{\tau}-\tau\|=\|\hat{\sigma}-\sigma\|$.
	Since $ \tau\!=\!\frac{1}{d}\operatorname{diag}(P_1, \! \cdots,\! P_n) $,
	the eigenvalues of $ d\times\hat \tau $ correspond to the collection of the eigenvalues of $\left\{\hat P_{j}\right\}_{j=1}^n$.
	
	It is clear that the MSE $\mathbb{E}(\|\hat\tau-{\tau}\|^2)$ scales as $ O(1/N) $ if and only if the MSE of each POVM element $\mathbb{E}(\|\hat P_j-P_j\|^2)$ scales as  $ O(1/N) $ for $ 1\leq j \leq n $. The estimated null eigenvalues of $ \hat \tau $ scale as $ O(1/N) $ if and only if the estimated null eigenvalues of each POVM element scale as  $ O(1/N) $. Therefore, using Theorem \ref{theorem1}, the infidelity $1-F_s\left(\sigma, \hat\sigma\right)$ scales as $ O(1/N) $ if and only if the MSE of each POVM element $\mathbb{E}(\|\hat P_j-P_j\|^2)$ scales as  $ O(1/N) $ and the estimated null eigenvalues of each POVM element scale as  $ O(1/N) $.
\end{proof}

In QPT, Ref.~\cite{PhysRevA.95.012302} extends the Bures distance $d_{B}(X, \hat{X})$ to a non-trace-preserving process $X$ and its estimate $\hat{X}$ as follows
\begin{equation}\label{dbref}
	d_{B}^{2}( \hat{X},X) = \operatorname{Tr}(\hat{X})+\operatorname{Tr}(X)   - 2 \operatorname{Tr}\left( \sqrt{\sqrt{\hat{X}} {X} \sqrt{\hat{X}}} \right).
\end{equation}

We introduce the following theorem to establish that our proposed infidelity metric and the squared Bures distance share the same optimal scaling behavior.
\begin{theorem}\label{addth3}
	The infidelity $1-F( \hat{X},X)$ achieves the optimal scaling $O(1/N)$ if and only if the squared Bures distance $d_{B}^{2}( \hat{X},X)$ scales as $O(1/N)$.
\end{theorem}

\begin{proof}
	Using~\eqref{dbref}, we have
	\begin{equation}\label{db1}
		\operatorname{Tr}\left(\sqrt{\sqrt{\hat{X}}{X}\sqrt{\hat{X}}}\right) = \frac{1}{2}\left( \operatorname{Tr}(X) + \operatorname{Tr}(\hat{X}) - d_B^2( \hat{X},X) \right).
	\end{equation}
	Let
	\begin{equation}
		b=d_B^2(\hat X,X),\;    
		t=\operatorname{Tr}(X),\;
		\delta t=\operatorname{Tr}(\hat X)-\operatorname{Tr}(X).
	\end{equation}
	Substituting Eq.~\eqref{db1} into Eq.~\eqref{s2} gives the exact identity
	\begin{equation}
		F_1(\hat X,X)
		=
		\frac{
			\left(
			2t+\delta t-b
			\right)^2
		}{
			4t(t+\delta t)
		}
		-
		\frac{(\delta t)^2}{d^2}.
	\end{equation}
	Using Eq.~\eqref{s3} and $f=-1$, we obtain
	\begin{equation}\label{exact_identity}
		\begin{aligned}
			&\quad 1-F(\hat X,X)
			=
			\frac{1}{2}
			\left(
			1-F_1(\hat X,X)
			\right)  =
			\frac{(2t+\delta t)b}{4t(t+\delta t)}
			-
			\frac{b^2}{8t(t+\delta t)}
			-
			\frac{(\delta t)^2}{8t(t+\delta t)}
			+
			\frac{(\delta t)^2}{2d^2}.
		\end{aligned}
	\end{equation}

	We now compare the asymptotic scaling of the unified infidelity and the squared Bures distance. First, suppose that
	\begin{equation}
		\mathbb{E}( d_B^2(\hat X,X))=O\left(\frac{1}{N}\right).
	\end{equation}
	For any two positive semidefinite matrices $X$ and $\hat{X}$, the trace norm satisfies the inequality~\cite[Corollary IV.2.6]{bhatia1997matrix}
	\begin{equation}
		\|\sqrt{X} \sqrt{\hat{X}} \|_{\operatorname{tr}} \leq \| \sqrt{\hat{X}} \| \| \sqrt{X} \|. 
	\end{equation}
	Therefore,
	\begin{equation}
		\operatorname{Tr}
		\sqrt{\sqrt {\hat X} X \sqrt {\hat X}}
		\leq
		\sqrt{\operatorname{Tr}(X)\operatorname{Tr}(\hat X)}.
	\end{equation}
	Substituting this bound into the definition of the squared Bures distance gives
	\begin{equation}
		\begin{aligned}
			d_B^2(\hat X,X)
			&\geq
			\operatorname{Tr}(X)
			+
			\operatorname{Tr}(\hat X)
			-
			2\sqrt{\operatorname{Tr}(X)\operatorname{Tr}(\hat X)}
			=
			\left(
			\sqrt{\operatorname{Tr}(X)}
			-
			\sqrt{\operatorname{Tr}(\hat X)}
			\right)^2 .
		\end{aligned}
	\end{equation}
	Since
	\begin{equation}
		\begin{aligned}
			&\quad \left(
			\operatorname{Tr}(X)-\operatorname{Tr}(\hat X)
			\right)^2=
			\left(
			\sqrt{\operatorname{Tr}(X)}
			-
			\sqrt{\operatorname{Tr}(\hat X)}
			\right)^2
			\left(
			\sqrt{\operatorname{Tr}(X)}
			+
			\sqrt{\operatorname{Tr}(\hat X)}
			\right)^2 ,
		\end{aligned}
	\end{equation}
	we obtain
	\begin{equation}
		\left(
		\operatorname{Tr}(X)-\operatorname{Tr}(\hat X)
		\right)^2
		\leq
		d_B^2(\hat X,X)
		\left(
		\sqrt{\operatorname{Tr}(X)}
		+
		\sqrt{\operatorname{Tr}(\hat X)}
		\right)^2 .
	\end{equation}
	As $X$ is fixed and $\operatorname{Tr}(X)>0$, we have
	\begin{equation}\label{dbn1}
		\mathbb{E}[(
		\operatorname{Tr}(X)-\operatorname{Tr}(\hat X)
		)^2]
		=O\left(\frac{1}{N}\right).
	\end{equation}
	Substituting \eqref{dbn1} and $\mathbb{E}(b)=O(1/N)$ into Eq.~\eqref{exact_identity}, we have
	\begin{equation}
		\mathbb{E}(1-F(\hat X,X))=O\left(\frac{1}{N}\right).
	\end{equation}
	
	Conversely, if
	\begin{equation}
		\mathbb{E}(1-F(\hat X,X))=O\left(\frac{1}{N}\right),
	\end{equation}
	we have already proven that 
	\begin{equation}
		\mathbb{E}((\delta t)^2)=\mathbb{E}[
		(\operatorname{Tr}(X)-\operatorname{Tr}(\hat X)
		)^2]
		=O\left(\frac{1}{N}\right)
	\end{equation}
	in Eq.~\eqref{qdt1}.
	Eq.~\eqref{exact_identity} can be rewritten as
	\begin{equation}
		\begin{aligned}
			1-F(\hat X,X)=
			\left[
			\frac{2t+\delta t}{4t(t+\delta t)}
			-
			\frac{b}{8t(t+\delta t)}
			\right]b
			+
			\left[
			\frac{1}{2d^2}
			-
			\frac{1}{8t(t+\delta t)}
			\right](\delta t)^2 .
		\end{aligned}
	\end{equation}
	Since
	\begin{equation}
		\begin{aligned}
			b=
			\operatorname{Tr}(X)
			+
			\operatorname{Tr}(\hat X)
			-
			2\operatorname{Tr}\left(\sqrt{\sqrt{\hat{X}}{X}\sqrt{\hat{X}}}\right)
			\leq
			\operatorname{Tr}(X)
			+
			\operatorname{Tr}(\hat X)
			=
			2t+\delta t,   
		\end{aligned}
	\end{equation}
	we have
	\begin{equation}
		\frac{2t+\delta t}{4t(t+\delta t)}
		-
		\frac{b}{8t(t+\delta t)}
		\geq
		\frac{2t+\delta t}{8t(t+\delta t)}.
	\end{equation}
	Since $t=\operatorname{Tr}(X)>0$ and using
	\begin{equation}
		\mathbb{E}(1-F(\hat X,X))=O\left(\frac{1}{N}\right),
		\quad
		\mathbb{E}((\delta t)^2)=O\left(\frac{1}{N}\right),
	\end{equation}
	we obtain
	\begin{equation}
		\mathbb{E}(b)=\mathbb{E}(d_B^2(\hat X,X))=O\left(\frac{1}{N}\right).
	\end{equation}
	
	Combining both directions, we conclude that
	\begin{equation}
		\mathbb{E}(1-F(\hat X,X))\!=\!O\left(\frac{1}{N}\right) \!\Longleftrightarrow \mathbb{E}(d_B^2(\hat X,X))=O\left(\frac{1}{N}\right).
	\end{equation}
\end{proof}

Consequently, by combining Theorem~\ref{theorem1} and Theorem~\ref{addth3}, the squared Bures distance $d_{B}^{2}( \hat{X},X)$ scales as $O(1/N)$ if and only if conditions C1 and C2 both hold.

\section{Adaptive quantum tomography}\label{sec6}
Using the sufficient condition of Theorem \ref{theorem1}, we aim to design adaptive algorithms achieving $1-F=O(1/N)$.
In the full-rank scenarios,  optimal infidelity scaling is achieved by the MSE scaling $O(1/N)$, which is trivial.
Thus, in the following, we focus  on  rank-deficient cases.


\subsection{Two-step  adaptive quantum state tomography}\label{sec61}

We first  consider adaptive QST. Without loss of generality, in our protocol we assume the eigenvalues of $\tilde{\rho}$ are arranged as $\tilde\lambda_1\geq \cdots \geq \tilde\lambda_d $.
From Theorem \ref{theorem1}, an important problem is how to achieve $ \mathbb{E}(\hat{\lambda}_i)=O(1/N) $ for  $ r+1\leq i\leq d $. This can be achieved, as pointed out in \cite{PhysRevLett.111.183601} for a one-qubit system, by aligning the state with the measurement bases. Namely, by measuring a qubit state with its eigenbases, we can accurately estimate the small eigenvalues. Motivated from this principle, we design the two-step adaptive approach for QST  in  multi-qubit or qudit systems in the main text. The total procedure is as follows:
\begin{equation*}
	\rho\quad\xrightarrow[\text{MLE/LRE}]{\text{Step-1:}} \quad \tilde{ \rho}\quad\xrightarrow[\text{ adaptive measurement}]{\text{Step-2:}}\quad \hat{\rho}.
\end{equation*}

Using this two-step adaptive QST algorithm,  we propose the following theorem to characterize the scalings of the MSE and the estimated zero eigenvalues.
\begin{theorem}\label{qstt1}
	For the proposed two-step adaptive QST algorithm,
	the MSE scales as
	\begin{equation*}
		\begin{aligned}
			\mathbb{E}\|\hat{\rho}-{\rho}\|^2 =&O\left(\frac{1}{N_0}\right)+O\left(\frac{1}{N-N_0}\right)+O\left(\frac{1}{\sqrt{N_0\left(N-N_0\right)}}\right)+O\left(\frac{1}{N_0^{3 / 4}\left(N-N_0\right)^{1 / 4}}\right),
		\end{aligned}
	\end{equation*}
	and the estimated zero eigenvalues of $ \hat{\rho} $ scale as     \begin{equation*}
		\mathbb{E}(\hat{\lambda}_i) =O\left(\frac{1}{N_0}\right)+O\left(\frac{1}{\sqrt{N_0(N-N_0)}}\right) \text{ for } r+1\leq i\leq d.
	\end{equation*}
\end{theorem}
\begin{proof}
	We first prove that $\mathbb{E}(\hat{\lambda}_i)=O(1/N_0)$ for $ r+1\leq i \leq d $.
	Let the spectral decomposition of the true state be $ \rho=\sum_{m=1}^r \lambda_m\left|\lambda_m\right\rangle\left\langle\lambda_m\right| $.
	After Step-1, we have $\mathbb{E}(\|\tilde{\rho}-\rho\|^2)=O\left(1/{N_0}\right)$.
	Let $ \tilde{p}_{i}=\left\langle\tilde{\lambda}_i|\rho| \tilde{\lambda}_i\right\rangle $ be the measurement result when we apply $\left|\tilde{\lambda}_{i}\right\rangle\left\langle\tilde{\lambda}_{i}\right|$ as a measurement operator.
	
	We say a target matrix to be estimated is \textit{degenerate} if it equips one eigenvalue with more than one eigenvectors. For $ 1\leq i \leq r $,
	if the $i$-th eigenvalue of $\rho$ is not degenerate, using Proposition \ref{theoremyu} and similar to the proof of Eq.~\eqref{orthlam}, we have
	\begin{equation}\label{qst11}
		\mathbb{E}\left(\left|\left\langle\tilde{\lambda}_i \mid \lambda_m\right\rangle\right|^2\right) =\mathbb{E}\left(\langle\tilde{\lambda}_i| \lambda_m\rangle\langle\lambda_m| \tilde{\lambda}_i\rangle\right)=O\left(\frac{1}{N_0}\right)
	\end{equation}
	for $ i\neq m $. For $ i=m $, we have
	\begin{equation}\label{qst12}
		\begin{aligned}			\mathbb{E}\left(\left|\left\langle\tilde{\lambda}_m \mid \lambda_m\right\rangle\right|^2 \right)&=\mathbb{E}\left(\langle\tilde{\lambda}_m|\Big(I_{d}-\sum_{j\neq m}| \lambda_j\rangle\langle\lambda_j|\Big)| \tilde{\lambda}_m\rangle\right)\\
			&=1+O\left(\frac{1}{N_0}\right).
		\end{aligned}
	\end{equation}
	Therefore, if the $i$-th $ (1\leq i \leq r) $ eigenvalue of $\rho$ is not degenerate, using Eqs.~\eqref{qst11} and \eqref{qst12}, we have
	\begin{equation}\label{rr1}
		\begin{aligned}			\mathbb{E}\left(\tilde{p}_i\right)&= \mathbb{E}\left(\left\langle\tilde{\lambda}_i|\rho| \tilde{\lambda}_i\right\rangle\right)=\mathbb{E}\left(\sum_{m=1}^r \lambda_m\left|\left\langle\tilde{\lambda}_i \mid \lambda_m\right\rangle\right|^2\right) \\
			&=  \mathbb{E}\left(\lambda_i\left|\left\langle\tilde{\lambda}_i \mid \lambda_i\right\rangle\right|^2\right)+ \mathbb{E}\left(\sum_{m\neq i} \lambda_m\left|\left\langle\tilde{\lambda}_i \mid \lambda_m\right\rangle\right|^2\right)\\
			&= \lambda_i\left(1+O\left(\frac{1}{N_0}\right)\right)+O\left(\frac{1}{N_0}\right)\\
			&=\lambda_i+O\left(\frac{1}{N_0}\right).
		\end{aligned}
	\end{equation}
	Otherwise, if the $i$-th $(1\leq i \leq r) $ eigenvalue of $\rho$ is degenerate, assuming that the degenerate interval is $ [f,h] $ and $ i\in [f,h] $, i.e., $ \lambda_f=\cdots=\lambda_i=\cdots=\lambda_h $.
	Then using Proposition \ref{theoremyu} and similarly to the proof of Eq.~\eqref{orthlam}, we have
	\begin{equation}\label{qst21}	\mathbb{E}\left(\langle\tilde{\lambda}_i|\Big(\sum_{m\notin [f,h]}| \lambda_m\rangle\langle\lambda_m|\Big) |\tilde{\lambda}_i\rangle\right)=O\left(\frac{1}{N_0}\right), 
	\end{equation}
	and thus
	\begin{equation}\label{qst22}
		\begin{aligned}			\mathbb{E}\left(\langle\tilde{\lambda}_i|\Big(\sum_{m=f}^{h}| \lambda_m\rangle\langle\lambda_m|\Big)| \tilde{\lambda}_i\rangle\right)&=\mathbb{E}\left(\langle\tilde{\lambda}_i|\Big(I_{d}-\sum_{m\notin [f,h]}| \lambda_m\rangle\langle\lambda_m|\Big)| \tilde{\lambda}_i\rangle\right)\\
			&=1+O\left(\frac{1}{N_0}\right).
		\end{aligned}
	\end{equation}
	Therefore, using Eqs.~\eqref{qst21} and \eqref{qst22}, we have
	\begin{equation}\label{pst1}
		\begin{aligned}
			\mathbb{E}\left(\tilde{p}_i\right)&= \mathbb{E}\left(\left\langle\tilde{\lambda}_i|\rho| \tilde{\lambda}_i\right\rangle\right)=\mathbb{E}\left(\sum_{m=1}^r \lambda_m\left|\left\langle\tilde{\lambda}_i \mid \lambda_m\right\rangle\right|^2 \right)\\
			&= \mathbb{E} \left(\sum_{m=f}^{h}\lambda_m\left|\left\langle\tilde{\lambda}_i \mid \lambda_m\right\rangle\right|^2\right)+ \mathbb{E}\left(\sum_{m\notin [f,h]} \lambda_m\left|\left\langle\tilde{\lambda}_i \mid \lambda_m\right\rangle\right|^2\right)\\
			&=\lambda_i\left(1+O\left(\frac{1}{N_0}\right)\right)+O\left(\frac{1}{N_0}\right)\\
			&=\lambda_i+O\left(\frac{1}{N_0}\right).
		\end{aligned}
	\end{equation}
	Using Eqs.~\eqref{rr1} and \eqref{pst1}, for $1\leq i \leq r $, the variance of $ \hat{\lambda}_{i} $ is
	\begin{equation}\label{varp}
		\begin{aligned}		\operatorname{var}\left(\hat{\lambda}_{i}\right)=\frac{1}{N-N_0}\left(\mathbb{E}(\tilde{p}_i-\tilde{p}_i^2)\right)=O\left(\frac{1}{N-N_0}\right).
		\end{aligned}
	\end{equation}
	Therefore, we have
	\begin{equation}
		\mathbb{E}\left((\hat{\lambda}_i-  \tilde{p}_i)^2\right)=O\left(\frac{1}{N-N_0}\right),\quad \mathbb{E}(|\hat{\lambda}_i-  \tilde{p}_i|)\leq \sqrt{\mathbb{E}\left((\hat{\lambda}_i-  \tilde{p}_i)^2\right)}=O\left(\frac{1}{\sqrt{N-N_0}}\right),
	\end{equation}
	and thus
	\begin{equation}\label{rr2}
		\mathbb{E}(\hat{\lambda}_i)\leq\mathbb{E}\left(\tilde{p}_i\right)+O\left(\frac{1}{\sqrt{N-N_0}}\right)=\lambda_i+O\left(\frac{1}{N_0}\right)+O\left(\frac{1}{\sqrt{N-N_0}}\right).
	\end{equation}
	For $ r+1\leq i \leq d $, using Proposition \ref{theoremyu} and Eq.~\eqref{orthlam}, we have
	\begin{equation}
		\begin{aligned}
			\mathbb{E}\left(\tilde{p}_i\right)&=\mathbb{E}\left(\left\langle\tilde{\lambda}_i|\rho| \tilde{\lambda}_i\right\rangle\right)\\		&=\mathbb{E}\left(\left\langle\tilde{\lambda}_i\left|\left(\sum_{m=1}^r \lambda_m\left|\lambda_m\right\rangle\left\langle\lambda_m\right|\right)\right| \tilde{\lambda}_i\right\rangle\right)\\
			&=O\left(\frac{1}{N_0}\right),
		\end{aligned}
	\end{equation}
	and similarly,
	\begin{equation}\label{var2}	\operatorname{var}\left(\hat{\lambda}_{i}\right)=\frac{1}{N-N_0}\left(\mathbb{E}(\tilde{p}_i-\tilde{p}_i^2)\right)=O\left(\frac{1}{N_0\left(N-N_0\right)}\right).
	\end{equation}
	Therefore, we have
	\begin{equation}
		\mathbb{E}\left((\hat{\lambda}_i-  \tilde{p}_i)^2\right)=O\left(\frac{1}{N_0(N-N_0)}\right),\quad \mathbb{E}(|\hat{\lambda}_i-  \tilde{p}_i|)\leq \sqrt{\mathbb{E}\left((\hat{\lambda}_i-  \tilde{p}_i)^2\right)}=O\left(\frac{1}{\sqrt{N_0(N-N_0)}}\right),
	\end{equation}
	and thus
	\begin{equation}\label{lamp}
		\mathbb{E}(\hat{\lambda}_i)\leq \mathbb{E}\left(\tilde{p}_i\right)+O\left(\frac{1}{\sqrt{N_0(N-N_0)}}\right)=O\left(\frac{1}{N_0}\right)+O\left(\frac{1}{\sqrt{N_0(N-N_0)}}\right).
	\end{equation}
	

	Now, we consider the scaling of the  MSE. Since $\mathbb{E}\|\tilde{\rho}-{\rho}\|^2=O(1/N_{0})$ in Step-1, using Lemma \ref{lemma3}, we have
	\begin{equation}\label{e1}
		\mathbb{E}\left((\tilde{\lambda}_i-\lambda_i)^2\right)=O\left(\frac{1}{N_0}\right).
	\end{equation}
	In Step-2, 
	for $ 1\leq i \leq r $,  using Eqs.~\eqref{varp} and \eqref{rr2}, we have
	\begin{equation}\label{la}
		\begin{aligned}
			\mathbb{E}\left((\hat{\lambda}_i-\lambda_i)\right)^2&=\left(\mathbb{E}(\hat{\lambda}_i-\lambda_i)\right)^2+\operatorname{var}(\hat{\lambda}_i-\lambda_i)\\
			&=\left[O\left(\frac{1}{N_0}\right)+O\left(\frac{1}{\sqrt{N-N_0}}\right)\right]^2+O\left(\frac{1}{N-N_0}\right) \\
			&=O\left(\frac{1}{N-N_0}\right)+O\left(\frac{1}{N_0^2}\right)+O\left(\frac{1}{N_0\sqrt{N-N_0}}\right).
		\end{aligned}
	\end{equation}
	Using the Cauchy–Schwarz inequality and Eqs.~\eqref{e1} and \eqref{la}, we have
	\begin{equation}\label{r3}
		\begin{aligned}
			\left|\mathbb{E}\left(\left(\hat{\lambda}_i-\lambda_i\right)\left(\tilde{\lambda}_i-\lambda_i\right)\right)\right| &\leq \sqrt{\mathbb{E}\left((\hat{\lambda}_{i}-\lambda_i)^2 \right)\mathbb{E}\left((\tilde{\lambda}_i-\lambda_i)^2\right)}\\
			&=\sqrt{\left[O\left(\frac{1}{N-N_0}\right)+O\left(\frac{1}{N_0^2}\right)+O\left(\frac{1}{N_0\sqrt{N-N_0}}\right)\right]O\left(\frac{1}{N_0}\right)}\\
			&\leq O\left(\frac{1}{\sqrt{N_0\left(N-N_0\right)}}\right)+O\left(\frac{1}{N_0^{3/2}}\right)+O\left(\frac{1}{N_0(N-N_0)^{1/4}}\right),
		\end{aligned}
	\end{equation}
	where we use $ \sqrt{a+b+c}\leq \sqrt{a}+\sqrt{b}+\sqrt{c}$ for $a, b, c\geq0  $ in the second inequality.
	Therefore, for $ 1\leq i \leq r $, using Eqs.~\eqref{e1}--\eqref{r3},  we have
	\begin{equation}\label{mse1}
		\begin{aligned}
			\mathbb{E}\left((\hat{\lambda}_i-\tilde{\lambda}_i)^2\right)&=\mathbb{E}\left(\left(\hat{\lambda}_i-\lambda_i\right)-\left(\tilde{\lambda}_i-\lambda_i\right)\right)^2 \\
			&=\mathbb{E}\left((\hat{\lambda}_i-\lambda_i)^2\right)+\mathbb{E}\left((\tilde{\lambda}_i-\lambda_i)^2\right)-2 \mathbb{E}\left(\left(\hat{\lambda}_i-\lambda_i\right)\left(\tilde{\lambda}_i-\lambda_i\right)\right) \\
			&\leq\mathbb{E}\left((\hat{\lambda}_i-\lambda_i)^2\right)+\mathbb{E}\left((\tilde{\lambda}_i-\lambda_i)^2\right)+2\sqrt{\mathbb{E}\left((\hat{\lambda}_{i}-\lambda_i)^2 \right)\mathbb{E}\left((\tilde{\lambda}_i-\lambda_i)^2\right)}\\
			&\leq O\left(\frac{1}{N-N_0}\right)+O\left(\frac{1}{N_0^2}\right)+O\left(\frac{1}{N_0\sqrt{N-N_0}}\right)+O\left(\frac{1}{N_0}\right)\\
			&\quad+O\left(\frac{1}{\sqrt{N_0\left(N-N_0\right)}}\right)
			+O\left(\frac{1}{N_0^{3/2}}\right)+O\left(\frac{1}{N_0(N-N_0)^{1/4}}\right)\\
			&\leq O\left(\frac{1}{N-N_0}\right)+O\left(\frac{1}{N_0}\right)  +O\left(\frac{1}{\sqrt{N_0\left(N-N_0\right)}}\right).
		\end{aligned}
	\end{equation}	
	For $ r+1\leq i \leq d $, $ \lambda_{i}=0 $ and  using Eqs.~\eqref{var2} and \eqref{lamp}, we have
	\begin{equation}\label{lamp2}
		\begin{aligned}
			\mathbb{E} (\hat{\lambda}_{i}^2)&=\left(\mathbb{E} \hat{\lambda}_{i}\right)^2+\text{var}\left(\hat{\lambda}_{i}\right)=O\left(\frac{1}{N_0^{2}}\right)+O\left(\frac{1}{N_0\sqrt{N_0\left(N-N_0\right)}}\right)+O\left(\frac{1}{N_0\left(N-N_0\right)}\right).
		\end{aligned}
	\end{equation}
	Using Eqs.~\eqref{e1} and \eqref{lamp2}, further, we have
	\begin{equation}\label{rr3}
		\begin{aligned}			\left|\mathbb{E}\left(\hat{\lambda}_{i} \tilde{\lambda}_i\right)\right| &\leq \sqrt{\mathbb{E}(\hat{\lambda}_{i}^2)\mathbb{E}(\tilde{\lambda}_i^{2}})=\sqrt{\mathbb{E}(\hat{\lambda}_{i}^2) \mathbb{E}\left((\tilde{\lambda}_{i}-{\lambda}_{i})^2\right)}\\
			&=\sqrt{\left[O\left(\frac{1}{N_0^{2}}\right)+O\left(\frac{1}{N_0\sqrt{N_0\left(N-N_0\right)}}\right)+O\left(\frac{1}{N_0\left(N-N_0\right)}\right)\right] O\left(\frac{1}{N_0}\right)}\\
			&\leq O\left(\frac{1}{N_0^{3/2}}\right)+O\left(\frac{1}{N_0^{5/4}(N-N_0)^{1/4}}\right)+O\left(\frac{1}{N_0\sqrt{N-N_0}}\right),
		\end{aligned}
	\end{equation}
	where we use $ \sqrt{a+b+c}\leq \sqrt{a}+\sqrt{b}+\sqrt{c}$ for $a, b, c\geq0  $ in the second inequality.
	For $ r+1\leq i \leq d $, using Eqs.~\eqref{e1}, \eqref{lamp2} and \eqref{rr3}, we have
	\begin{equation}\label{mse2}
		\begin{aligned}		\mathbb{E}\left((\hat{\lambda}_i-\tilde{\lambda}_i)^2\right)&=\mathbb{E}(\hat{\lambda}_i^2)+\mathbb{E}(\tilde{\lambda}_i^2)-2 \mathbb{E}\left(\hat{\lambda}_i \tilde{\lambda}_i\right) \\
			&\leq\mathbb{E}(\hat{\lambda}_i^2)+\mathbb{E}(\tilde{\lambda}_i^2)+2\sqrt{\mathbb{E}(\hat{\lambda}_{i}^2) \mathbb{E}(\tilde{\lambda}_i^{2})} \\
			&\leq O\left(\frac{1}{N_0^2}\right)+O\left(\frac{1}{N_0\sqrt{N_0\left(N-N_0\right)}}\right)+O\left(\frac{1}{N_0\left(N-N_0\right)}\right)+O\left(\frac{1}{N_0}\right)\\
			&\quad+O\left(\frac{1}{N_0^{3/2}}\right)+O\left(\frac{1}{N_0^{5/4}(N-N_0)^{1/4}}\right)+O\left(\frac{1}{N_0\sqrt{N-N_0}}\right) \\
			&=O\left(\frac{1}{N_0}\right).
		\end{aligned}
	\end{equation}
	Therefore, using Eqs.~\eqref{mse1} and \eqref{mse2}, the MSE $ \mathbb{E}\|\hat{\rho}-\tilde{\rho}\|^2 $ scales as
	\begin{equation}\label{msefi}
		\begin{aligned}
			\mathbb{E}(\|\hat{\rho}-\tilde{\rho}\|^2)&= \mathbb{E}\sum_{i=1}^r\left(\hat{\lambda}_i-\tilde{\lambda}_i\right)^2+ \mathbb{E}\sum_{i=r+1}^d\left(\hat{\lambda}_i-\tilde{\lambda}_i\right)^2\\
			&=O\left(\frac{1}{N_0}\right)+ O\left(\frac{1}{N-N_0}\right)+O\left(\frac{1}{\sqrt{N_0\left(N-N_0\right)}}\right).
		\end{aligned}
	\end{equation}
	Since $\mathbb{E}(\|\tilde{\rho}-\rho\|^2)=O\left(1/{N_0}\right)$, using the Cauchy–Schwarz inequality, we have
	\begin{equation}\label{eh}
		\begin{aligned}
			\left|\mathbb{E}\left(\|\hat{\rho}-\tilde{\rho}\|\|\tilde{\rho}-\rho\|\right)\right|&\leq\sqrt{\mathbb{E}(\|\hat{\rho}-\tilde{\rho}\|^2)\mathbb{E}(\|\tilde{\rho}-\rho\|^2)}\\
			&=\sqrt{\left[O\left(\frac{1}{N_0}\right)+ O\left(\frac{1}{N-N_0}\right)+O\left(\frac{1}{\sqrt{N_0\left(N-N_0\right)}}\right)\right]O\left(\frac{1}{N_0}\right)}\\
			&\leq O\left(\frac{1}{N_0}\right)+O\left(\frac{1}{\sqrt{N_0\left(N-N_0\right)}}\right)+O\left(\frac{1}{N_0^{3 / 4}\left(N-N_0\right)^{1 / 4}}\right).
		\end{aligned}
	\end{equation}
	Then since
	\begin{equation}
		\|\hat{\rho}-\rho\|\leq \|\hat{\rho}-\tilde{\rho}\|+\|\tilde{\rho}-\rho\|,
	\end{equation}	
	using Eqs.~\eqref{msefi} and \eqref{eh}, the MSE $\mathbb{E}\|\hat{\rho}-\rho\|^2$ scales as 	
	\begin{equation}
		\begin{aligned}
			\mathbb{E}(\|\hat{\rho}-\rho\|^2)&=\mathbb{E}(\|\hat{\rho}-\tilde{\rho}\|^2)+\mathbb{E}(\|\tilde{\rho}-\rho\|^2)+2\mathbb{E}\left(\|\hat{\rho}-\tilde{\rho}\|\|\tilde{\rho}-\rho\|\right) \\
			&=O\left(\frac{1}{N_0}\right)+O\left(\frac{1}{N-N_0}\right)+O\left(\frac{1}{\sqrt{N_0\left(N-N_0\right)}}\right)+O\left(\frac{1}{N_0^{3 / 4}\left(N-N_0\right)^{1 / 4}}\right).
		\end{aligned}
	\end{equation}	
\end{proof}

Using Theorem \ref{qstt1}, we have the following corollary.
\begin{corollary}\label{cc2}
	If we use the two-step adaptive QST algorithm with $ N_0=\alpha N $, where $ 0<\alpha<1$ is a constant, 
	the infidelity $ \mathbb{E}\left(1-F\left(\hat{\rho}, \rho\right)\right)$ can achieve the optimal scaling $O\left(1/N\right)  $.
\end{corollary}

The proof is straightforward. Using Theorem \ref{qstt1},
if we choose $ N_0=\alpha N $, where $ \alpha $ is a constant and $ 0<\alpha<1$, we have the fact that the estimated zero eigenvalues of $ \hat{\rho} $ scale as $ \mathbb{E}(\hat{\lambda}_i)=O\left({1}/{N}\right) $ 
for $r+1\leq i \leq d $ and the MSE also scales as $ \mathbb{E}\|\hat{\rho}-{\rho}\|^2= O\left(1/N\right)$. Therefore,  the conditions C1 and C2 in Theorem \ref{theorem1} are both satisfied and thus the infidelity $ \mathbb{E}\left(1-F\left(\hat{\rho}, \rho\right)\right)$ has the optimal scaling $O\left(1/N\right)  $.

The computational complexity in Step-2 is $O(d^3)$. Therefore,
if we use LRE \cite{Qi2013} in Step-1,  the total  computational complexity  is still $ O(Ld^2) $  where 
$ L\geq d^2 $ is the type number of different measurement operators for QST in Step-1.

\subsection{Two-step  adaptive quantum detector tomography}\label{sec6d}

For adaptive QDT, we first propose the following lemma.
\begin{lemma}\label{prof}
	Let $ \tilde{P}=\tilde{P}(N) $ be a positive semidefinite estimate of $ P $ depending on resource number $ N $.
	Let the spectral decomposition be $\tilde{P}=\tilde{U}\operatorname{diag}\left( \tilde{\lambda}_{1}, \cdots, \tilde{\lambda}_{d}\right) \tilde{U}^{\dagger}  $, where $ \tilde\lambda_1\geq\cdots \geq \tilde\lambda_d\geq 0$. Given an integer $ r\geq 0 $,
	if $ \mathbb{E}\left(\sum_{j=r+1}^{d}\tilde{\lambda}_j\right)=O\left(1/N\right) $, then for any bounded matrix $ \mathcal{S}=\mathcal{S}(N) \in \mathbb{C}^{d\times d}, \|\mathcal{S}\|\leq c $, where $ c $ is a constant, we have $ \mathbb{E}\left(\sum_{j=r+1}^{d}\hat{\lambda}_j\right)=  O\left(1/N\right)  $, where $ \hat{\lambda}_1 \geq \cdots \geq\hat{\lambda}_d \geq 0 $ are the eigenvalues of $ \hat{P}=\mathcal{S}\tilde{P}\mathcal{S}^{\dagger} $.
\end{lemma}
\begin{proof}
	If $r=d$, it is clear that this lemma holds. Thus, we focus on $r<d$.
	Since $\tilde{P}\geq 0$, $\hat{P}= \mathcal{S}\tilde{P}\mathcal{S}^{\dagger}\geq0 $.
	Assume that the  singular value decomposition (SVD) of $ \mathcal{S} $ is
	\begin{equation*}
		\begin{aligned}
			\mathcal{S}&={U}_{\mathcal{S}} \operatorname{diag}\left(s_1, \cdots, s_d\right){V}_{\mathcal{S}} ^{\dagger},
		\end{aligned}
	\end{equation*}
	where $ {U}_{\mathcal{S}} $ and $ {V}_{\mathcal{S}} $ are two $ d\times d $ unitary matrices.
	Define
	\begin{equation*}
		\begin{aligned}
			\mathcal{W}&\triangleq\operatorname{diag}\left(s_1, \cdots, s_d\right) {V}_{\mathcal{S}} ^{\dagger} \tilde{U}.
		\end{aligned}
	\end{equation*}
	Assuming that the QR decomposition of $ \mathcal{W} $ is $ \mathcal{W}=\mathcal{Q}\mathcal{M} $, where $ \mathcal{Q} $ is a unitary matrix and $ \mathcal{M} $ is an upper triangular matrix.
	Since $ \|\mathcal{S}\|\leq c $, we have $ s_i=O(1) $ for $ 1\leq i \leq d $ and thus $ \|\mathcal{W}\|=O(1) $, $ \|\mathcal{M}\|=O(1) $.	
	Therefore,
	\begin{equation}
		\begin{aligned}
			{U}_{\mathcal{S}}^{\dagger} \hat{P} {U}_{\mathcal{S}} &={U}_{\mathcal{S}}^{\dagger} {\mathcal{S}} \tilde{P} {\mathcal{S}}^{\dagger} {U}_{\mathcal{S}}  \\
			&=\operatorname{diag}\left(s_1, \cdots, s_d\right) {V}^{\dagger}_{\mathcal{S}} \tilde{U} \operatorname{diag}\left(\tilde{\lambda}_{1}, \cdots, \tilde{\lambda}_{d}\right) \tilde{U}^{\dagger} {V}_{\mathcal{S}} \operatorname{diag}\left(s_1, \cdots, s_d\right) \\
			&=\mathcal{Q}\mathcal{M} \operatorname{diag}\left(\tilde{\lambda}_{1}, \cdots, \tilde{\lambda}_{d}\right) \mathcal{M}^{\dagger} \mathcal{Q}^{\dagger}.
		\end{aligned}
	\end{equation}
	Let $ \mathcal{O}=\mathcal{M} \operatorname{diag}\left(\tilde{\lambda}_{1}, \cdots, \tilde{\lambda}_{d}\right) \mathcal{M}^{\dagger} $.
	Since $ \mathbb{E}\left(\sum_{j=r+1}^{d}\tilde{\lambda}_j\right)=O\left(1/N\right) $, we have  $ \mathbb{E}(\tilde{\lambda}_j)=O\left(1/N\right) $ for $ r+1\leq j\leq d $.  Thus,
	\begin{equation}
		\begin{aligned}
			&\mathbb{E}(\mathcal{O})=\mathbb{E}\left(\mathcal{M}\operatorname{diag}\left(\tilde{\lambda}_{1}, \cdots, \tilde{\lambda}_{d}\right) \mathcal{M}^{\dagger}\right)=\left[\begin{array}{cc}
				\mathcal{C}_{r \times r}^{(0)}+\mathcal{C}_{r \times r}^{(1)}O(1 / N)  & \mathcal{C}_{r \times(d-r)}^{(2)} O(1 / N) \\
				\mathcal{C}_{(d-r) \times r}^{(3)} O(1 / N) & \mathcal{C}_{(d-r) \times(d-r)}^{(4)} O(1 / N),
			\end{array}\right],
		\end{aligned}
	\end{equation}
	where $ \mathcal{C}^{(i)} $ are bounded matrices. Assuming that a rearranging of the diagonal elements of $\mathcal{O}   $ in decreasing order is $ o_1\geq \cdots \geq o_d $. Therefore, the expectation values of the $ d-r $ diagonal elements $ o_{r+1}, \cdots, o_{d} $ scale as $O(1 / N)$. Since the eigenvalues of $ \mathcal{O} $ are the same as the eigenvalues of $\hat{P}  $, from \cite{CHAN1983562}, we have
	\begin{equation}
		0 \leq \sum_{j=k}^d \hat{\lambda}_j \leq \sum_{j=k}^d o_j,
	\end{equation}
	for $1\leq k \leq d  $.  Therefore, $ \mathbb{E}\left(\sum_{j=r+1}^{d}\hat{\lambda}_j\right)=  O\left(1/N\right)  $.
\end{proof}

Now we consider adaptive QDT using a similar two-step adaptive algorithm.
In Step-1, we apply MLE \cite{PhysRevA.64.024102} or the two-stage algorithm \cite{wang2019twostage} with $ N_0 $ copies of probe states and obtain $ \{\bar P_i\}_{i=1}^{n} $ where $ \bar P_i\geq0 $ and $ \sum_{i=1}^{n}\bar P_i=I $. Assume that the spectral decomposition of the $ i $-th POVM element $ \bar P_i $ is 
\begin{equation}\label{qq1}
	\bar P_i=\bar U_{i} \operatorname{diag}(\bar\lambda_1^{i},\cdots,\bar\lambda_d^{i}) \bar U_{i}^{\dagger}.
\end{equation}
Then in Step-2, for each $i,j$, we apply adaptive tomography using $ \frac{N-N_0}{nd} $  copies of new probe states  $\tilde\rho_{j}^{i}=\left|\bar{\lambda}_{j}^{i}\right\rangle\left\langle\bar{\lambda}_{j}^{i}\right|$. Let $\tilde p_{j}^{\,i,k}$ denote the measurement outcome probability when the input state $\tilde\rho_{j}^{\,i}$ is measured with the $k$-th POVM element, with $\sum_{k=1}^{n}\tilde p_{j}^{\,i,k} = 1$.
Let $\tilde{\lambda}_{j}^{\,i}= \tilde p_{j}^{\,i,i}$ and Step-2 estimate of \(P_i\) is then
\begin{equation}\label{qq2}
	\tilde{P}_{i} =\bar{U}_{i} \operatorname{diag}\left(\tilde{\lambda}_{1}^{i}, \cdots, \tilde{\lambda}_{d}^{i}\right) \bar{U}_{i}^{\dagger}.
\end{equation}
However, the sum of $ \tilde{P}_{i} $ may not be equal to the identity.
Let 
\begin{equation}\label{tildei}
	\sum_{i=1}^n \tilde{P}_i=\tilde{I}_{d},  
\end{equation}
and we can assume that  $\tilde{I}_{d} >0  $ because $ \tilde{I}_{d} $  converges to $ I_{d} $.
Then, we  obtain the final estimate $\hat{P}_i $ as
\begin{equation}\label{correct}
	\begin{aligned}
		\hat{P}_i=\tilde{I}_{d}^{-1 / 2} \tilde{P}_i \tilde{I}_{d}^{-1 / 2},
	\end{aligned}
\end{equation}
where $ \hat{P}_i\geq 0 $ and $ \sum_{i=1}^{n}\hat P_i=I_{d} $.
The total procedure is
\begin{equation}\label{qdtpre}
	{P}_i\quad\xrightarrow[\text{MLE/Two-stage algorithm}]{\text{Step-1:}} \quad\bar{P}_i\quad\xrightarrow[\text{ adaptive probe states}]{\text{Step-2:}} \quad \tilde{P}_{i} \quad \xrightarrow[]{\text{Correction Eq.~\eqref{correct}}} \quad\hat{P}_i.
\end{equation}

Then using the above two-step adaptive QDT algorithm, we present the following theorem to characterize the scaling of the infidelity. 
\begin{theorem}\label{qdtt1}
	If we use the  two-step adaptive QDT algorithm and  $ N_0=\alpha N $, where $ 0<\alpha<1$ is a constant,
	the infidelity $\mathbb{E}\left(1-F\left(\hat{P}_{i}, P_{i}\right)\right)\left(\text{or } \mathbb{E}\left(1- F_{d,p}\left(\hat{P}_{i}, P_{i}\right)\right)\right) $ scales as $ O(1/N) $ for each POVM element.
\end{theorem}
\begin{proof}
	If we choose $ N_0=\alpha N $, where $ 0<\alpha<1$ is a constant, similar to the proof of two-step adaptive QST in Theorem \ref{qstt1}, we can prove  $\mathbb{E} (\tilde{\lambda}_j^{i})=O\left(1/N\right) $ for $ r+1\leq j\leq d $ and $\mathbb{E}\left\|\tilde{P}_i-P_i\right\|^2=O\left(1/N\right)$.
	Then  we consider
	\begin{equation}
		\|\tilde{I}_{d} -I_{d}\|=\left\|\sum_{i=1}^n\left(\tilde{P}_i-{P}_i\right)\right\| \leq \sum_{i=1}^{n}\left\|\tilde{P}_i-{P}_i\right\|.
	\end{equation}
	Thus, $ \mathbb{E}\|\tilde{I}_{d} -I_{d}\|^2= O\left(1/N\right) $.
	Assuming that the spectral decomposition of $ \tilde{I}_{d}  $ is
	\begin{equation}\label{93}
		\tilde{I}_{d} =U_{\tilde{I}_{d} } \operatorname{diag}\left(1+t_1, \cdots, 1+t_d\right) U_{\tilde{I}_{d}}^{\dagger},
	\end{equation}
	and $ \|\tilde{I}_{d} -I_{d}\|^2=\sum_{i=1}^d t_i^2 $, then we have
	\begin{equation}
		\left|\frac{1}{\sqrt{1+t_i}}-1\right| \sim \frac{t_i}{2}+o\left(t_i\right),
	\end{equation}
	and thus
	\begin{equation}\label{f12}
		\left\|\tilde{I}_{d} ^{-1 / 2}-I_{d}\right\|^2=\sum_{i=1}^{d}\left(\frac{1}{\sqrt{1+t_i}}-1\right)^2\sim\sum_{i=1}^d \frac{t_i^2}{4} = \frac{1}{4}\|\tilde{I}_{d} -I_{d}\|^2.
	\end{equation}
	Therefore, $ \mathbb{E}\left\|\tilde{I}_{d}^{-1 / 2}-I_{d}\right\|^2 =O\left(1/N\right)$ and  $ \left\|\tilde{I}_{d}^{-1 / 2}\right\|^2 \sim d$. Since $ \left\|\tilde{P}_i\right\|^2 \leq d $, we have
	\begin{equation}
		\begin{aligned}
			\quad\left\|\hat{P}_i-\tilde{P}_i\right\|&=\left\|\tilde{I}_{d}^{-1 / 2} \tilde{P}_i \tilde{I}_{d}^{-1 / 2}-\tilde{P}_i\right\| \\
			&=\left\|\tilde{I}_{d}^{-1 / 2} \tilde{P}_i \tilde{I}_{d}^{-1 / 2}-\tilde{I}_{d}^{-1 / 2} \tilde{P}_i+\tilde{I}_{d}^{-1 / 2} \tilde{P}_i-\tilde{P}_i\right\| \\
			&\leq \left\|\tilde{I}_{d}^{-1 / 2} \tilde{P}_i \tilde{I}_{d}^{-1 / 2}-\tilde{I}_{d}^{-1 / 2} \tilde{P}_i\right\|+\left\|\tilde{I}_{d}^{-1 / 2} \tilde{P}_i-\tilde{P}_i\right\|\\
			&\leq \left\|\tilde{I}_{d}^{-1 / 2}\right\|\left\|\tilde{P}_i\right\|\left\|\tilde{I}_{d}^{-1 / 2}-I_{d}\right\|+\left\|\tilde{P}_i\right\|\left\|\tilde{I}_{d}^{-1 / 2}-I_{d}\right\| \\
			&=O\left(\left(d+\sqrt{d}\right)\left\|\tilde{I}_{d}^{-1 / 2}-I_{d}\right\|\right),
		\end{aligned}
	\end{equation}
	and thus $\mathbb{E}\left(\left\|\hat{P}_i-\tilde P_i\right\|^2\right)= O\left(1/N\right)$. Since
	\begin{equation}\label{qdtfi}
		\left\|\hat{P}_i-P_i\right\| \leq\left\|\hat{P}_i-\tilde{P}_i\right\|+\left\|\tilde{P}_i-\bar{P}_i\right\|+\left\|\bar{P}_i-P_i\right\|,
	\end{equation}
	the MSE scales as $ \mathbb{E}\left(\left\|\hat{P}_i-P_{i}\right\|\right)^2 =O\left(1/{N}\right)$  and thus the condition C1 in Theorem \ref{theorem1} is satisfied.
	Since $ \hat{P}_{i}=\tilde{I}_{d}^{-1 / 2}\tilde {P}_{i} \tilde{I}_{d}^{-1 / 2} $, where $ \tilde{I}_{d}^{-1 / 2} $ is a bounded matrix and $ \mathbb{E}\left(\sum_{j=r+1}^{d}\tilde{\lambda}_j^{i}\right)=O\left(1/N\right) $, using Lemma \ref{prof}, we have $ \mathbb{E}\left(\sum_{j=r+1}^{d}\hat{\lambda}_j^{i}\right)=O\left(1/N\right) $.
	Therefore, condition C2 is also satisfied
	and thus using Theorem \ref{theorem1}, the infidelity for QDT $ \mathbb{E}\left(1-F\left(\hat{P}_{i}, P_{i}\right)\right) $$\left(\text{or } \mathbb{E}\left( 1- F_{d,p}\left(\hat{P}_{i}, P_{i}\right)\right)\right) $   has the optimal scaling $O\left(1/N\right)  $. 
\end{proof}

The computational complexities in Step-2 and  in correction using Eq.~\eqref{correct} are both $O(nd^3)$. Therefore,
if we use two-stage estimation \cite{wang2019twostage} in Step-1,  the total  computational complexity  is still $ O(nMd^2) $,  where $ M\geq d^2 $ is the type number of different probe states.

\subsection{Three-step  adaptive ancilla-assisted quantum process tomography}\label{sec63}
The proposed adaptive AAPT method in the main text comprises three steps and is applicable to both trace-preserving and non-trace-preserving quantum processes.
We first input a pure state  $ \sigma^{\text{in}}=|\Phi\rangle\langle\Phi| $ with the full Schmidt number, i.e., $\operatorname{Sch}(\sigma^{\text{in}})=d^2  $.  Let the Schmit decomposition of $ |\Phi\rangle $ be $|\Phi\rangle=\sum_{i=1}^dh_{i}|\phi_{i}^{(1)}\rangle \otimes|\phi_{i}^{(2)}\rangle$, where $ h_i>0 $ for $ 1\leq i \leq d $, and  $ \left\{|\phi_{i}^{(1)}\rangle\right\}_{i=1}^{d} $ and $ \left\{|\phi_{i}^{(2)}\rangle\right\}_{i=1}^{d} $ are two orthonormal bases. We assume that $|\phi_{i}^{(1)}\rangle=U|i\rangle   $ and $|\phi_{i}^{(2)}\rangle=V|i\rangle   $, where $ U $ and $ V $ are two $ d\times d $ unitary matrices.
Thus, the input state can be represented as
\begin{equation}\label{aaqtin}
	\begin{aligned}
		\sigma^{\text{in}}&=|\Phi\rangle\langle\Phi |=\sum_{i, j=1}^{d, d} h_i h_j(U|i\rangle \otimes V|i\rangle)\left(\langle j| U^{\dagger} \otimes\langle j| V^{\dagger}\right) \\
		&=\sum_{i, j=1}^{d, d} h_i h_j U|i\rangle\langle j|U^{\dagger} \otimes V| i\rangle\langle j| V^{\dagger}.
	\end{aligned}
\end{equation}
For trace-preserving processes, we apply the two-step adaptive QST on the output state and obtain $\hat{\sigma}^{\text{out}}$ satisfying
$ \operatorname{Tr}\left(\hat{\sigma}^{\text{out}}\right)=1 $ and $ \hat{\sigma}^{\text{out}}\geq 0 $. Using Corollary  \ref{cc2}, we have
\begin{equation}\label{qptqst1}
	\mathbb{E}\left(1- F\left(\hat{\sigma}^{\text{out}}, \sigma^{\text{out}}\right)\right)=O(1/N).
\end{equation} 
For non-trace-preserving processes,   we also apply the same procedure as the two-step adaptive QST. In Step-1, we use LRE Eq.~\eqref{lse} without the correction algorithm in \cite{effqst} and 
obtain $ \tilde{\sigma}^{\text{out}}=\sum_{i=1}^{d^2}\tilde{\lambda}_{i}\left|\tilde{\lambda}_{i}\right\rangle\left\langle\tilde{\lambda}_{i}\right| $ with $ N_0=\alpha N $ copies where $ 0<\alpha<1$ is a constant. Here $ \tilde{\sigma}^{\text{out}} $ may be non-physical. Then in Step-2, we  use the eigenbasis $\left\{\left|\tilde{\lambda}_{i}\right\rangle\left\langle\tilde{\lambda}_{i}\right|\right\}_{i=1}^{d^2}$ as the new measurement operators where the resource number is $ N-N_0 $ and obtain the corresponding  new measurement frequency data $ \{\hat p_i\}_{i=1}^{d^2} $. Let the eigenvalues be $ \hat{\lambda}_{i}= \hat p_i$ and the final estimate is
\begin{equation}
	\hat{\sigma}^{\text{out}}=\sum_{i=1}^{d^2}\hat{\lambda}_{i}\left|\tilde{\lambda}_{i}\right\rangle\left\langle\tilde{\lambda}_{i}\right|
\end{equation}
satisfying $ \hat{\sigma}^{\text{out}} \geq 0$ because $ \hat{\lambda}_{i}= \hat p_i\geq0 $. As opposed to two-step adaptive QST, here the unit trace does not hold for $ \hat{\sigma}^{\text{out}} $.  But we still have $ \operatorname{Tr}\left({\hat\sigma}^{\text{out}}\right)<1 $ because $ \sum_{i=1}^{d^2}\hat{\lambda}_{i}=\sum_{i=1}^{d^2}\hat{p}_{i}<1 $ for non-trace-preserving processes.
We call the above procedures two-step adaptive quantum pseudo-state tomography (QPST) where the whole procedure is
\begin{equation}\label{qqpst}
	{\sigma}^{\text{out}}\quad \xrightarrow[]{\text{Step-1: LRE}} \quad \tilde{\sigma}^{\text{out}} \quad \xrightarrow[\text{adaptive measurements}]{\text{Step-2: }} \quad \hat{\sigma}^{\text{out}}\; (\operatorname{Tr}\left({\hat\sigma}^{\text{out}}\right)< 1).
\end{equation}
Similar to  Theorem \ref{qstt1} and Corollary  \ref{cc2}, we can also prove
\begin{equation}\label{qpst1}
	\mathbb{E}\left(1- F\left(\hat{\sigma}^{\text{out}}, \sigma^{\text{out}}\right)\right)=O(1/N).
\end{equation}

After the above two-step adaptive QST/QPST, we obtain an estimate $ \hat\sigma^{\text{out}}$.
Let $ H\triangleq\operatorname{diag}\left(h_1,\cdots, h_d\right) $ and thus $ H|i\rangle=h_i|i\rangle $. Using Eq.~\eqref{sp1}, we have
\begin{equation}\label{abcd}
	\begin{aligned}
		\sum_{i, j=1}^{d, d} h_i h_j A B|i\rangle\langle j|B^{\dagger} A^{\dagger} \otimes| i\rangle\langle j|=&\sum_{i, j=1}^{d, d} A(B H)| i\rangle\langle j|(B H)^{\dagger} A^{\dagger} \otimes| i\rangle\langle j| \\
		=&\sum_{i, j=1}^{d, d} A|i\rangle\langle j|A^{\dagger} \otimes(B H)^T| i\rangle\langle j|(B H)^*\\
		=&\sum_{i, j=1}^{d, d} A| i\rangle\langle j|A^{\dagger} \otimes H B^T| i\rangle\langle j| B^* H.
	\end{aligned}
\end{equation}
Let $ \widetilde{\mathcal{E}} $ be the estimate of the quantum process $\mathcal{E}$ using $\hat\sigma^{\text {out }}$. Using Eq.~\eqref{abcd} and letting $ B=U $, we have 
\begin{equation}\label{aapt1}
	\begin{aligned}
		\sum_{i, j=1}^{d, d}  h_i h_j\tilde{\mathcal{E}}\left(U|i\rangle\langle j| U^{\dagger}\right) \otimes|i\rangle\langle j|=&\sum_{i, j=1}^{d, d} \tilde{\mathcal{E}}(|i\rangle\langle j|) \otimes H U^T| i\rangle\langle j| U^* H^*\\
		=&\left(I_{d} \otimes HU^T\right)\left(\sum_{i, j=1}^{d, d}\tilde{\mathcal{E}}(|i\rangle\langle j|) \otimes|i\rangle\langle j|\right)\left(I_{d} \otimes U^*H\right).
	\end{aligned}
\end{equation}
Therefore, using  Eqs.~\eqref{aaqtin} and \eqref{aapt1},  $ \hat\sigma^{\text {out }} \geq 0 $ can be represented as
\begin{equation}
	\begin{aligned}
		\hat{\sigma}^{\text{out}}=&(\tilde{\mathcal{E}} \otimes {I})\left(\sigma^{\text{in}} \right)=(\tilde{\mathcal{E}} \otimes {I})\left(|\Phi\rangle\langle\Phi | \right)\\
		=&(I \otimes V)\sum_{i,j=1}^{d,d}  h_{i}h_{j} \left( \tilde{\mathcal{E}}\left(U| i\rangle\langle j|U^{\dagger}\right) \otimes   | i\rangle\langle j| \right)(I \otimes V)^{\dagger}\\
		=&\left(I \otimes VHU^T\right)\sum_{i, j=1}^{d, d}(\widetilde{\mathcal{E}}(|i\rangle\langle j|) \otimes|i\rangle\langle j|)\left(I \otimes U^*HV^{\dagger}\right).
	\end{aligned}
\end{equation}
Define $ \tilde{X}_{0}\triangleq\sum_{i,j=1}^{d,d}\widetilde{\mathcal{E}}(|i\rangle\langle j|) \otimes|i\rangle\langle j| $ and
thus we can calculate it  as (note that $ H>0 $)
\begin{equation}\label{ghzprue}
	\tilde{X}_{0}=\left(I \otimes U^{*}H^{-1}V^{\dagger}\right) \hat{\sigma}^{\text {out }} \left(I \otimes VH^{-1}U^T\right).
\end{equation}
Since $ \hat{\sigma}^{\text {out }}\geq 0$, we have $ \tilde{X}_{0}\geq  0$.
Crucially, if $\widetilde{\mathcal{E}}\rightarrow\mathcal{E}$, from Eqs.~\eqref{choi} and \eqref{ghzx} we know $\tilde{X}_{0}\rightarrow X$. Therefore, $\tilde{X}_{0}$ is potentially a candidate estimate of the process matrix $X$. However, $\tilde{X}_{0}$ may not satisfy the partial trace requirement $ \operatorname{Tr}_{1}(\hat{X})=I_d $ or $ \operatorname{Tr}_{1}(\hat{X})\leq I_d $, and we thus need another step to correct it as follows.

In Step-3, our goal is to correct the partial trace of $\tilde{X}_0$. A feasible approach is to employ the second stage of the two-stage solution proposed in \cite{xiaoqpt}. For trace-preserving processes, we apply Eq.~\eqref{tss21} and obtain $ \hat X $ satisfying
$ \operatorname{Tr}_{1}(\hat{X})=I_d $, and for non-trace-preserving processes,  we apply Eq.~\eqref{tss22} and obtain $ \hat X $ satisfying $ \operatorname{Tr}_{1}(\hat{X})\leq I_d $.
The total procedure of three-step adaptive AAPT is 
\begin{equation}\label{aaptpro}
	{\sigma}^{\text {out }}\quad  \xrightarrow[\text{adaptive QST/QPST}]{\text{Steps 1-2:}} \quad \hat{\sigma}^{\text{out}} \quad \xrightarrow[\text{Eq}.~\eqref{ghzprue}]{\text{Step-3:}} \quad  \tilde{X}_{0} \quad  \xrightarrow[\text{Correction on partial trace}]{\text{Step-3:}} \quad \hat{X}.
\end{equation}

Here, we present the following theorem to characterize the performance of the above algorithm for both trace-preserving and non-trace-preserving processes with a pure input state in AAPT.
\begin{theorem}\label{theorem2}
	Let the input state $ \sigma^{\text{in}} $ of AAPT be a pure state with $\operatorname{Sch}(\sigma^{\text{in}})=d^2  $.
	If we use the proposed three-step adaptive AAPT algorithm and take $ N_0=\alpha N $, where $ 0<\alpha<1$ is a constant,
	the infidelity $\mathbb{E}(1-F(\hat{X}, X))$ or $\mathbb{E}(1- F_{d,p}(\hat{X}, X)) $ scales as $ O(1/N) $.
\end{theorem}
\begin{proof}
	For a trace-preserving process,
	after the two-step adaptive QST, we obtain $ \hat{\sigma}^{\text{out}}$, where $ \operatorname{Tr}\left(\hat{\sigma}^{\text{out}}\right)=1 $ and reconstruct $ \tilde{X}_{0} $  as Eq.~\eqref{ghzprue}. From Theorem \ref{qstt1}, we know $\mathbb{E}\left(\left\|\hat{\sigma}^{\text {out }}-{\sigma}^{\text {out }}\right\|^2\right)=O\left(1/{N}\right)$. Using Eq.~\eqref{ghzprue}, we have 
	\begin{equation}
		\mathbb{E}\left(\left\|\tilde{X}_{0}-X\right\|^2\right)=O\left(\frac{1}{N}\right).
	\end{equation}
	Then using Eq.~\eqref{tss21} in Step-3 for trace-preserving processes, from Eq.~\eqref{stage22}, we have $\mathbb{E}\left(\left\|\hat{X}-\tilde{X}_{0}\right\|^2\right)=O\left({1}/{N}\right)  $.
	Therefore, the ultimate MSE scales as $\mathbb{E}\left(\left\|\hat{X}-X\right\|^2\right)=O\left({1}/{N}\right)  $,  which satisfies the condition C1 in Theorem \ref{theorem1}.
	
	Since we apply the two-step adaptive QST for trace-preserving processes, using Corollary \ref{cc2},  we have $\mathbb{E}(1- F(\hat{\sigma}^{\text{out}},$ $\sigma^{\text{out}}))=O(1/N)$. Let $r = \operatorname{Rank}(\sigma^{\text{out}}) = \operatorname{Rank}(X)$ (by Eq.~\eqref{ghzprue}). Then from Theorem \ref{theorem1}, the eigenvalues of $ \hat{\sigma}^{\text{out}}  $ (arranged in descending order) scale as $ \mathbb{E}(\hat{\lambda}_i)=O\left(1/{N}\right) $ for $ r+1 \leq i \leq d^2 $.
	Thus
	using Eq.~\eqref{ghzprue} and Lemma \ref{prof}, the $i$-th eigenvalue (in descending order) of the reconstructed process matrix $ \tilde{X}_{0} $ also scales as $ O\left(1/{N}\right) $ for $ r+1 \leq i \leq d^2 $.
	Then using Eq.~\eqref{tss21} in Step-3 and Lemma \ref{prof}, the $i$-th eigenvalue (in descending order) of the reconstructed process matrix $ \hat X$ also scales as $ O\left(1/{N}\right) $ for $ r+1 \leq i \leq d^2 $, which satisfies the condition C2 in Theorem \ref{theorem1}.
	Therefore, using Theorem \ref{theorem1}, we have  $\mathbb{E}\left( 1- F\left(\hat{X}, X\right)\right) = O(1/N) $ for trace-preserving processes.
	
	For a non-trace-preserving process,
	after the two-step adaptive QPST, we obtain $ \hat{\sigma}^{\text{out}}$, where $ \operatorname{Tr}\left(\hat{\sigma}^{\text{out}}\right)<1 $ and we also reconstruct $ \tilde{X}_{0} $  as Eq.~\eqref{ghzprue}. Then using Eq.~\eqref{tss22} in Step-3, we obtain the final estimate $ \hat X $.
	Similar to trace-preserving processes, using Eqs.~\eqref{ghzprue} and \eqref{tss22}, we also have $\mathbb{E}\left(\left\|\hat{X}-X\right\|^2\right)=O\left({1}/{N}\right)  $ satisfying the condition C1 in Theorem \ref{theorem1}.
	Since we use the two-step  adaptive QPST for non-trace-preserving processes, from Eq.~\eqref{qpst1}, we have $\mathbb{E}\left(1- F\left(\hat{\sigma}^{\text{out}}, \sigma^{\text{out}}\right)\right)=O(1/N)$. Then using Theorem~\ref{theorem1} and similar to the case of trace-preserving processes, the $i$-th eigenvalue (in descending order) of the reconstructed process matrix $ \hat X$ also scales as $O\left(1/{N}\right) $ for $ r+1 \leq i \leq d^2 $, which satisfies the condition C2 in Theorem~\ref{theorem1}.
	Therefore, using Theorem~\ref{theorem1}, we have  $\mathbb{E}\left( 1- F\left(\hat{X}, X\right)\right)\left(\text{or } \mathbb{E}\left(1- F_{d,p}\left(\hat{X}, X\right)\right)\right) = O(1/N) $ for non-trace-preserving processes. 
\end{proof}

The computational complexity in Step-2 is $O(d^4)$ in correction of Eq.~\eqref{ghzprue} is $O(d^6)$ and in Step-3 is $O(d^6)$. Therefore,
if we use LRE \cite{Qi2013} in Step-1,  the total  computational complexity  is still $ O(Ld^4) $  where 
$ L\geq d^4 $ is the type number of different measurement operators for QST.

\begin{remark}
	To satisfy $ \operatorname{Tr}_{1}( \hat {X})= I_{d} $ or $ \operatorname{Tr}_{1}( \hat {X})\leq I_{d} $, in addition to Eqs.~\eqref{tss21} and \eqref{tss22} of the Stage-2 algorithm in the two-stage solution \cite{xiaoqpt}, one can also use other algorithms. If these algorithms can keep  $ O\left(1/{N}\right) $  scaling of the estimated zero eigenvalues, the final infidelity also has  optimal scaling.
\end{remark}

\section{Additional numerical results}\label{sec7}
In the numerical and experimental examples, Cube measurements are typically employed.
For one-qubit systems, the Cube measurements are $\frac{I\pm\sigma_x}{2}, \frac{I\pm\sigma_y}{2}, \frac{I\pm\sigma_z}{2}$.
For two-qubit systems, the Cube measurements are the tensor products of one-qubit Cube measurements. Therefore,   there are nine detectors, with each detector containing four POVM elements for the two-qubit Cube measurements.

\begin{figure}
	\centering
	\includegraphics[width=4.4in]{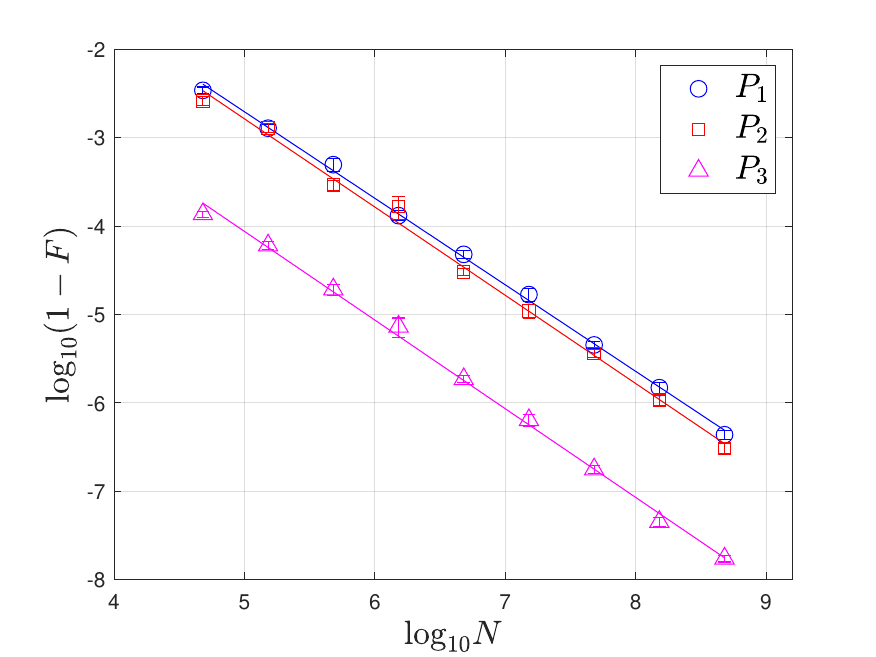}
	\centering{\caption{Log-log plot of the infidelity $ \mathbb{E}\left(1-F\left(\hat{P}_{i}, P_{i}\right)\right) $ versus  the total resource number $N$  for the three-valued detectors in Eq.~\eqref{three1} using the two-step adaptive QDT algorithm in Section \ref{sec6d}. The infidelities for all the POVM elements scale as $ O(1/N) $ satisfying Theorem \ref{qdtt1}}}\label{adqdtinf}
\end{figure}

For  QDT, we consider a three-valued detector as in \cite{xiao2021optimal}:
\begin{equation}
	\label{three1}
	\begin{aligned}
		&P_{1}+P_{2}+P_{3}=I,\\
		&P_{1}=U_{1} \operatorname{diag}\left(0.4, 0,0,0\right) U_{1}^{\dagger}=0.4 U_{1}(|00\rangle\langle 00|) U_{1}^{\dagger},\\
		&P_{2}=U_{2} \operatorname{diag}\left(0, 0.5,0,0\right) U_{2}^{\dagger}=0.5 U_{2}(|01\rangle\langle 01|) U_{2}^{\dagger},
	\end{aligned}
\end{equation}
where $ U_1 $ and $ U_2 $ are randomly generated unitary matrices \cite{qetlab,Zyczkowski_1994,qutip} that are subsequently fixed throughout the simulation.
We implement the two-step adaptive QDT algorithm in Section \ref{sec6d} on the detector in Eq.~\eqref{three1} where the resource number in Step-1 and Step-2 are both $ N/2 $. In Step-1, the probe states are $24 $ random pure states \cite{qetlab,qutip,MISZCZAK2012118} and in Step-2, we apply $ 12 $ adaptive probe states $\tilde\rho_{j}^{i}=\left|\bar{\lambda}_{j}^{i}\right\rangle\left\langle\bar{\lambda}_{j}^{i}\right|$.
The results are shown in Fig. \ref{adqdtinf}, where the infidelities for all the POVM elements scale as $ O(1/N) $ satisfying Theorem \ref{qdtt1}.

Then we apply the three-step adaptive AAPT in Section \ref{sec63} for a non-trace-preserving phase damping process characterized by two Kraus operators
\begin{equation}
	\mathcal{A}_1=\operatorname{diag}(1,\sqrt{1/3}),\quad \mathcal{A}_2=\operatorname{diag}(0,\sqrt{1/3}).
\end{equation}
The input state is a random pure state and then fixed in the simulation, where the Schmidt number is four. For adaptive AAPT, we apply the two-step adaptive QPST in Section \ref{sec63}
with resource number $ N_0=N/2 $, and obtain $ \hat\sigma^{\text{out}}$ and then reconstruct $ \tilde{X}_{0} $ using Eq.~\eqref{ghzprue}.
Then in Step-3, we apply Eq.~\eqref{tss22} on $ \tilde{X}_{0} $ to obtain a physical estimate $ \hat X $.
As a comparison, we also simulate the results of a non-adaptive AAPT algorithm where we assume that we know a prior the value of $\operatorname{Tr}\left(\sigma^{\text{out}}\right)$ such that the non-adaptive version of the algorithm can be straightforwardly designed and realized.
Note that in the adaptive AAPT, we do not have this prior knowledge.
In non-adaptive AAPT,  we apply non-adaptive QST via LRE Eq.~\eqref{lse}  with the resource number $ N_0=N $ and obtain $ \tilde{\sigma}^{\text{out}}=\sum_{i=1}^{d^2}\tilde{\lambda}_{i}\left|\tilde{\lambda}_{i}\right\rangle\left\langle\tilde{\lambda}_{i}\right| $ where $\tilde{\lambda}_{1}\geq \cdots \geq \tilde{\lambda}_{k} \geq 0 > \tilde{\lambda}_{k+1} \geq \cdots \geq \tilde{\lambda}_{d^2} (k\leq d^2)  $. To satisfy the positive semidefinite constraint and the trace value $\operatorname{Tr}\left(\sigma^{\text{out}}\right)$, the final estimate is given by
\begin{equation}
	\hat{\sigma}^{\text {out }}=\frac{\operatorname{Tr}\left({\sigma}^{\text{out}}\right)}{\sum_{i=1}^k \tilde{\lambda}_i}\sum_{i=1}^k \tilde{\lambda}_i
	\left|\tilde{\lambda}_i\right\rangle\left\langle\tilde{\lambda}_i\right|.
\end{equation}
Then we obtain a physical estimate $\hat X$ using Step-3 of the adaptive AAPT.

\begin{figure}
	\centering
	\includegraphics[width=4.4in]{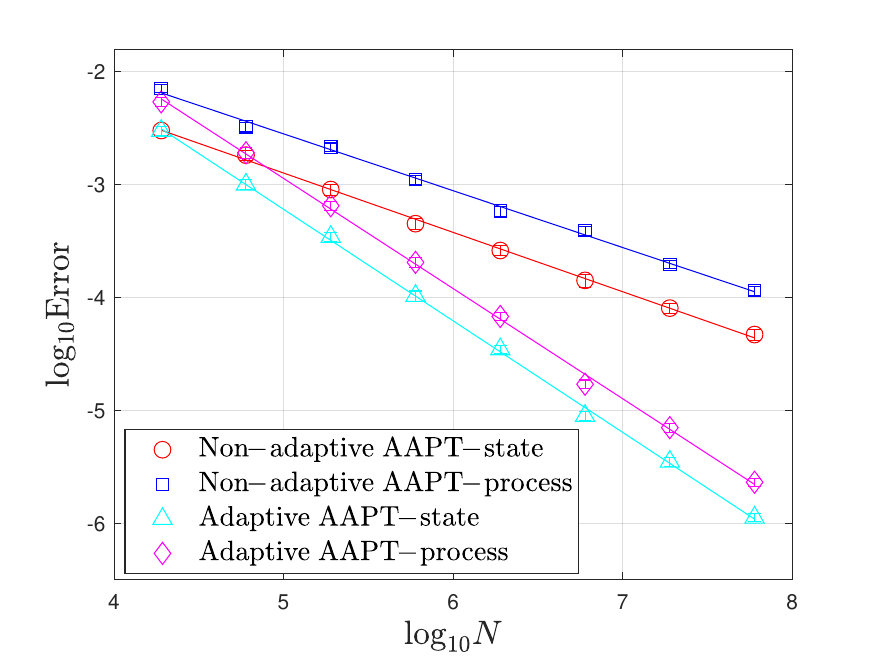}
	\centering{\caption{Log-log plot of the infidelities $\mathbb{E}\left(1- F\left(\hat{\sigma}^{\text{out}}, \sigma^{\text{out}}\right)\right) $ for the reconstructed output state $ \hat{\sigma}^{\text{out}} $ and $\mathbb{E}(1-F(\hat X,X )) $ for the reconstructed process $ \hat{X} $ versus  the total resource number $N$ for the phase damping process in Eq.~\eqref{phase2} using the three-step adaptive AAPT algorithm in Section \ref{sec63} and its non-adaptive version. The infidelities of adaptive methods scale as $O(1/N)$ satisfying Theorem \ref{theorem1} and Theorem \ref{theorem2}, while non-adaptive methods scale as $O(1/\sqrt{N})$.}\label{adaaptinfa}}
\end{figure}

The results are shown in Fig. \ref{adaaptinfa}, where we plot the infidelities $\mathbb{E}\left(1- F\left(\hat{\sigma}^{\text{out}}, \sigma^{\text{out}}\right)\right) $ for the reconstructed output state $ \hat{\sigma}^{\text{out}} $ and $\mathbb{E}(1-F(\hat X,X )) $ for the reconstructed process $ \hat{X} $.
From Fig. \ref{adaaptinfa}, using non-adaptive AAPT, even with a prior knowledge of $ \operatorname{Tr}\left({\sigma}^{\text{out}}\right) $, the infidelities of the reconstructed output state $\mathbb{E}\left(1- F\left(\hat{\sigma}^{\text{out}}, \sigma^{\text{out}}\right)\right) $ and of the reconstructed process matrix  $\mathbb{E}(1- F(\hat{X}, X)) $ both scale only  as $ O(1/\sqrt{N}) $ because the non-adaptive method cannot ensure to achieve $ O(1/{N}) $ scaling of the estimated zero eigenvalues. However, using the three-step adaptive AAPT in Section \ref{sec63}, the infidelities $\mathbb{E}\left(1- F\left(\hat{\sigma}^{\text{out}}, \sigma^{\text{out}}\right)\right) $ and $\mathbb{E}(1- F(\hat{X}, X)) $  both scale as $ O(1/{N}) $.

\section{Additional ancilla-assisted quantum process tomography experimental setup and results}

\subsection{Bell state preparation}

In the state preparation stage, a Ti-sapphire laser first outputs a light pulse centered at a wavelength of $780$ nm, with a repetition rate of approximately $76$ MHz and a pulse duration of about $150$ fs. After passing through a frequency doubler, the pulse frequency is doubled, converting the infrared light into ultraviolet light. Subsequently, the enhanced ultraviolet pulse is focused onto a BBO crystal that is specifically cut to achieve a type-II phase-matched spontaneous parametric down-conversion (SPDC). High-intensity ultraviolet pulses within the crystal trigger nonlinear optical effects, thereby splitting a high-energy ultraviolet photon into a pair of lower-energy infrared photons.

One of the photons serves as a trigger and is detected by a single-photon counter, while the other photon is used as a single-photon source and input at the beginning of the optical path. This photon is prepared in the $\lvert H\rangle$ state  after passing through a polarizing beam splitter (PBS). It then becomes a $\frac{1}{\sqrt{2}}(\lvert H \rangle + \lvert V\rangle)$ state after passing through a half-wave plate oriented at $22.5^{\circ}$. To transform the polarization state into a path state, a beam displacer (BD) is used to direct the $H$ component into path $1$ and the $V$ component into path $0$. A half-wave plate oriented at $0^{\circ}$ is placed in path $1$, and a half-wave plate oriented at $90^{\circ}$ is placed in path $0$. This successfully prepares the photon in the $\frac{1}{\sqrt{2}}(\lvert H,1\rangle + \lvert V,0\rangle)$ state, which represents the maximally entangled state of the polarization and path qubits.

\subsection{Measurement setup}

In this section, we provide more detailed information about the adaptive AAPT experiment. The experimental optical setup for the measurement part is shown in Fig. \ref{measurement}. The optical setup allows for arbitrary two-qubit projective measurements, with the four POVM elements corresponding to the four exits E$_1$-E$_4$ in Fig. \ref{measurement}(a). It can be proven that the combination of waveplates in the form of Q-H-Q can implement any two-dimensional unitary transformation, while the combination Q-H can convert the $|H\rangle$ state into any state. The six variable waveplate combinations in the figure implement six unitary transformations, which are the six so-called coin operators $C(x,t)$. Each passage through BD corresponds to the application of the operator $T = \lvert x,H\rangle\langle x,H\lvert \ + \lvert x-1,V\rangle\langle x,V\lvert$, where $x$ represents the path. Let the upper path in the diagram be $x_1$ and the lower path be $x_0$. BD$_1$ is a special BD composed of two BDs bonded together, as illustrated in Fig. \ref{measurement}(b). Its overall effect is equivalent to applying the operator $T_1 = \lvert x_0,-V\rangle\langle x_1,H\lvert + \lvert x_1,V\rangle\langle x_1,V\lvert + \lvert x_0,-H\rangle\langle x_0,V\lvert + \lvert x_0,H\rangle\langle x_1,H\lvert$. The outputs of BD$_i$ divide Fig. \ref{measurement}(a) into four stages, and the evolution in each stage is determined by a unitary transformation of the form $U(t) = TC(t)$, where $C(t) = \sum_x \lvert x\rangle\langle x\lvert \otimes C(x,t)$ is determined by site-dependent coin operators $C(x,t)$. After $k$ stages, the unitary operator generated by the coin operators and translation operator reads $U = TC(k)\cdots TC(2)TC(1)$.  

\begin{figure}
	\centering 
	\includegraphics[width=6.25in]{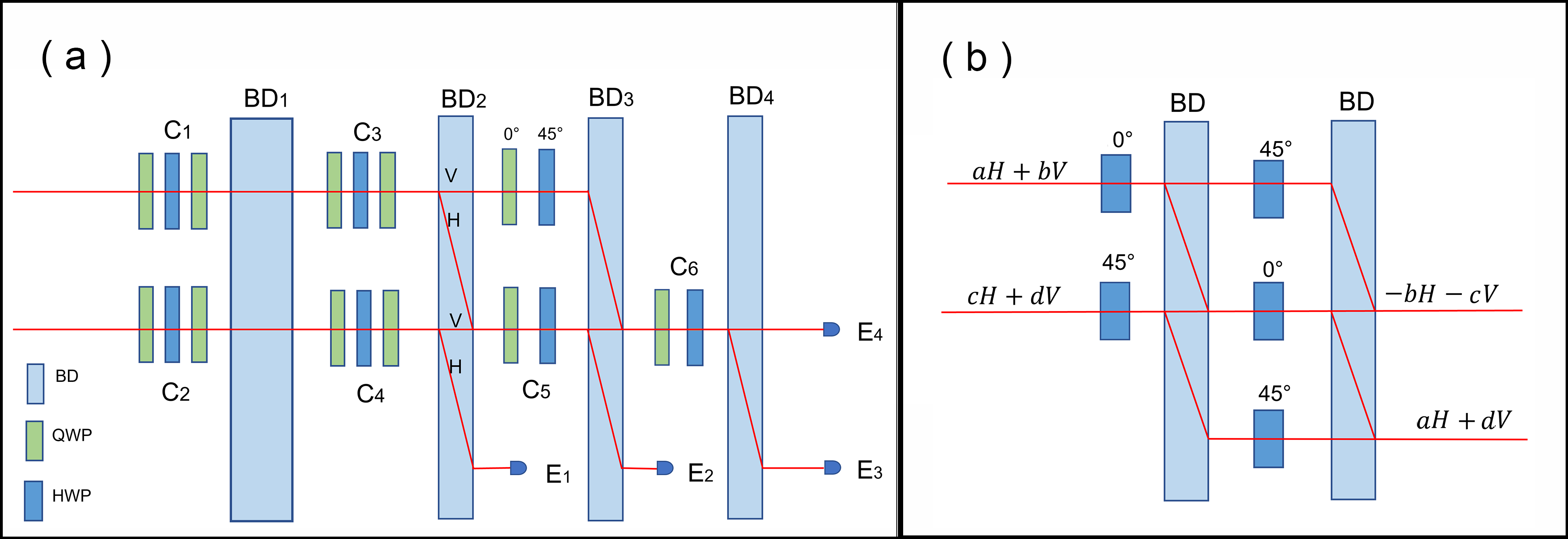}
	\centering{\caption{(a): Realization of arbitrary projective measurements using photonic quantum walks. The translation operator is realized by beam displacers (BDs). The nontrivial coin operators are realized by half wave plates (HWPs) and quarter wave plates (QWPs) with rotation angles specified in Table \ref{tab}. (b): The specific construction of BD$_1$ involves two BDs, with small HWPs fixed at specific angles attached at the front and in the middle.}\label{measurement} }  
\end{figure}

\begin{table}
	\caption{The waveplate angles for the nine detectors, with C$_1$-C$_4$ using the Q-H-Q waveplate combination, while C$_5$ and C$_6$ using the Q-H waveplate combination.}
	\label{tab}
	
	\small %
	\centering
	
	\resizebox{\textwidth}{!}{
		\begin{tabular}{ccccccc} 
			\toprule
			detector& \textbf{C$_1$} & \textbf{C$_2$} & \textbf{C$_3$} & \textbf{C$_4$} & \textbf{C$_5$} & \textbf{C$_6$} \\
			\midrule
			\begin{tabular}{c}
				\\ D1 \\ D2 \\ D3 \\ D4 \\ D5 \\ D6 \\ D7 \\ D8 \\ D9  
			\end{tabular} 
			& \begin{tabular}{ccc}
				QWP & HWP & QWP \\67.5 & 135 & 22.5\\0   & 67.5 & 0\\0   &  0   & 0\\157.5& 22.5 & 157.5\\90  & 22.5 & 0\\0   & 90 & 90\\0  & 0 & 0\\0  & 0 & 0\\0  & 0 & 0\\
			\end{tabular} 
			& \begin{tabular}{ccc}
				QWP & HWP & QWP \\ 22.5 & 135 & 67.5\\ 0 & 112.5 & 90\\  0   & 135 & 90\\  22.5  & 135 & 67.5\\  0  & 112.5 & 90\\  0  & 135 & 90\\  22.5 & 135 & 67.5\\ 0 & 112.5 & 90\\ 0 & 135 & 90\\
			\end{tabular}
			& \begin{tabular}{ccc}
				QWP & HWP & QWP \\ 22.5 & 135 & 67.5\\112.5 & 157.5 & 112.5\\112.5 & 135 & 157.5\\157.5 & 135 & 112.5\\ 67.5 & 112.5 & 67.5\\ 67.5 & 135 & 22.5\\ 0 & 135 & 90\\ 90 & 135 & 90\\ 90 & 135 & 0\\
			\end{tabular}
			& \begin{tabular}{ccc}
				QWP & HWP & QWP\\ 67.5 & 135 & 22.5\\ 67.5 & 135 & 22.5\\ 67.5 & 135 & 22.5\\ 67.5 & 135 & 22.5\\ 67.5 & 135 & 22.5\\
				67.5 & 135 & 22.5\\ 90 & 135 & 0\\  90 & 135 & 0\\ 90 & 135 & 0\\
			\end{tabular}
			& \begin{tabular}{cc}
				QWP & HWP \\ 0 & 0 \\ 0 & 0 \\ 0 & 0 \\ 0 & 0 \\  0 & 0 \\  0 & 0 \\  0 & 0 \\ 0 & 0 \\0 & 0\\
			\end{tabular}
			& \begin{tabular}{cc}
				QWP & HWP \\  0 & 45 \\ 0 & 45 \\ 0 & 45 \\ 0 & 45 \\ 0 & 45 \\ 0 & 45 \\  45 & 22.5 \\ 45 & 67.5 \\ 0 & 45 \\
				
			\end{tabular}\\
			\bottomrule
	\end{tabular}}
\end{table}

Assume that the four POVM elements of a measurement basis are $f_1$, $f_2$, $f_3$ and $f_4$. Since they are mutually orthogonal, by adjusting the angles of the waveplate combinations C$_1$ to C$_6$ as shown in Table~\ref{tab}, $U$ can be set to satisfy $Uf_1 = [0,1,0,0]^{T}$, $Uf_2 = [1,0,0,0]^{T}$, $Uf_3 = [0,0,1,0]^{T}$, $Uf_4 = [0,0,0,1]^{T}$, which can then be output from the four exits, respectively. The four exits are all connected to single-photon counters via fiber optic coupling, with a coupling efficiency of about $80\%$.  By comparing the number of photons counted in each single-photon counter, the values of the detectors  can be obtained.

\subsection{Experimental results with different resource distributions}

In Section \ref{sec7}, we have discussed how the proportion of resource allocation influences tomography precision. We experimentally validate this effect using the same phase damping process as in the main text. Specifically, we examine the case where the resource allocation in the first step in three-step adaptive AAPT is  $N_0=0.9N$, as shown in Fig.~\ref{0.9-0.1}. The total numbers of resources used are set to $N=30$, $100$, $300$, $950$, $3000$, $9490$, $30000$ and $94870$, respectively. The results indicate that with a $0.9N$--$0.1N$ allocation in two steps, compared to a $0.5N$--$0.5N$ allocation, higher precision is achievable with the same number of photons, and the overall scaling of infidelity remains consistent as $O(1/N)$.

\begin{figure}[!t]
	\centering 
	\includegraphics[width=4.4in]{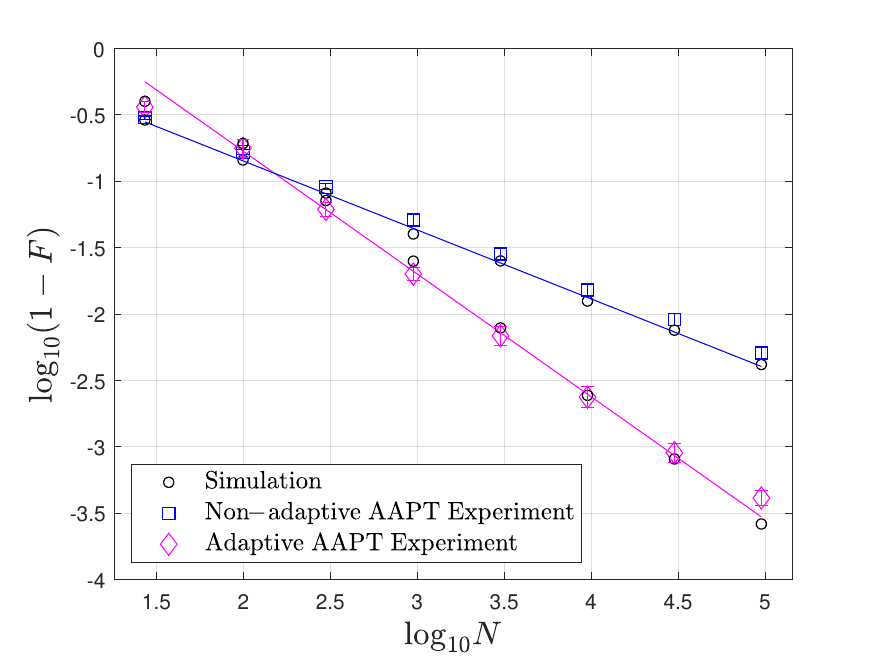}
	\centering\caption{Log-log plot of the infidelity $\mathbb{E}(1-F(\hat{X}, X))$ versus  the total resource number $N$ for the phase damping process with a resource distribution of $N_0=0.9N$, where all the other conditions remain unchanged. Square markers indicate experimental results obtained via the non-adaptive method, and diamond markers represent results using the adaptive AAPT method. Black circles represent simulation outcomes, with solid lines representing fitted simulations. Error bars show the standard deviations across 100 repeated experiments. The adaptive method achieves infidelity scaling as $O(1/N)$, whereas the non-adaptive method scales as $O(1/\sqrt{N})$.}
	\label{0.9-0.1} 
\end{figure}

\end{document}